%% file: H_epsilon_IT_Revised.tex
\documentclass[10pt, conference,onecolumn]{IEEEtran}

\include{mypreamble}

\begin{document}

\title{Estimation Efficiency Under Privacy Constraints}

\author{Shahab~Asoodeh,~Mario~Diaz,~Fady~Alajaji,~\IEEEmembership{Senior Member,~IEEE,}
        and~Tam\'{a}s~Linder,~\IEEEmembership{Fellow,~IEEE}
\thanks{This  work  was  supported  by  the  Natural  Sciences  and Engineering  Research  Council  of  Canada.  This  paper  was  presented  in  part at the IEEE International  Symposium on Information Theory 2016 and 2017 \cite{Asoode_MMSE_submitted, Asoode_ISIT17_submitted}.}%
\thanks{S. Asoodeh is with the Computation Institute, The University of Chicago, Chicago, IL 60637 USA (e-mail: shahab@uchicago.edu).}
\thanks{M. Diaz is with the School of Electrical, Computer and Energy Engineering, Arizona State University, Tempe, AZ 85287-5706 USA and the School of Engineering and Applied Sciences, Harvard University, Cambridge, MA 02138 USA (emails: mdiaztor@\{asu,g.harvard\}.edu).}%
\thanks{F. Alajaji and T. Linder are with the Department of Mathematics and Statistics, Queen's University, Kingston, ON K7L3N6 Canada (e-mails: fa@queensu.ca; tamas.linder@queensu.ca).}}
\maketitle
\begin{abstract}
We investigate the problem of estimating a random variable $Y$ under a privacy constraint dictated by another correlated random variable $X$. 
When $X$ and $Y$ are discrete, we express the underlying privacy-utility tradeoff in terms of the privacy-constrained guessing probability $\mathcalboondox{h}(P_{XY}, \eps)$, the maximum probability $\cP(Y|Z)$ of correctly guessing $Y$ given an auxiliary random variable $Z$, where the maximization is taken over all $P_{Z|Y}$ ensuring that $\cP(X|Z)\leq \eps$ for a given privacy threshold $\eps \geq 0$. We prove that $\mathcalboondox{h}(P_{XY}, \cdot)$ is concave and piecewise linear, which allows us to derive its expression in closed form for any $\eps$  when $X$ and $Y$ are binary. 
In the non-binary case, we derive $\mathcalboondox{h}(P_{XY}, \eps)$ in the high utility regime (i.e., for sufficiently large, but nontrivial, values of $\eps$) under the assumption that $Y$ and $Z$ have the same alphabets. 
We also analyze the privacy-constrained guessing probability for two scenarios in which $X$, $Y$ and $Z$ are binary vectors. 
When $X$ and $Y$ are continuous random variables, we formulate the corresponding privacy-utility tradeoff in terms of $\sM(P_{XY}, \eps)$, the smallest normalized minimum mean squared-error (mmse) incurred in estimating $Y$ from a Gaussian perturbation $Z$. Here the minimization is taken over a family of Gaussian perturbations $Z$ for which the mmse of $f(X)$ given $Z$ is within a factor $1-\eps$ from the variance of $f(X)$ for any non-constant real-valued function $f$. We derive tight upper and lower bounds for $\sM$ when $Y$ is Gaussian. For general absolutely continuous random variables, we obtain a tight lower bound for $\sM(P_{XY}, \eps)$ in the high privacy regime, i.e., for small $\eps$.
\end{abstract} 

\section{Introduction}
\IEEEPARstart{W}{e} consider the following constrained estimation problem: given two correlated random variables $X$ and $Y$, how accurately can $Y$ be estimated from another correlated random variable $Z$, while ensuring that the "information leakage" about $X$ is limited? More precisely, we seek to design a randomized mechanism $\M$ which maps $Y$ to an auxiliary random variable $Z$ such that the information leakage from $X$ to $Z$ is limited, and the "estimation efficiency" of $Y$ given $Z$ is maximal. This basic question arises often in data privacy problems, where Alice wishes to disclose {\em non-private information} $Y$ to Bob as accurately as possible in order to receive a payoff, 
but in such a way that her {\em private information} $X$ cannot be effectively inferred by Bob. For instance, her browsing history might constitute the non-private information which a social media website collects in order to provide personalized recommendations. In an ideal world, her browser should sanitize $Y$ before its release in order to avoid compromising her private information $X$ (which may for example include her political leanings). In this context, her browser has access only to $Y$, but the potential correlation between $X$ and $Y$ makes the sanitization of $Y$ critical. Motivated by this type of applications, we assume throughout the paper that $X$, $Y$, and $Z$ form a Markov chain in that order, denoted by $X\markov Y\markov Z$.

 Given the joint distribution $P_{XY}$, Alice chooses a random mapping $\M$ to generate the \emph{displayed data} $Z$ in such a way that Bob can guess $Y$ from $Z$ as accurately as possible while being unable to use $Z$ to efficiently guess $X$. Note that $\M$, the so-called \emph{privacy filter}, is completely determined by $P_{Z |Y}$. The system block diagram of  this model is depicted in Fig.~\ref{fig:XYZ}.

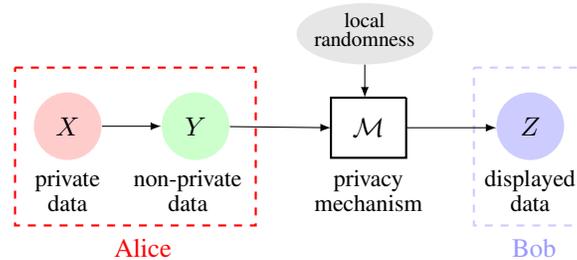
\begin{figure}[h]      
	\centering
	\begin{tikzpicture}
	\draw (-7.5,-1) node[fill=red!20,ellipse, anchor=base, minimum size=0.9cm] (ind2) {$X$};
	\draw (-5.8,-1) node[fill=green!20, ellipse, anchor=base, minimum size=0.9cm] (ind1) {$Y$};
	\node at (-5.9, -1.6) {\small{non-private}};
	\node at (-5.9, -1.9) {\small{data}};
	\node at (-7.5, -1.6) {\small{private}};
	\node at (-7.5, -1.9) {\small{data}};
	\path[arrow] (ind2) -- (ind1);
	\path [arrow] (ind1) -- (-4, -0.9);
	\draw[black,thick] (-3,-0.5) -- (-4,-0.5) --(-4,-1.3) -- (-3,-1.3)-- (-3,-0.5);
	\draw[red,dashed, thick] (-8.2,-0.1) -- (-5,-0.1) --(-5,-2.2) -- (-8.2,-2.2)-- (-8.2,-0.1);
	\draw[blue!20,dashed, thick] (-2.1,-0.1) -- (-0.6,-0.1) --(-0.6,-2.2) -- (-2.1,-2.2)-- (-2.1,-0.1);
	\node at (-3.5, -0.9) {$\mathcal{M}$};
	
	\draw [gray!20, fill=gray!20, anchor=base] (-3.55,0.35) ellipse (0.9cm and 0.4cm);
	\node at (-3.55, 0.53) {\footnotesize{local}};
	\node at (-3.55, 0.28) {\footnotesize{randomness}};
	\draw[arrow] (-3, -0.9) -- (-1.8, -0.9);
	\draw[arrow]  (-3.55, -0.06)-- (-3.55, -0.5);
	\draw (-1.35, -1)  node[fill=blue!20, ellipse, anchor=base, minimum size=0.9cm] (generate) {$Z$};
	\node at (-1.35, -1.6) {\small{displayed}};
	\node at (-1.35, -1.9) {\small{data}};
	\node at (-6.5, -2.5) {\textcolor{red}{Alice}};
	\node at (-1.3, -2.5) {\textcolor{blue!40}{Bob}};
	\node at (-3.5, -1.6) {\small{privacy}};
	\node at (-3.5, -1.9) {\small{mechanism}};
	\end{tikzpicture}
	\caption{The system block diagram.}
	\label{fig:XYZ}
\end{figure}

 A quantitative answer to this problem requires: (i) an appropriate measure $\L(X\to Z)$ of information leakage from $X$ to $Z$; and (ii) an appropriate measure $\S(Y|Z)$ of the estimation efficiency of $Y$ given $Z$. A quantitative and operationally well-justified measure of information leakage has been long sought to assess the performance of different mechanisms used in practice.  In this paper, we set $\mathcal{S}(Y|Z)=\L(Y\to Z)$ and propose two measures of information leakage depending on the support of $X$ and $Y$. 
\begin{description}
\item[Discrete case:]  When $X\in \X$ and $Y\in \Y$ are both discrete, it is natural to define information leakage as Bob's efficiency in guessing $X$. Hence, we propose $\L(X\to Z)$ to be $\frac{\cP(X|Z)}{\cP(X)}$, where $\cP(X) \coloneqq \max_{x\in\X} P_X(x)$ is the \emph{probability of correctly guessing} $X$ and
\begin{equation}
    \begin{aligned}
   	\cP(X|Z)&\coloneqq\sum_{z\in \Z}P_{Z}(z)\max_{x\in \X} P_{X|Z}(x|z)\\
    &=\sum_{z\in \Z}\max_{x\in \X}P_{X}(x)P_{Z|X}(z|x), \label{Def_PC}
    \end{aligned}
	\end{equation}
is the \emph{probability of correctly guessing} $X$ given $Z$. Note that a large value of $\L(X\to Z)$ corresponds to a small probability of error in guessing $X$ upon observing $Z$. Although we only assume that $\Z$, the alphabet of $Z$, has finite cardinality, we will show that any $\Z$ with cardinality $|\Y|+1$ is sufficient for our purpose.

\item[Continuous case:]  When $X$ and $Y$ are continuous random variables with $\X=\Y=\R$, we associate information leakage with Bob's efficiency in estimating $X$ given $Z$. Consequently, we define $\L(X\to Z)$ to be $\frac{\var(X)}{\mmse(X|Z)},$	where $\var(X)\coloneqq\E[(X-\E[X])^2]$ is the variance of $X$ and $\mmse(X|Z)\coloneqq\E[(X-\E[X|Z])^2]$ is the minimum mean squared-error of $X$ given $Z$.
\end{description}

Returning to the setup of Fig.\ \ref{fig:XYZ}, recall that in order to receive a utility, Alice wishes to disclose her non-private information $Y$ to Bob. However, $Y$ might be correlated with her private information, represented by $X$. In order to quantify the tradeoff between information display and privacy leakage, we investigate the quantity
\begin{equation}\label{Def_Optimization}
\sup_{P_{Z|Y}:X\markov Y\markov Z\atop \L(X\to Z)\leq \eps} \L(Y\to Z).
\end{equation} 
We seek to characterize this constrained optimization problem in both the discrete and the continuous cases.  
It is worth mentioning that the chosen information leakage functions are special cases of leakage functions based on a large family of general loss functions, see the discussion in \cite[Section 6.2]{Shahab_PhD_thesis} and references therein. For example, Hamming and squared-error loss functions give rise to the proposed leakage functions in the discrete and continuous cases, respectively.

In the discrete case, the optimization problem in  \eqref{Def_Optimization} gives rise to the following definition.
\begin{definition}\label{Def_RPF}
	Let $(X,Y)$ be a pair of discrete random variables with joint distribution $P_{XY}$. We define the \emph{privacy-constrained guessing} function,
	$$\mathcalboondox{h}(P_{XY}, \,\cdot\,):[\cP(X),1]\to [0,1],$$
	by
	\begin{equation}\label{Def_h_eps}
	\mathcalboondox{h}(P_{XY}, \eps)\coloneqq \sup_{P_{Z|Y}:X\markov Y\markov Z\atop \cP(X|Z)\leq \eps}\cP(Y|Z).
	\end{equation}
\end{definition}
We write $\mathcalboondox{h}(\eps)$ whenever $P_{XY}$ is clear from the context.

    Let $H_\infty(X)\coloneqq -\log\cP(X)$  be the R\'{e}nyi entropy of order $\infty$ and $H_\infty(X|Z)\coloneqq-\log\cP(X|Z)$ be its conditional version \cite{Arimoto's_Conditional_Entropy}. It follows that $\cP(X|Z)=2^{-H_\infty(X|Z)}$ and $\cP(X)=2^{-H_\infty(X)}$. Then, $\mathcalboondox{h}$ is in correspondance with the function $g^{\infty}(P_{XY},\ \cdot\ ):\R^+\to\R^+$ defined by
    \begin{equation}\label{Def:g_infty}
    g^{\infty}(P_{XY},\eps)\coloneqq\sup_{P_{Z|Y}:X\markov Y\markov Z\atop I_\infty(X;Z)\leq \eps} I_\infty(Y;Z),
    \end{equation}
    where $I_\infty(X;Z)\coloneqq H_\infty(X)-H_\infty(X|Z)$ is  Arimoto's mutual information of order $\infty$ \cite{Arimoto_Original_Paper,csiszar95,Verdu_ALPHA}. Indeed, it is straightforward to show that \eqn{eq:ghFunctionalRelation}{g^\infty(P_{XY},\eps)=\log\frac{\mathcalboondox{h}(P_{XY},2^\eps\cP(X))}{\cP(Y)}.}
    The above functional relationship allows us to translate results for $\mathcalboondox{h}$ into results for $g^\infty$. 
    Two functions closely related to $g^\infty$ are the "rate-privacy function" \cite{Asoode_submitted}, defined as in \eqref{Def:g_infty} with $I_\infty$ replaced by Shannon's mutual information, and the "privacy funnel" \cite{Funnel} which is the dual representation of the rate-privacy function. Consequently, $g^\infty$ can be thought of as the \textit{rate-privacy function  of order $\infty$}. 
    
    In the machine learning literature, the {\em information bottleneck} (IB) method  has been proposed by Tishby et al.\ \cite{Bottleneck} to quantify a fundamental relevance-compression tradeoff. Specifically, the IB method minimizes the "compression rate" $I(Y;Z)$ subject to a relevance constraint given by $I(X;Z)\geq R$ for some $R\geq 0$. Thus, the IB problem is conceptually the dual of the privacy funnel problem. Recently, the privacy funnel and the IB function were unified in a single geometric framework \cite{Asoode_Flavio} which also encompasses the privacy funnel of order $\infty$ (or equivalently $g^\infty$) and its dual which may be called the \textit{IB function of order $\infty$}. The relation between the different properties of IB function (of order $\infty$) and the privacy funnel (of order $\infty$) within this framework is the subject of ongoing research.
     
It is important to note that Arimoto's mutual information of order $\infty$ differs from other notions of information leakage, for example the ones studied in \cite{Asoode_submitted, Issa_Sibson, Calmon_fundamental-Limit, Calmon_bounds_Inference}, in the fact that $I_\infty(X;Z)=0$ is not necessarily equivalent to $X$ and $Z$ being independent. Indeed, if $X\sim \sBer(p)$ with $p\in[\frac{1}{2},1]$ and $P_{Z|X}=\mathsf{BSC}(\alpha)$  with  $\alpha\in[0,\frac{1}{2}]$ (the binary symmetric channel with crossover probability $\alpha$), then $\cP(X)=p$ and $\cP(X|Z)=p\bar{\alpha}+\max\{\bar{p}\bar{\alpha}, \alpha p\}$, where $\bar{a}=1-a$. In this case, it is straightforward to verify that $\cP(X|Z)=\cP(X)$ if and only if $p\geq\bar{\alpha}$. Therefore, for $\frac{1}{2} < \bar{\alpha} \leq p < 1$, $I_\infty(X;Z)=0$ despite the fact that $X$ and $Z$ are not independent.

For continuous real-valued random variables $X$, $Y$, and $Z$, the optimization problem  in \eqref{Def_Optimization} is hard and seems intractable in general. In order to have a tractable model, we assume that the displayed data $Z$ is a Gaussian perturbation of $Y$, i.e., $Z=Z_\gamma\coloneqq \sqrt{\gamma}Y+N_\sG$, where $\gamma\geq 0$ and $N_\sG\sim\mathcal{N}(0, 1)$ is independent of $(X,Y)$. We thus consider the following privacy-utility tradeoff, which is a dual representation of \eqref{Def_Optimization} with the privacy constraint strengthened.

\begin{definition}\label{Def_StrongENSR}
	Let $(X,Y)$ be a pair of real-valued random variables with joint density $P_{XY}$. We define the strong estimation noise-to-signal ratio $\sM(P_{XY},\ \cdot\ ):\R^+\to\R^+$ by
	$$\sM(P_{XY},\eps)\coloneqq\inf_{\gamma\geq 0} \frac{\mmse(Y|Z_\gamma)}{\var(Y)},$$
	where the infimum is taken over all $\gamma\geq0$ such that
	$$\mmse(f(X)|Z_\gamma)\geq (1-\eps)\var(f(X))$$
	whenever $f:\R\to\R$ is measurable and $\var(f(X))<\infty$.
\end{definition}

	\subsection{Main Contributions}
We begin in Section~\ref{Sec:DiscreteScalarCase} by investigating the salient properties of $\mathcalboondox{h}$. In Theorem~\ref{Thm:PiecewiseLinearity}, we show that the map $\mathcalboondox{h}(P_{XY}, \cdot)$ is piecewise linear (Fig.~\ref{fig:Typical_h}). The proof relies on a geometric reformulation of $\mathcalboondox{h}$ and a careful study of the directional derivatives in the space of stochastic matrices. As a byproduct of Theorem~\ref{Thm:PiecewiseLinearity}, a formula for the derivative of $\mathcalboondox{h}$ at $\cP(X|Y)$ is established in \eqref{eq:DirivativehMinimum}. This formula, along with the concavity of $\mathcalboondox{h}$, permits us to obtain a tight upper bound for $\mathcalboondox{h}$. In particular, when $|\X|=|\Y|=2$, this upper bound and the chord lower bound for concave functions allow us to derive a closed form expression for $\mathcalboondox{h}$ in Theorem~\ref{Theorem_Linearity_BIBO}. Moreover, it is also shown that, depending on the backward channel $P_{X|Y}$, either a Z-channel or a \emph{reverse} Z-channel (Fig.~\ref{fig:Optimal_Filter_BIBO}) achieves $\mathcalboondox{h}(P_{XY},\eps)$ for each $\eps$.

We next consider a variant quantity $\underline{\mathcalboondox{h}}$ which we define analogously to $\mathcalboondox{h}$ except that $Z$ is required to be supported over $\Y$. By definition, $\underline{\mathcalboondox{h}}$ captures the fundamental trade-off between privacy and utility in situations where enlarging the alphabet is not possible. This is particularly relevant when the displayed data might be used by parties not aware of the implemented privatization scheme. The function $\underline{\mathcalboondox{h}}$ may not be concave and consequently the techniques developed to study $\mathcalboondox{h}$ do not apply. Nevertheless, we can still study the functional properties of $\ul{\mathcalboondox{h}}$ in the {\em high utility regime} (i.e., for sufficiently large privacy threshold $\eps$), deriving a closed form expression in Theorem~\ref{Thm:GeneralizedLocalLinearity}.
	
We then specialize Theorem~\ref{Thm:GeneralizedLocalLinearity} to the binary vector case. 
 Here, $Z^n$ is revealed publicly and the goal is to guess $Y^n$ under the privacy constraint $\cP(X^n|Z^n)\leq\eps^n$. We consider two models for the pair of random vectors $(X^n, Y^n)$. In the first model (Theorem~\ref{Prop:IIDDataUnderbar}), we assume that $X^n$ consists of $n$ independent and identically distributed (i.i.d.) $\sBer(p)$ samples with $p\in [\frac{1}{2}, 1)$. In the second model (Theorem~\ref{Theorem_MarkovMemory}), we assume that $X^n$ comprises the first $n$ samples of a first-order homogeneous Markov process having a simple symmetric transition matrix. We assume that in both cases $Y_k$, $k=1,\dots,n$, is the output of a $\mathsf{BSC}(\alpha)$, $\alpha\in [0, \frac{1}{2})$, whose input is $X_k$.  We also study in detail the problem of \textit{learning from a private distribution}, which corresponds to the special case $X_1=\dots=X_n$ of the second model (Proposition~\ref{Proposition_Parametric_Dis_Privacy}).

	In the continuous case, we first show that the strong privacy constraint in Definition~\ref{Def_StrongENSR} is equivalent to a condition on the maximal correlation (also referred to as the Hirschfeld-Gebelein-R\'enyi maximal correlation~\cite{hirschfild,gebelien,Renyi-dependence-measure}) between $X$ and $Z$.  We then derive the value of $\sM$ for the Gaussian case (Example~\ref{example_Gaussian}) and obtain sharp lower and upper bounds for general $X$ and Gaussian $Y$ in Theorem~\ref{Lemma_Gaussian_Y_Arbit_X}. Finally, we establish in Lemma~\ref{Lemma_Approx_S} a tight lower bound for $\sM(P_{XY}, \eps)$ for general $(X,Y)$ in the high privacy regime (i.e., sufficiently small $\eps$).

	\subsection{Related Work}
	There have been several choices proposed for an appropriate measure $\L$ of information leakage  in the information theory and computer science literature.  Shannon's mutual information $I(X;Z)$ (or equivalently the conditional entropy $H(X|Z)$), while an intuitively reasonable choice, does not lead to an arguably "operational" privacy guarantee and thus
	may not satisfactorily serve as an appropriate information leakage function, see \cite{Min_Entropy_leakage} and \cite{efv_MI_Deficiency}. Smith~\cite{Min_Entropy_leakage} discussed that the guessing entropy \cite{Massey_Guessing_Entropy} (defined as the expected number of guesses required to guess $X$ from $Z$) cannot be adopted as an information leakage function and then proposed Arimoto's mutual information of order $\infty$  as an appropriate notion of information leakage. Operationally, $I_\infty(X;Z)\leq \eps$ for sufficiently small $\eps$ implies that it is nearly as hard for an adversary observing $Z$ to guess $X$ as it is without $Z$. Braun et al.\ \cite{Quantitative_notion_Leakage} proposed the information leakage measures
	$\cP(X|Z)-\cP(X)$ and $\max I_\infty(X;Z)$, where the maximization is taken over all priors $P_X$.  In \cite{Inf_leakage_Barthe}, Barthe and K\"{o}pf  studied the latter quantity in the context of differential privacy \cite{dwork2}. 

	Issa et al.\ \cite{Issa_Sibson} recently found an interesting operational interpretation for $I^\mathsf{s}_\infty(X;Z)$, Sibson's mutual information of order $\infty$ \cite{sibson69,Verdu_ALPHA}. Specifically, they showed that the requirement $I^\mathsf{s}_\infty(X;Z)\leq \eps$ is equivalent to $I_\infty(U;Z)\leq \eps$  for \emph{all} auxiliary random variables $U$ satisfying $U\markov X\markov Z$. Consequently, this constraint guarantees that no \emph{randomized} function of $X$ can be efficiently estimated from $Z$, which leads to a strong privacy guarantee. In contrast, the privacy requirement $I_\infty(X;Z)\leq\eps$ only guarantees to keep $X$ itself private. Nonetheless, the latter requirement comes at a lower utility cost, as illustrated by the following example. Suppose that $X$ and $Y$ are binary and that Alice wishes to reveal absolutely no information about $X$ (i.e., perfect privacy) when disclosing a sanitized version of $Y$. According to the privacy constraint dictated by Sibson's mutual information, perfect privacy leads to the independence of $X$ and $Z$. It can be shown that for binary $Y$ and $X\markov Y\markov Z$,  independence of $X$ and $Z$ implies independence of $Y$ and $Z$ (cf \cite[Corollary 11]{Asoode_submitted}). Hence, perfect privacy under Sibson's mutual information results in trivial utility. However, as shown in Theorem~\ref{Theorem_Linearity_BIBO},  a non-trivial utility might be achieved for the perfect privacy requirement $I_\infty(X;Z)=0$.
		
	There exist other estimation-theoretic measures of information leakage in the literature. For example, Makhdoumi and Fawaz \cite{Fawaz_Makhdoumi} proposed to use maximal correlation $\rho_m$ as a measure of information leakage. Later, Calmon et al.\ \cite[Theorem 9]{Calmon_PIC} showed that if $X$ and $Z$ are discrete random variables, then $\cP(f(X)|Z)-\cP(f(X))\leq \rho_m(X,Z)$ for every function $f$, thus providing an interesting operational interpretation for maximal correlation as a measure of information leakage. Similarly, we show that
	$$\mmse(f(X)|Z)\geq (1-\rho_m^2(X,Z))\var(f(X))$$
	for every measurable real-valued function $f$. This then provides an operational interpretation for the privacy guarantee $\rho_m^2(X,Z)\leq \eps$ that we study in Section~\ref{Sec:ContinuousCase} for $X$ and $Y$ absolutely continuous random variables. We refer the readers to \cite{Wagner_Survey_Privacy_Leakage} for a fairly comprehensive list of existing information leakage measures.

	The study of the privacy-utility tradeoff from an information theoretic point of view was initiated by Yamamoto \cite{yamamotoequivocationdistortion} and further extended by several authors, see, e.g., \cite{Funnel,Calmon_fundamental-Limit,Sankar_Utility_privacy_forensics,Asoodeh_Allerton, Asoode_submitted, Lalitha_Smart_grid, Calmon_MaximumLeakage, Calmon_MI_Euclidean}. In relation with the present work, as already noted the rate-privacy function $g(P_{XY},\eps)$ was introduced in \cite{Asoode_submitted} as the maximum $I(Y; Z)$ over all privacy filters $P_{Z|Y}$ such that $I(X;Z)\leq \eps$ (the privacy funnel \cite{Funnel} is a dual representation of $g(P_{XY},\eps)$). Motivated by \cite{Calmon_bounds_Inference}, a more operational privacy-rate function $\tilde{g}(P_{XY}, \eps)$ was introduced also in \cite{Asoode_submitted} by replacing the privacy guarantee $I(X;Z)\leq \eps$ with $\rho_m^2(X,Z)\leq \eps$. It was also shown that $g(P_{XY}, \eps)$ can bound $\tilde{g}(P_{XY}, \eps)$ from above. 	

\subsection{Notation}

Throughout, we use capital letters, e.g., $X$, to denote random variables and lowercase letters, e.g., $x$, to denote their realizations. We use $X^n$ to denote the vector $(X_1, X_2, \dots, X_n)$. 
We let $\mathsf{Z}(\beta)$ denote the Z-channel with crossover probability $\beta$. For any $a\in [0, 1]$, we write $\bar{a}$ for $1-a$. As already mentioned, we let $\mathsf{BSC}(\alpha)$ denote the binary symmetric channel with crossover probability $\alpha$; we also  use $X\indep Z$ to indicate the independence of random variables $X$ and $Z$ and we write $X\markov Y\markov Z$ when $X$ and $Z$ are conditionally independent given $Y$ (i.e., when $X, Y,$ and $Z$ form a Markov chain in this order). Finally, for real-valued random variables $X$ and $Z$, the conditional variance of $X$ given $Z$ is given by $\var(X|Z) \coloneqq \E[(X-\E(X|Z))^2|Z]$.

\subsection{Organization}
The rest of the paper is organized as follows. We study the discrete case in Section~\ref{Sec:DiscreteScalarCase}. In particular, we determine $\mathcalboondox{h}$ in the binary case and obtain a tight lower bound for $\mathcalboondox{h}$ for general discrete alphabets in the high utility regime by studying $\underline{\mathcalboondox{h}}$. In Section~\ref{Section:DiscreteVectorCase}, we specialize our results to study $\underline{\mathcalboondox{h}}$ when $X^n$, $Y^n$, and $Z^n$ are binary random vectors. In Section~\ref{Sec:ContinuousCase}, we focus on the continuous case and obtain sharp bounds on $\sM$. We summarize our findings in Section~\ref{Section:Conclusion}. Finally, we point out that all proofs in the paper are deferred to the appendix.

\section{Discrete Scalar Case}
\label{Sec:DiscreteScalarCase}

In this section,  we assume that $X$ and $Y$ are finite-alphabet random variables taking values in $\X=\{1, \dots, M\}$ and $\Y=\{1, \dots, N\}$, respectively. Let $P(x,y)$ with $x\in\X$ and $y\in\Y$ be their joint distribution and $p_X$ and $q_Y$ the marginal distributions of $X$ and $Y$, respectively. The goal here is to maximize  the information leakage from $Y$ to $Z$ (i.e., utility) while ensuring that the information leakage from $X$ to $Z$ (i.e., privacy leakage) remains bounded. As stated earlier, we quantify the tradeoff between privacy and utility by means of $\mathcalboondox{h}$, as defined in \eqref{Def_h_eps}.

\subsection{Geometric Properties of $\mathcalboondox{h}$}
\label{Subsection:GeometricalFormulation}
First,  note that $\cP(X|Y,Z) \geq \cP(X|Z) \geq \cP(X)$ for jointly distributed  random variables $X$, $Y$ and $Z$. Therefore, from \eqref{Def_h_eps} we have that $\cP(Y)\leq \mathcalboondox{h}(\eps)\leq 1$ and that $\mathcalboondox{h}(\eps)=1$ if and only if $\eps\geq\cP(X|Y)$. Thus it is enough to study $\mathcalboondox{h}$ on the interval $[\cP(X), \cP(X|Y)]$.

An application of the Support Lemma \cite[Lemma 15.4]{csiszarbook} shows that it is enough to consider random variables $Z$ supported on $\Z=\{1,\ldots,N+1\}$.  Thus, the privacy filter $P_{Z|Y}$ can be realized by an $N\times (N+1)$ stochastic matrix $F\in  \M_{N\times (N+1)}$, where $\M_{N\times M}$ denotes the set of all real-valued $N\times M$ matrices. Let $\F$ be the set of all such matrices $F$. Then both privacy  $\P(P, F)=\cP(X|Z)$ and utility $\U(P,F)=\cP(Y|Z)$ are functions of $F\in \F$ and can be written as 
\begin{equation}
\begin{aligned}
\P(P, F) &\coloneqq\sum_{z=1}^{N+1} \max_{1\leq x\leq M}\sum_{y=1}^N P(x,y)F(y,z),\\ 	
\U(P,F) &\coloneqq \sum_{z=1}^{N+1} \max_{1\leq y\leq N} q(y)F(y,z).
\end{aligned}
 \label{eq:DefPU}
\end{equation}
In particular, we can express $\mathcalboondox{h}(\eps)$ as
\eqn{Def_h_Eps2}{\mathcalboondox{h}(\eps)=\sup_{F\in\F, \atop \P(P, F)\leq \eps} \U(P, F).}
As before, we omit $P$ in $\P(P, F)$ and $\U(P, F)$ when there is no risk of confusion.

It is straightforward to verify that $\P$ and $\U$ are continuous and convex on $\F$. As a consequence, for every $\eps\in[\cP(X),\cP(X|Y)]$, there exists $G\in \F$ such that $\P(G)=\eps$ and $\U(G)=\mathcalboondox{h}(\eps)$. It is then direct to show that $\mathcalboondox{h}$ is continuous on $[\cP(X),\cP(X|Y)]$. 
Using a proof technique similar to \cite[Theorem 2.3]{Witsenhausen_Wyner}, it can also be shown\footnote{Note that \cite[Theorem 2.3]{Witsenhausen_Wyner} deals with a similar problem where $\cP(X|Z)$ and $\cP(Y|Z)$ are replaced by $H(X|Z)$ and $H(Y|Z)$, respectively. Just as $(H(X|Z), H(Y|Z))$, the pair $(\cP(X|Z), \cP(Y|Z))$ can be written as a convex combination of points in a two-dimensional set. In our setting, this set turns out to be $\{(\cP(X'), \cP(Y')):Y'\sim q'\in \P(\Y)~\text{and}~X'\sim p', \text{where}~p'(x)=\sum_{y}P_{X|Y}(x|y)q'(y)\}$. See \cite{Asoode_Flavio} for a generalization of this argument.} that the graph of $\mathcalboondox{h}$ is the upper boundary of the two-dimensional convex set  $\{(\P(F), \U(F)):F\in \F\}$ and thus $\mathcalboondox{h}$ is concave and strictly increasing. The following theorem, which is the most important and technically difficult result of this paper, states that $\mathcalboondox{h}$ is a piecewise linear function, as illustrated in Fig. \ref{fig:Typical_h}.

\begin{theorem}
\label{Thm:PiecewiseLinearity}
The function $\mathcalboondox{h}:[\cP(X),\cP(X|Y)]\to\R^+$ is piecewise linear, i.e., there exist $K\geq1$ and thresholds $\cP(X)=\eps_0\leq\eps_1\leq\ldots\leq\eps_K=\cP(X|Y)$ such that $\mathcalboondox{h}$ is linear on $[\eps_{i-1}, \eps_i]$ for all $i=1, \dots, K$.
\end{theorem}

\begin{figure}
\centering
\begin{tikzpicture}[thick, scale=0.8]
\begin{axis}[xmin=0,xmax=1,ymin=0.4,ymax=1, y=10.5cm,
  samples=100,  y label style={at={(axis description cs:0.01,0.4)},rotate=-90,anchor=south},
  grid=both,xlabel= $\eps$, ylabel=${\mathcalboondox{h}(\eps)}$,
  no markers]
\addplot[no marks,blue,dotted, line width=1pt,samples=600] {0.923*x+0.215};
\addplot[domain=0.2:0.35][ no marks,red,line width=1pt,samples=600] {2*x};
\addplot[domain=0.35:0.5][no marks,red,line width=1pt,samples=600] {x+0.35};
\addplot[domain=0.5:0.65][no marks,red,line width=1pt,samples=600] {0.6*x+0.55};
\addplot[domain=0.65:0.85][no marks,red,line width=1pt,samples=600] {0.3*x+0.745};
\addplot[mark=none, black, dotted, line width=1pt] coordinates {(0.85,0) (0.85,1)};
\end{axis}
\node[below,black] at (5.91, -0.39) {\small{$\cP(X|Y)$}};
\node[below,black] at (1.25, -0.39) {\small{$\mathsf{P}_\mathsf{c}(X)$}};
\end{tikzpicture}
\caption{Typical $\mathcalboondox{h}$ and its trivial lower bound, the chord connecting $(\cP(X),\mathcalboondox{h}(\cP(X)))$ and $(\cP(X|Y), 1)$.}
\label{fig:Typical_h}
\end{figure}

The proof of this theorem, which is given in Appendix \ref{Appendix:ProofPiecewiseLinearity}, relies on the geometrical formulation of $\mathcalboondox{h}$. In particular, it is proved that  $\P$ and $\U$, are piecewise linear functions on $\F$. Using this fact, we establish the existence of a piecewise linear path of \emph{optimal filters} in $\F$. The proof technique allows us to derive the slope of $\mathcalboondox{h}$ on $[\eps_{i-1}, \eps_i]$, given the family of optimal filters  at a single point $\eps\in [\eps_{i-1}, \eps_i]$. For example, since the family of optimal filters at $\eps=\cP(X|Y)$ is easily obtainable, it is possible to compute $\mathcalboondox{h}$ on the last interval. We utilize this observation in Section \ref{Subsection:BinaryCase} to prove that in the binary case $\mathcalboondox{h}$ is indeed linear. 

\subsection{Perfect Privacy}
\label{Subsection:PerfectPrivacy}
When $\eps=\cP(X)$, observing $Z$ does not increase the probability of guessing $X$. In this case we say that perfect privacy holds. An interesting problem is to characterize when non-trivial utility can be obtained under perfect privacy, that is,  to characterize when  $\mathcalboondox{h}(\cP(X))>\cP(Y)$ holds. To the best of our knowledge, a general necessary and sufficient condition for this requirement is unknown.

Note that $\mathcalboondox{h}(\cP(X))>\cP(Y)$ is equivalent to $g^{\infty}(0)>0$. As opposed to the Shannon mutual information, $I_\infty(X;Z)=0$ does not necessarily imply that $X\indep Z$. In particular, the \emph{weak independence}\footnote{$X$ is said to be weakly independent of $Z$ if the vectors $\{P_{X|Z}(\cdot|z):~z\in \Z\}$ are linearly dependent \cite{berger}.} argument from \cite[Lemma 10]{Asoode_submitted} (see also \cite{Calmon_fundamental-Limit}) cannot be applied for $g^\infty$. However, we have the following result whose proof is given in Appendix~\ref{Appenddix:Proposition}.

\begin{proposition}\label{Proposition_Independent}
\label{Prop:MutualInformationIndependenceUniform}
Let $(X,Z)$ be a pair of random variables with $X$ uniformly distributed. If $I_\infty(X;Z)=0$, then $X\indep Z$.
\end{proposition}

As a consequence of Proposition~\ref{Prop:MutualInformationIndependenceUniform}, when $X$ and $Y$ are uniformly distributed, one can apply the weak independence arguments from \cite[Lemma 10]{Asoode_submitted} to obtain the following. 
\begin{corollary}
If $X$ and $Y$ are uniformly distributed, then $g^\infty(0)>0$ if and only if $X$ is weakly independent of $Y$.
\end{corollary}

When $X$ is uniform, the privacy requirement $I_\infty(X;Z)\leq \eps$ guarantees that an adversary observing $Z$ cannot efficiently estimate any arbitrary \emph{randomized function} of $X$. To see this, consider a random variable $U$ satisfying $U\markov X\markov Z$. Then we have 
\begin{eqnarray*}
\cP(U|Z) &=&  \sum_{z\in \Z}\max_{u\in \U} \sum_{x\in \X}P_{UX}(u, x)P_{Z|X}(z|x)\\
 &\leq&  \sum_{z\in \Z} \left[\max_{x\in \X}P_{Z|X}(z|x)\right] \left[\max_{u\in \U} \sum_{x\in \X}P_{UX}(u, x)\right]\\
 &=&\frac{\cP(X|Z)\cP(U)}{\cP(X)},
\end{eqnarray*}
which can be rearranged to yield  $I_\infty(U;Z)\leq I_\infty(X;Z)$. It is worth mentioning that the data processing inequality for $I_\infty$  \cite{Arimoto's_Conditional_Entropy}  states that  $I_\infty(Z;U)\leq I_\infty(Z;X)$. However,  $I_\infty(Z;U)$ is not necessarily equal to $I_\infty(U; Z)$.

\subsection{Binary Case}
\label{Subsection:BinaryCase}

A channel $\mathsf{W}$ is called a binary input binary output channel with crossover probabilities $\alpha$ and $\beta$, denoted by $\mathsf{BIBO}(\alpha,\beta)$, if $\mathsf{W}(\cdot|0)=(\bar{\alpha}, \alpha)$ and $\mathsf{W}(\cdot|1)=(\beta, \bar{\beta})$. Note that if $X\sim\sBer(p)$ with $p\in[\frac{1}{2},1)$ and $P_{Y|X}=\mathsf{BIBO}(\alpha,\beta)$ with $\alpha,\beta\in[0,\frac{1}{2})$, then $\cP(X)=p$ and
$$\cP(X|Y)=\max\{\bar{\alpha}\bar{p}, \beta p\}+\bar{\beta}p.$$
In this case, if $\bar{\alpha}\bar{p}\leq\beta p$ then $\cP(X|Y)=p=\cP(X)$ and hence $\mathcalboondox{h}(p)=1$. The following theorem, whose proof is given in Appendix~\ref{Appendix_Linearity}, establishes the linear behavior of $\mathcalboondox{h}$ in the non-trivial case $\bar{\alpha}\bar{p}>\beta p$.

\begin{theorem}
\label{Theorem_Linearity_BIBO}
Let $X\sim\sBer(p)$ with $p\in[\frac{1}{2},1)$ and $P_{Y|X}=\mathsf{BIBO}(\alpha,\beta)$ with $\alpha,\beta\in[0,\frac{1}{2})$ such that $\bar{\alpha}\bar{p}>\beta p$. Then, for any $\eps\in[p,\bar{\alpha}\bar{p}+\bar{\beta}p]=[\cP(X), \cP(X|Y)]$,
\eq{\mathcalboondox{h}(\eps)=\begin{cases}1-\zeta(\eps) q, & \alpha\bar{\alpha}\bar{p}^2<\beta\bar{\beta}p^2,\\1-\tilde{\zeta}(\eps) \bar{q}, & \alpha\bar{\alpha}\bar{p}^2\geq\beta\bar{\beta}p^2,\end{cases}}
where $q\coloneqq q_Y(1)=\alpha\bar{p}+\bar{\beta}p$,
\begin{equation}\label{Def_Zeta_Zeta_tilde}
\zeta(\eps)\coloneqq\frac{\bar{\alpha}\bar{p}+\bar{\beta}p-\eps}{\bar{\beta} p-\alpha\bar{p}}, \quad \text{and} \quad \tilde{\zeta}(\eps)\coloneqq\frac{\bar{\alpha}\bar{p}+\bar{\beta}p-\eps}{\bar{\alpha}\bar{p}-\beta p}.
\end{equation}
Furthermore, the Z-channel $\mathsf{Z}(\zeta(\eps))$ and the reverse Z-channel $\tilde{\mathsf{Z}}(\tilde{\zeta}(\eps))$ achieve $\mathcalboondox{h}(\eps)$ when $\alpha\bar{\alpha}\bar{p}^2<\beta\bar{\beta}p^2$ and  $\alpha\bar{\alpha}\bar{p}^2\geq\beta\bar{\beta}p^2$, respectively. The optimal privacy filters are depicted in Fig.~\ref{fig:Optimal_Filter_BIBO}.
\end{theorem}
\begin{figure}[t]
\subfigure[$\alpha\bar{\alpha}\bar{p}^2<\beta\bar{\beta}p^2$.]{
\begin{tikzpicture}[scale=1.5]
\node (a) [circle] at (0,0) {$1~$};
\node (b) [circle] at (0,1.4) {$0~~$};
\node (c) [circle] at (2,0) {$~1$};
\node (d) [circle] at (2,1.4) {$~0$};
\node (ee1) [circle] at (2.1,0.7) {$Y$};
\node (ee2) [circle] at (-0.1,0.7) {$X$};
\node (ee3) [circle] at (4.35,0.7) {$Z$};
\draw[->] (0.15,0) -- (1.85,0) node[pos=.5,sloped,below] {};
\draw[->] (0.15,0) -- (1.85,1.4) node[pos=.35,sloped,below] {\footnotesize{$\beta$}};
\draw[->] (0.15,1.4) -- (1.85,0) node[pos=.35,sloped,above] {\footnotesize{$\alpha$}};
\draw[->] (0.15,1.4) -- (1.85,1.4) node[pos=.5,sloped,above] {};
\node (e) [circle] at (2.3,0) {};
\node (f) [circle] at (2.3,1.4) {};
\node (g) [circle] at (4.27,0) {$~1$};
\node (h) [circle] at (4.27,1.4) {$~0$};
\draw[->] (2.3,0) -- (4,0) node[pos=.5,sloped,below] {\footnotesize{1-$\zeta(\eps)$}};
\draw[->] (2.3,0) -- (4, 1.4) node[pos=.5,sloped,above] {\footnotesize{$\zeta(\eps)$}};
\draw[->] (2.3,1.4) -- (4,1.4) node[pos=.5,sloped,above] {};
\end{tikzpicture}}
\hfill
  \subfigure[$\alpha\bar{\alpha}\bar{p}^2\geq\beta\bar{\beta}p^2$.]{
\begin{tikzpicture}[scale=1.5]
        \node (a) [circle] at (0,0) {$1~$};
\node (b) [circle] at (0,1.4) {$0~~$};
\node (c) [circle] at (2,0) {$~1$};
\node (d) [circle] at (2,1.4) {$~0$};
\node (ee1) [circle] at (2.1,0.7) {$Y$};
\node (ee2) [circle] at (-0.1,0.7) {$X$};
\node (ee3) [circle] at (4.35,0.7) {$Z$};
\draw[->] (0.15,0) -- (1.85,0) node[pos=.5,sloped,below] {};
\draw[->] (0.15,0) -- (1.85,1.4) node[pos=.35,sloped,below] {\footnotesize{$\beta$}};
\draw[->] (0.15,1.4) -- (1.85,0) node[pos=.35,sloped,above] {\footnotesize{$\alpha$}};
\draw[->] (0.15,1.4) -- (1.85,1.4) node[pos=.5,sloped,above] {};
\node (e) [circle] at (2.3,0) {};
\node (f) [circle] at (2.3,1.4) {};
\node (g) [circle] at (4.27,0) {$~1$};
\node (h) [circle] at (4.27,1.4) {$~0$};
\draw[->] (2.3,0) -- (4,0) node[pos=.5,sloped,below] {};
\draw[->] (2.3,1.4) -- (4,1.4) node[pos=.5,sloped,above] {\footnotesize{$1-\tilde{\zeta}(\eps)$}};
\draw[->] (2.3,1.4) -- (4, 0) node[pos=.5,sloped,above]{\footnotesize{$\tilde{\zeta}(\eps)$}}; 
\end{tikzpicture}}
  \caption{Optimal privacy mechanisms in Theorem \ref{Theorem_Linearity_BIBO}.}
  \label{fig:Optimal_Filter_BIBO}
\end{figure}
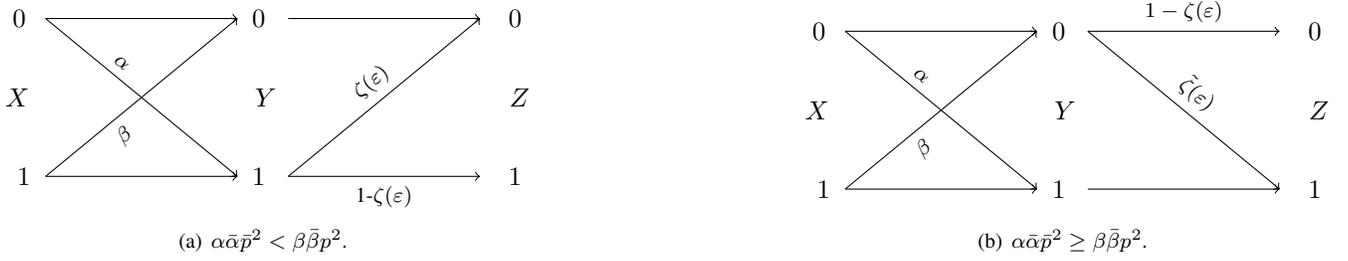
Note that the condition $\alpha\bar{\alpha}\bar{p}^2< \beta\bar{\beta}p^2$ is equivalent to
$$P_{X|Y}(1|1)>P_{X|Y}(0|0),$$
and that $P_{X|Y}(0|0) > \frac{1}{2}$ whenever $\bar{\alpha} \bar{p} > \beta p$. Hence, intuitively speaking, the event $Y=1$ reveals more information about $X$ than the event $Y=0$. Consequently, an optimal privacy mechanism $\M$ needs to distort the event  $Y=1$. 

Under the hypotheses of Theorem~\ref{Theorem_Linearity_BIBO}, there exists a Z-channel for every $\eps\in[\cP(X),\cP(X|Y)]$ that achieves $\mathcalboondox{h}(\eps)$. A minor modification to the proof of Theorem~\ref{Theorem_Linearity_BIBO} shows that the Z-channel is the only binary privacy filter with this optimality property for $p\in (\frac{1}{2}, 1)$. It is worth mentioning that in the symmetric case ($\alpha=\beta$) with uniform input ($p=\frac{1}{2}$), the channel $\mathsf{BSC}(0.5\zeta(\eps))$ can be shown to also achieve $\mathcalboondox{h}(\eps)$.


It is straightforward to show that $1-\zeta(p)q > \bar{q}$ if and only if $p\in(\frac{1}{2},1)$, and $1-\zeta(p) q > q$ if and only if $\alpha\bar{\alpha}\bar{p}^2< \beta\bar{\beta}p^2$. Also, note that $\mathcalboondox{h}(p)=q$ when $\alpha\bar{\alpha}\bar{p}^2\geq \beta\bar{\beta}p^2$.
In particular, we have the following necessary and sufficient condition for the non-trivial utility under perfect privacy.
\begin{corollary}
Let $X\sim\sBer(p)$ with $p\in[\frac{1}{2},1)$ and $P_{Y|X}=\mathsf{BIBO}(\alpha,\beta)$ with $\alpha,\beta\in[0,\frac{1}{2})$ such that $\bar{\alpha}\bar{p}>\beta p$. Then $g^\infty(0)>0$ if and only if $\alpha\bar{\alpha}\bar{p}^2<\beta\bar{\beta}p^2$ and $p\in(\frac{1}{2},1)$.
\end{corollary}

\subsection{A variant of $\mathcalboondox{h}$}
	Thus far, we studied the privacy-constrained guessing probability $\mathcalboondox{h}(\eps)$ where no constraint on the cardinality of the alphabet of the displayed  data $Z$ is imposed (other than being finite). Nevertheless, we know that it is sufficient to consider $\Z$ with cardinality $|\Y|+1$. However, as mentioned in the introduction, it may be desirable to generate the displayed data on the same alphabet as that of $Y$. In this section, we consider the case where $\Z$ is constrained to satisfy $|\Z|=|\Y|$, which leads to the following variant of  $\mathcalboondox{h}$, denoted by $\underline{\mathcalboondox{h}}$.   

\begin{definition}\label{eq:DefPsi}
	For arbitrary discrete random variables $X$ and $Y$ supported on $\X$ and $\Y$ respectively, we define the function $\underline{\mathcalboondox{h}}:[\cP(X),\cP(X|Y)]\to\R^+$ by
	$$\underline{\mathcalboondox{h}}(\eps) \coloneqq \sup_{P_{Z|Y} \in \ul{\mathfrak{D}}_\eps} \cP(Y|Z),$$
	where $$\ul{\mathfrak{D}}_\eps \coloneqq \left\{P_{Z|Y} : \Z=\Y,X\markov Y\markov Z, \cP(X|Z)\leq \eps\right\}.$$
\end{definition}
Unlike $\mathcalboondox{h}$, the definition of $\underline{\mathcalboondox{h}}$ requires $\Z=\Y$. This difference makes the tools from \cite{Witsenhausen_Wyner} unavailable. In particular, the concavity and hence the piecewise linearity of $\mathcalboondox{h}$ do not carry over to $\underline{\mathcalboondox{h}}$. However, we have the following theorem for $\underline{\mathcalboondox{h}}$ whose proof is given in Appendix~\ref{Appendix:ProofGeneralizedLocalLinearity}.  For notational convenience, we adopt the convention $\frac{x}{0} = +\infty$ for $x>0$. For $(y_0,z_0)\in\Y\times\Y$, a channel $\mathsf{W}$ is said to be an $N$-ary Z-channel  with crossover probability $\gamma$ from $y_0$ to $z_0$, denoted by $\mathsf{Z}^{y_0,z_0}(\gamma)$, if the input and output alphabets are $\Y$ and $\mathsf{W}(y|y)=1$ for $y\neq y_0$, $\mathsf{W}(z_0|y_0)=\gamma$, and $\mathsf{W}(y_0|y_0)=\bar{\gamma}$. We also let $\underline{\mathcalboondox{h}}'(\cP(X|Y))$ denote the left derivative of $\underline{\mathcalboondox{h}}(\cdot)$ evaluated at $\eps=\cP(X|Y)$.
\begin{theorem}
	\label{Thm:GeneralizedLocalLinearity}
	Let $X$ and $Y$ be discrete random variables. If $\cP(X)<\cP(X|Y)$, then there exists $\eps_\mathsf{L} \in (\cP(X),\cP(X|Y))$ such that $\underline{\mathcalboondox{h}}$ is linear on $[\eps_\mathsf{L},\cP(X|Y)]$. In particular, for every $\eps\in[\eps_\mathsf{L},\cP(X|Y)]$,
	\eqn{eq:PsiExtremePoint}{\underline{\mathcalboondox{h}}(\eps) = 1 - (\cP(X|Y)-\eps) \underline{\mathcalboondox{h}}'(\cP(X|Y)).}
	Moreover, if $q_Y(y)>0$ for all $y\in\Y$ and for each $y\in\Y$ there exists (a unique) $x_y\in\X$ such that $P_{X|Y}(x_y|y) > P_{X|Y}(x|y)$ for all $x\neq x_y$, then
	\eqn{eq:DerivativePsi}{\underline{\mathcalboondox{h}}'(\cP(X|Y)) = \min_{(y,z)\in\Y\times\Y} \frac{q_Y(y)}{P_{XY}(x_y,y) - P_{XY}(x_z,y)}.}
	In addition, if $(y_0,z_0)\in\Y\times\Y$ attains the minimum in \eqref{eq:DerivativePsi}, then there exists $\eps_\mathsf{L}^{y_0,z_0}<\cP(X|Y)$ such that $\mathsf{Z}^{y_0,z_0}(\zeta^{y_0,z_0}(\eps))$ achieves $\underline{\mathcalboondox{h}}(\eps)$ for every $\eps\in[\eps_\mathsf{L}^{y_0,z_0},\cP(X|Y)]$, where
	\eq{\zeta^{y_0,z_0}(\eps) = \frac{\cP(X|Y) - \eps}{P_{XY}(x_{y_0},y_0) - P_{XY}(x_{z_0},y_0)}.}
\end{theorem}

It is clear, from Definition~\ref{eq:DefPsi}, that $\underline{\mathcalboondox{h}}(\eps) \leq \mathcalboondox{h}(\eps)$ for all $\eps\in[\cP(X),\cP(X|Y)]$. Hence, Theorem~\ref{Thm:GeneralizedLocalLinearity} provides a lower bound for $\mathcalboondox{h}$ in the high utility regime.

Although \eqref{eq:PsiExtremePoint} establishes the linear behavior of $\underline{\mathcalboondox{h}}$ over  $[\eps_{\mathsf{L}}, \cP(X|Y)]$ for general $X$ and $Y$, a priori it is not clear how to obtain $\underline{\mathcalboondox{h}}'(\cP(X|Y))$. Under the assumptions of Theorem~\ref{Thm:GeneralizedLocalLinearity}, \eqref{eq:DerivativePsi} expresses $\underline{\mathcalboondox{h}}'(\cP(X|Y))$ as the minimum of \emph{finitely} many numbers, and a suitable $\mathsf{Z}$-channel achieves $\underline{\mathcalboondox{h}}$ for $\eps$ close to $\cP(X|Y)$. As we will see in the following section, these assumptions are rather general and allow us to derive a closed form expression for $\underline{\mathcalboondox{h}}$ in the high utility regime for some pairs of binary random \emph{vectors} $(X^n, Y^n)$ with  $X^n, Y^n\in\{0, 1\}^n$.

\section{Binary Vector Case}
\label{Section:DiscreteVectorCase}
We next study privacy aware guessing for a pair of binary random vectors $(X^n, Y^n)$. First note that since having more side information only improves the probability of correct guessing,  one can write
\begin{equation*}
\cP(X^n)\leq\cP(X^n|Z^n)\leq \cP(X^n|Y^n,Z^n)=\cP(X^n|Y^n)
\end{equation*}
for $X^n\markov Y^n\markov Z^n$ and thus, we can restrict $\eps^n$ in the following definition to $[\cP(X^n), \cP(X^n|Y^n))]$.
\begin{definition}
For a given pair of binary random vectors $(X^n, Y^n)$, let $\underline{\mathcalboondox{h}}_n:[\cP^{1/n}(X^n),\cP^{1/n}(X^n|Y^n)]\to \R^+$ be the function defined by
\begin{equation}
\label{Def_h_n}
\underline{\mathcalboondox{h}}_n(\eps) \coloneqq \sup_{P_{Z^n|Y^n}\in\ul{\mathfrak{D}}_{n,\eps}} \cP^{1/n}(Y^n | Z^n),
\end{equation}
where \ndsty{\ul{\mathfrak{D}}_{n,\eps} \coloneqq \{P_{Z^n|Y^n} : \Z^n=\{0,1\}^n,X^n\markov Y^n\markov Z^n, \cP^{1/n}(X^n|Z^n) \leq \eps\}}.
\end{definition}
Note that this definition does not make any assumption about the privacy filters $P_{Z^n|Y^n}$ apart from $\Z^n=\{0,1\}^n$. Nonetheless, this restriction makes the functional properties of $\underline{\mathcalboondox{h}}_n$ different from those of $\mathcalboondox{h}$. 

We study $\underline{\mathcalboondox{h}}_n$ in the following two scenarios for $(X^n, Y^n)$:
\begin{itemize}
	\item[($\textnormal{a}_1$)] $\repdc{X}{1}{n}$ are i.i.d. samples drawn from $\sBer(p)$,
	\item[($\textnormal{a}_2$)] $X_1\sim \sBer(p)$ and $X_k=X_{k-1}\oplus U_k$ for $k=2,\dots,n$, where $\repdc{U}{2}{n}$ are i.i.d. samples drawn from $\sBer(r)$  and  independent of $X_1$, and $\oplus$ denotes mod 2 addition,
\end{itemize}
and in both cases, we assume that
\begin{itemize}
	\item[($\textnormal{b}$)] $Y_k=X_k\oplus V_k$ for $k=1,\dots,n$, where $\repdc{V}{1}{n}$ are i.i.d. samples drawn from  $\sBer(\alpha)$ and  independent of $X^n$. 
\end{itemize}
We study $\underline{\mathcalboondox{h}}_n$ for $(X^n, Y^n)$ satisfying the assumptions ($\textnormal{a}_1$) and ($\textnormal{b}$) in Section~\ref{Subsection:IIDPrivateData} and for $(X^n, Y^n)$ satisfying the assumptions ($\textnormal{a}_2$) and ($\textnormal{b}$) in Section~\ref{Section:Memory}. In the latter section, we also study $\underline{\mathcalboondox{h}}_n$ in the special case $r=0$ in more detail.
\subsection{I.I.D.\  Case}
\label{Subsection:IIDPrivateData}
Here, we assume that $(X^n, Y^n)$ satisfy ($\textnormal{a}_1$) and ($\textnormal{b}$) and  apply Theorem~\ref{Thm:GeneralizedLocalLinearity} to derive a closed form expression for $\underline{\mathcalboondox{h}}_n(\eps)$ for $\eps$ close to $\cP(X^n|Y^n)$. Additionally, we determine an optimal filter in the same regime. 

We begin by identifying  the domain $[\cP(X^n), \cP(X^n|Y^n)]$ of $\underline{\mathcalboondox{h}}_n$  in the following lemma, whose proof follows directly from the definition of $\cP$.
\begin{lemma}\label{Lemma_P_C_IID}
	Assume that $(X_1,Z_1),\ldots,(X_n,Z_n)$ are independent pairs of random variables. Then
	\eq{\cP(X^n|Z^n) = \prod_{k=1}^n \cP(X_k|Z_k).}
\end{lemma}
Thus, according to this lemma, if $p\in[\frac{1}{2},1)$ and $\alpha\in[0,\bar{p})$ then $\cP(X^n)=p^n$ and $\cP(X^n|Y^n)=\bar{\alpha}^n$.  The following theorem, whose proof is given in Appendix~\ref{Appendix:ProofPropIIDDataUnderbar}, is a straightforward consequence of Theorem~\ref{Thm:GeneralizedLocalLinearity}. 
A channel $\mathsf{W}$ is said to be a $2^n$-ary Z-channel with crossover probability $\gamma$, denoted by $\mathsf{Z}_n(\gamma)$, if its input and output alphabets are $\{0, 1\}^n$ and it is $\mathsf{Z}^{\bf 1, \bf 0}(\gamma)$, where ${\bf 0}=(0, 0, \dots, 0)$ and ${\bf 1}=(1, 1, \dots, 1)$.
\begin{figure}[t]
	\centering
	\begin{tikzpicture}[scale=0.9]
	\node (a) [circle] at (0,0) {$11~~$};
	\node (b) [circle] at (0,1.8) {$01~~$};
	\node (c) [circle] at (2,0) {$~~~~~~~11~~~~~~$};
	\node (d) [circle] at (2,1.8) {$~~~~~~~01~~~~~~$};
	\node (a1) [circle] at (0,0.9) {$10~~$};
	\node (b1) [circle] at (2,0.9) {$~~~~~~~10~~~~~~$};
	\node (c1) [circle] at (2,2.7) {$~~~~~~~00~~~~~~$};
	\node (d1) [circle] at (0,2.7) {$00~~$};
	\draw[->] (0.15,0) -- (1.83,0) node[pos=.5,sloped,below] {};
	\draw[->] (0.15,0) -- (1.83,0.8) node[pos=.5,sloped,below] {};
	\draw[->] (0.15,0) -- (1.83,1.7) node[pos=.5,sloped,below] {};
	\draw[->] (0.15,0) -- (1.83,2.6) node[pos=.5,sloped,below] {};
	\draw[->] (0.15,1.7) -- (1.83,0) node[pos=.5,sloped,below] {};
	\draw[->] (0.15,0.8) -- (1.83,0.8) node[pos=.5,sloped,below] {};
	\draw[->] (0.15,1.7) -- (1.83,1.7) node[pos=.5,sloped,below] {};
	\draw[->] (0.15,2.6) -- (1.83,2.6) node[pos=.5,sloped,below] {};
	\draw[->] (0.15,2.6) -- (1.83,0) node[pos=.5,sloped,below] {};
	\draw[->] (0.15,1.7) -- (1.83,0.8) node[pos=.5,sloped,below] {};
	\draw[->] (0.15,0.8) -- (1.83,1.7) node[pos=.5,sloped,below] {};
	\draw[->] (0.15,0.8) -- (1.83,2.6) node[pos=.5,sloped,below] {};
	\draw[->] (0.15,0.8) -- (1.83,0) node[pos=.5,sloped,below] {};
	\draw[->] (0.15,2.6) -- (1.83,0.8) node[pos=.5,sloped,below] {};
	\draw[->] (0.15,2.6) -- (1.83,1.7) node[pos=.5,sloped,below] {};
	\draw[->] (0.15,1.7) -- (1.83,2.6) node[pos=.5,sloped,below] {};
	\node (g) [circle] at (4.47,0) {$~~~11~~~$};
	\node (h) [circle] at (4.47,0.9) {$~~~10~~$};
	\node (g1) [circle] at (4.47,1.8) {$~~~01~~$};
	\node (h1) [circle] at (4.47,2.6) {$~~~00~~$};
	\draw[->] (2.35,0) -- (4.2,2.6) node[pos=.65,sloped,below] {};
	\draw[->] (2.35,0) -- (4.2, 0) node[pos=.55,sloped,below] {\footnotesize{1-$\zeta_2(\eps)$}};
	\draw[->] (2.35,0.8) -- (4.2,0.8) node[pos=.5,sloped,above] {};
	\draw[->] (2.35,1.7) -- (4.2,1.7) node[pos=.5,sloped,above] {};
	\draw[->] (2.35,2.6) -- (4.2,2.6) node[pos=.5,sloped,above] {};
	\end{tikzpicture}
	\caption{The optimal mechanism for $\underline{\mathcalboondox{h}}_2(\eps)$ for $\eps\in[\eps_\mathsf{L}, \bar{\alpha}]$.}
	\label{fig:Optimal_Vector}
\end{figure}
\begin{theorem}
\label{Prop:IIDDataUnderbar}
Assume that $(X^n, Y^n)$ satisfy ($\textnormal{a}_1$) and ($\textnormal{b}$)  with $p\in[\frac{1}{2},1)$ and $\alpha\in[0,\frac{1}{2})$ such that $\bar{\alpha}>p$. Then there exists $\eps_\mathsf{L}<\bar{\alpha}$ such that, for all $\eps\in[\eps_\mathsf{L},\bar{\alpha}]$,
\eq{\underline{\mathcalboondox{h}}_n^n(\eps) = 1-\zeta_n(\eps) q^n}
where $q\coloneqq\alpha\bar{p}+\bar{\alpha}p$ and
\eq{\zeta_n(\eps)\coloneqq\frac{\bar{\alpha}^n-\eps^n}{(\bar{\alpha}p)^n-(\alpha\bar{p})^n}.}
Moreover, the $2^n$-ary Z-channel $\mathsf{Z}_n(\zeta_n(\eps))$ achieves $\underline{\mathcalboondox{h}}_n(\eps)$ in this interval.
\end{theorem}

The optimal privacy mechanism achieving $\underline{\mathcalboondox{h}}_2(\eps)$ is depicted in Fig.~\ref{fig:Optimal_Vector}. From an implementation point of view, the simplest family of privacy mechanisms consists of those mechanisms for which $Z_k$ is a noisy version of $Y_k$ for each $k=\repn{n}$. Specifically, the family of mechanisms that generate $Z_k$, given $Y_k$, using a \emph{single} BIBO channel $\mathsf{W}$, and thus
\eqn{eq:Defhni}{P_{Z^n|Y^n}(z^n|y^n) = \prod_{k=1}^n \mathsf{W}(z_k|y_k),}
for all $y^n,z^n\in\{0,1\}^n$. Now, let $\mathcalboondox{h}_n^\mathsf{i}(\eps)=\sup \cP^{1/n}(Y^n|Z^n)$, where the supremum is taken over all $P_{Z^n|Y^n}$ satisfying \eqref{eq:Defhni} and $\cP^{1/n}(X^n|Z^n)\leq\eps$. It is clear that $\mathcalboondox{h}_n^\mathsf{i}(\eps) \leq \underline{\mathcalboondox{h}}_n(\eps)$ for all $\eps\in[\cP^{1/n}(X^n),\cP^{1/n}(X^n|Y^n)]$. 
The following proposition, whose proof is given in Appendix~\ref{Appendix:PropositionIID}, shows that if we restrict the privacy filter $P_{Z^n|Y^n}$ to be memoryless, then the optimal filter coincides with the optimal filter in the scalar case, which in this case is $ \mathsf{Z}(\zeta(\eps))$ as defined in Theorem~\ref{Theorem_Linearity_BIBO}.
\begin{proposition}
	\label{Propo_h_iid}
	Assume that $(X^n, Y^n)$ satisfy ($\textnormal{a}_1$) and ($\textnormal{b}$) with $p\in[\frac{1}{2},1)$ and $\alpha\in[0,\frac{1}{2})$ such that $\bar{\alpha}>p$. Then, for all $\eps\in[p,\bar{\alpha}]$,
	\eq{\mathcalboondox{h}_n^\mathsf{i}(\eps) = 1-\zeta(\eps)q,}
	where $q\coloneqq\alpha\bar{p}+\bar{\alpha}p$ and \dsty{\zeta(\eps)\coloneqq\frac{\bar{\alpha}\bar{p}+\bar{\alpha}p-\eps}{\bar{\alpha} p-\alpha\bar{p}}}.
\end{proposition}

It must be noted that, despite the fact that $(X^n, Y^n)$ is i.i.d., the memoryless privacy filter associated to $\mathcalboondox{h}_n^\mathsf{i}(\eps)$ is not optimal, as $\underline{\mathcalboondox{h}}_n(\eps)$ is a function of $n$ while $\mathcalboondox{h}_n^\mathsf{i}(\eps)$ is not. The following corollary, whose proof is given in Appendix~\ref{Appendix:ProofCorollaryDifferenceGap}, bounds the loss resulting from using a memoryless filter instead of an optimal one for $\eps\in [\eps_\mathsf{L},\bar{\alpha}]$. Clearly, for $n=1$, there is no gap as $\underline{\mathcalboondox{h}}_1(\eps)=\mathcalboondox{h}(\eps)=\mathcalboondox{h}^\mathsf{i}_1(\eps)$.
\begin{corollary}
\label{Corollary:DifferenceGap}
Let $(X^n,Y^n)$ satisfy ($\textnormal{a}_1$) and ($\textnormal{b}$) with $p\in[\frac{1}{2},1)$ 
and $\alpha\in [0,\frac{1}{2})$ such that $\bar{\alpha}>p$. Let $\eps_\mathsf{L}$ be as in Theorem~\ref{Prop:IIDDataUnderbar}. If $p>\frac{1}{2}$ and $\alpha>0$, then for $\eps\in [\eps_\mathsf{L}, \bar{\alpha}]$ and sufficiently large $n$
\begin{equation}\label{LB_Gap}
\underline{\mathcalboondox{h}}_n(\eps)-\mathcalboondox{h}^\mathsf{i}_n(\eps)\geq (\bar{\alpha}-\eps)[\Phi(1)-\Phi(n)],
\end{equation}
where $q=\alpha\bar{p}+\bar{\alpha}p$ and $$\Phi(n)\coloneqq \frac{q^n\bar{\alpha}^{n-1}}{(\bar{\alpha}p)^n-(\alpha\bar{p})^n}.$$
If $p=\frac{1}{2}$, then
\begin{equation}\label{UB_Gap}
 \mathcalboondox{h}^\mathsf{i}_n(\eps) \leq \underline{\mathcalboondox{h}}_n(\eps)\leq \mathcalboondox{h}^\mathsf{i}_n(\eps) + \frac{\alpha}{2\bar{\alpha}}, 
\end{equation}
for every $n\geq 1$ and $\eps\in [\eps_\mathsf{L}, \bar{\alpha}]$.
\end{corollary}
\begin{figure}
\centering
\begin{tikzpicture}[thick, scale=1]
\begin{axis}[axis on top, xticklabels={$\eps_{\mathsf{L}}$, $\bar{\alpha}$},ymin=0.7,ymax=1,
  samples=50, y=16.5cm, x label style={at={(axis description cs:0.5,0)},rotate=0,anchor=south},
    x=29.5cm, xlabel= $\eps$,
  no markers]
\addplot+[domain=0.6:0.8][name path=A, no marks,blue,dotted, line width=1pt,samples=600] {1.4*x-0.12};
\addplot+[domain=0.6:0.8][name path=B, no marks,red, dashed, line width=1pt,samples=600] {sqrt(1.4*x^2+0.104)};
\addplot[domain=0.6:0.8][name path=B2, no marks,green,line width=1pt,samples=600] {4.67162*x^(10)+0.498388)^(0.1)};
\end{axis}
\node[below,blue] at (0.65, -0.1) {\footnotesize{$\eps_{\mathsf{L}}$}};
\node[below,blue] at (6.5, -0.1) {\footnotesize{$\bar{\alpha}$}};
\end{tikzpicture}
\caption{The graphs of $\underline{\mathcalboondox{h}}_{10}(\eps)$ (green solid curve), $\underline{\mathcalboondox{h}}_{2}(\eps)$ (red dashed curve), and $\mathcalboondox{h}_2^{\mathsf{i}}(\eps)=\mathcalboondox{h}_{10}^{\mathsf{i}}(\eps)$ (blue dotted line) given in Proposition~\ref{Propo_h_iid} and Theorem~\ref{Prop:IIDDataUnderbar} for i.i.d.\ $(X^n, Y^n)$ with $X\sim\sBer(0.6)$ and $P_{Y|X}=\mathsf{BSC}(0.2)$.}\label{fig:1}
\end{figure}  

Note that $\Phi(n)\downarrow0$ as $n\to\infty$. Thus \eqref{LB_Gap} implies that, as expected, the gap between the performance of the optimal privacy filter and that of  the optimal memoryless privacy filter increases as $n$ increases. This observation is numerically illustrated in Fig.~\ref{fig:1}, where $\underline{\mathcalboondox{h}}_n(\eps)$ is plotted as a function of $\eps$ for $n=2$ and $n=10$.
Moreover, \eqref{UB_Gap} implies that when $p=\frac{1}{2}$ and $\alpha$ is small, $\underline{\mathcalboondox{h}}_n(\eps)$ can be approximated by $\mathcalboondox{h}^\mathsf{i}_n(\eps)$. Thus, we can approximate the optimal filter $\mathsf{Z}_n(\zeta_n(\eps))$ with a simple memoryless filter given by $Z_k=Y_k\oplus W_k$, where $W_1, \dots, W_n$ are i.i.d.\ $\sBer(0.5\zeta(\eps))$ random variables that are independent of $(X^n, Y^n)$.

\subsection{Markov Private Data}
\label{Section:Memory}

In this section, we assume that $X^n$ comprises the first $n$ samples of a homogeneous first-order Markov process  having a symmetric transition matrix; i.e.,  $(X^n, Y^n)$ satisfy ($\textnormal{a}_2$) and ($\textnormal{b}$). In practice, this may account for data that follows a pattern, such as a password.

It is easy to see that under assumptions  ($\textnormal{a}_2$) and ($\textnormal{b}$),
\eq{\Pr(X^n=x^n) = \bar{p}\bar{r}^{n-1} \left(\frac{p}{\bar{p}}\right)^{x_1} \prod_{k=2}^n \left(\frac{r}{\bar{r}}\right)^{x_k\oplus x_{k-1}}.}
In particular, if $r<\frac{1}{2}\leq p$, then a direct computation shows that $\cP(X^n)=p\bar{r}^{n-1}.$
The values of $\cP(X^n|Y^n)$ for odd and even $n$ are slightly different. For simplicity, in what follows we assume that $n$ is odd. In this case, as shown in equation \eqref{End_Point_Epsilon_Memory} in Appendix~\ref{Appendix_Proof_Memory},
\begin{equation}\label{PcXnYn_Markov}
\cP(X^n|Y^n)= \bar{\alpha}^n \bar{r}^{n-1} \sum_{k=0}^{(n-1)/2} \binom{n}{k} \left(\frac{\alpha}{\bar{\alpha}}\right)^k.
\end{equation}
Theorem~\ref{Thm:GeneralizedLocalLinearity} established the optimality of a Z-channel $\mathsf{Z}^{y_0,z_0}$ for some $y_0, z_0\in \{0,1\}^n$. In order to find a closed form expression for $\underline{\mathcalboondox{h}}_n$, it is necessary to find $(y_0,z_0)$ which in principle depends on the parameters $(p,\alpha,r)$. The following theorem, whose proof is given in Appendix~\ref{Appendix_Proof_Memory}, bounds $\underline{\mathcalboondox{h}}_n$ for different values of $(p,\alpha,r)$.

\begin{theorem}
\label{Theorem_MarkovMemory}
Assume that $n\in\mathbb{N}$ is odd and $(X^n,Y^n)$ satisfy ($\textnormal{a}_2$) and ($\textnormal{b}$) with $p\in[\frac{1}{2},1)$, $\alpha\in(0,\frac{1}{2})$, and $\bar{\alpha}\bar{p}>\alpha p$. If \dsty{\frac{r}{\bar{r}} < \left(\frac{\alpha}{\bar{\alpha}}\right)^{n-1}}, then there exists $\eps_\mathsf{L}<\cP(X^n|Y^n)$ such that
\eq{1 - \zeta_n(\eps) \Pr(Y^n={\bf 1}) \leq \underline{\mathcalboondox{h}}_n^n(\eps) \leq 1 - \zeta_n(\eps)\alpha^n,}
for every $\eps\in[\eps_\mathsf{L},\cP(X^n|Y^n)]$, where
\eq{\zeta_n(\eps) \coloneqq \bar{r}\frac{\cP(X^n|Y^n)-\eps^n}{p(\bar{\alpha}\bar{r})^n-\bar{p}(\alpha\bar{r})^n}.}
Furthermore, the $2^n$-ary Z-channel $\mathsf{Z}_n(\zeta_n(\eps))$ achieves the lower bound in this interval.
\end{theorem}

The special case of $r=0$ is of particular interest. Note that when $r=0$, then ($\textnormal{a}_2$) corresponds to $X_1=\dots=X_n=\theta\in \{0,1\}$. Here, $Y^n\in \{0, 1\}^n$ are i.i.d. copies drawn from $P_{Y|\theta}=\sBer(\bar{\alpha}^{\theta}\alpha^{\bar{\theta}})$. The prior distribution of the parameter $\theta$ is $\sBer(p)$. The parameter $\theta$ is considered to be private and $Y^n$ must be guessed as accurately as possible. This problem can be viewed as a reverse version of \textit{privacy-aware learning} studied in \cite{privacyaware}. The following proposition, whose proof is given in Appendix~\ref{Appendix:Parametric}, provides a closed form expression for $\underline{\mathcalboondox{h}}_n$ in the low privacy regime.  Note that in this case, $\cP(\theta)=p$ and the value of $\cP(\theta|Y^n)$ is obtained from \eqref{PcXnYn_Markov} by setting $r=0$. 
\begin{proposition}
\label{Proposition_Parametric_Dis_Privacy}
Assume that $n$ is odd. Let $\theta\sim \sBer(p)$ with $p\in[\frac{1}{2},1)$ and $Y^n$ be $n$ i.i.d.  $\sBer(\bar{\alpha}^\theta\alpha^{\bar{\theta}})$ samples with $\alpha\in(0,\frac{1}{2})$, $\bar{\alpha}\bar{p}>\alpha p$ and $p<\cP(\theta|Y^n)$. Then, there exists $\eps_\mathsf{L}<\cP(\theta|Y^n)$ such that
\eq{ \max_{P_{Z^n|Y^n}: \Z^n=\{0,1\}^n,\atop \cP(\theta|Z^n)\leq \eps^n} \cP(Y^n|Z^n)= 1 - \zeta_n(\eps)(p\bar{\alpha}^n+\bar{p}\alpha^n),}
for every $\eps\in[\eps_\mathsf{L},\cP(\theta|Y^n)]$ where
\eq{\zeta_n(\eps) = \frac{\cP(\theta|Y^n)-\eps^n}{p\bar{\alpha}^n-\bar{p}\alpha^n}.}
Moreover, the $2^n$-ary Z-channel $\mathsf{Z}_n(\zeta_n(\eps))$ achieves $\underline{\mathcalboondox{h}}_n(\eps)$ in this interval.
\end{proposition}

\section{Continuous Case}
\label{Sec:ContinuousCase}

In this section, we assume that $X$ and $Y$ are real-valued random variables having a joint density $P_{XY}$ and the filter $P_{Z|Y}$ is realized by an independent additive Gaussian noise random variable. In particular, the privacy filter's output is 
$$Z_{\gamma}=\sqrt{\gamma}Y+N_{\mathsf{G}},$$
for some $\gamma\geq 0$, where  $N_{\mathsf{G}}\sim \mathcal{N}(0,1)$ is independent of $(X,Y)$. The choice of additive Guassian mechanisms is due to their implementation simplicity and mathematical tractability. Nonetheless, additive non-Gaussian and more general non-linear mechanisms might be natural in specific applications; their investigation is left as a future work. The goal of this section is to study $\sM$, defined in Definition~\ref{Def_StrongENSR}. To make the notation simpler, we define the following.
\begin{definition}
	\label{Def:strong_estimation_privacy}
	Given a pair of absolutely continuous random variables $(X, Y)$ with distribution $P_{XY}$ and $\eps\geq0$, we say that $Z_\gamma$ satisfies \emph{$\eps$-strong estimation privacy}, denoted as $Z_\gamma\in \Gamma(P_{XY}, \eps)$, if 
	\begin{equation}
	\label{Def:Strong_Privacy}
	1-\eps\leq\frac{\mmse(f(X)|Z_\gamma)}{\var(f(X))}\leq 1,
	\end{equation}	
	holds for every measurable function $f:\R\to\R$ with $0<\var(f(X))<\infty$. Similarly, $Z_\gamma$ is said to satisfy \emph{$\eps$-weak estimation privacy}, denoted by $Z_\gamma\in \partial\Gamma(P_{XY}, \eps)$, if \eqref{Def:Strong_Privacy} holds for identity function, i.e., $f(x)=x$.
\end{definition}
Similar to privacy, the utility between $Y$ and $Z_\gamma$ will be measured in terms of $\mmse(Y|Z_\gamma)$, 
and hence 
$\sM$ (Definition~\ref{Def_StrongENSR}) quantifies the tradeoff between utility and privacy.  In fact, $\sM$ can be equivalently written as 
$$\sM(P_{XY}, \eps)=\inf_{\gamma\geq 0:Z_\gamma\in \Gamma(P_{XY}, \eps)}\frac{\mmse(Y|Z_\gamma)}{\var(Y)}.$$
We can analogously define the \emph{weak estimation noise-to-signal ratio} as
$$\sW(P_{XY}, \eps)\coloneqq\inf_{\gamma\geq 0: Z_\gamma\in \partial\Gamma(P_{XY}, \eps)}\frac{\mmse(Y|Z_\gamma)}{\var(Y)}.$$
Note that $\sM$ and $\sW$ are non-increasing since $\Gamma(P_{XY},\eps)\subseteq \Gamma(P_{XY},\eps')$ and $\partial\Gamma(P_{XY},\eps)\subseteq \partial\Gamma(P_{XY},\eps')$ if $\eps\leq \eps'$. For the sake of brevity, we omit $P_{XY}$ in $\Gamma(P_{XY}, \eps)$, $\partial\Gamma(P_{XY}, \eps)$, $\sM(P_{XY}, \eps)$, and $\sW(P_{}, \eps)$ when there is no risk of confusion.

In what follows we derive equivalent conditions for $Z_\gamma\in\Gamma(\eps)$ and $Z_\gamma\in\partial\Gamma(\eps)$, respectively. Recall that the (Pearson) correlation coefficient of the random variables $U$ and $V$ is defined as
\eq{\rho(U,V) = \frac{\cov(U,V)}{\sqrt{\var(U)\var(V)}}}
provided that $0<\var(U),\var(V)<\infty$. For a random variable $U$, let $\mathcal{S}_U$ be the set of all measurable functions $f:\R\to\R$ such that $0<\var(f(U))<\infty$. Consider the following.

\begin{definition}[\hspace{-0.007cm}\cite{sarmanov, Renyi-dependence-measure}]
\label{Definition-Maximal_corr}
Let $U$ and $V$ be a pair of random variables.
\begin{itemize}
	\item[i)] The \emph{maximal  correlation} of $U$ and $V$, denoted by $\rho_m(U,V)$, is defined as
\begin{eqnarray*}
\rho_m(U,V)&\coloneqq& \sup_{(f,g)\in \mathcal{S}_U\times\mathcal{S}_V} \rho(f(U),g(V)),
\end{eqnarray*}
	provided that $0<\var(U), \var(V)<\infty$. If either $\mathcal{S}_U\times \mathcal{S}_V$ is empty (which happens precisely when either $U$ or $V$ is constant almost surely), then we set $\rho_m(U, V)=0$.

	\item[ii)] The \emph{one-sided maximal correlation}\footnote{This name is taken from  \cite[Def. 7.4]{Calmon_PhD_thesis}. Originally, R\'enyi named this quantity as the "correlation ratio" of $U$ on $V$ \cite[eq.\ (6)]{Renyi-dependence-measure}. } between $U$ and $V$, denoted by $\eta_V(U)$, is defined as 
$$\eta_V(U) \coloneqq \sup_{g\in\mathcal{S}_V} \rho(U, g(V)),$$
provided that $0<\var(U)<\infty$. If $\mathcal{S}_V$ is empty, then we set $\eta_V(U)=0$.
\end{itemize}
\end{definition}
R\'enyi \cite{Renyi-dependence-measure} showed that $\eta^2_V(U)=\frac{\var(\E[U|V])}{\var(U)}$. Therefore, the law of total variance implies
\begin{equation}
\label{Eq:Eta_MMSE}
\frac{\mmse(U|V)}{\var(U)} = \frac{\mathbb{E}[\var(U|V)]}{\var(U)} = 1 - \frac{\var(\mathbb{E}[U|V])}{\var(U)}= 1-\eta^2_V(U).
\end{equation}
It can also be shown that $0\leq\rho_m(U, V)\leq 1$, where the lower bound is achieved if and only if $U$ and $V$ are independent, and the upper bound is achieved if and only if there exists a pair of functions $(f, g)\in\S_U\times\S_V$ such that $f(U)=g(V)$ almost surely \cite{Renyi-dependence-measure}. It is well known that if $(X_\sG, Y_\sG)$ is a pair of jointly Gaussian random variables with correlation coefficient $\rho$, then $\rho_m^2(X_\sG, Y_\sG)=\rho^2(X_\sG, Y_\sG)$, see \cite{gebelien} or \cite{Gaussian_Maximal_Correlation} for a more recent proof. R\'{e}nyi \cite{Renyi-dependence-measure} derived an equivalent characterization of maximal correlation as
\begin{equation}
\label{Maximal_correlation_Equivalent}
\rho^2_m(U;V) = \sup_{f\in\S_U} \eta_V^2(f(U)).
\end{equation}
The following theorem, whose proof is given in Appendix~\ref{Appendidx:MC}, provides an equivalent characterization of $\eps$-strong estimation privacy $Z_\gamma\in\Gamma(\eps)$.
\begin{theorem}\label{Theorem_equivakebt_strong_MC}
Let $U$ and $V$ be non-degenerate random variables and $\eps\in[0,1]$. Then
$$\mmse(f(U)|V)\geq (1-\eps)\var(f(U)),$$
for all $f\in\S_U$ if and only if $\rho_m^2(U,V)\leq \eps$. In particular, $Z_\gamma\in \Gamma(\eps)$ if and only if $\rho_m^2(X,Z_\gamma)\leq \eps$.
\end{theorem}

From this theorem and  \eqref{Eq:Eta_MMSE}, we can equivalently express $\sM(\eps)$ and $\sW(\eps)$ as
\al{
\sM(\eps) &= 1-\sup_{\gamma\geq 0:~\rho_m^2(X,Z_\gamma)\leq \eps} \eta^2_{Z_{\gamma}}(Y),\\
\sW(\eps) &= 1-\sup_{\gamma\geq 0:~\eta^2_{Z_{\gamma}}(X)\leq\eps} \eta^2_{Z_{\gamma}}(Y).
}

It is known that both $\eta$ and $\rho_m$ satisfy the data processing inequality (see e.g., \cite{Calmon_bounds_Inference} and \cite{Ulukus_Data_processing}) and hence $\eta_{Z_\gamma}(X)\leq \eta_Y(X)$ and  $\rho_m(X, Z_\gamma)\leq \rho_m(X,Y)$. Therefore, we can restrict $\eps$ in the definition of $\sW(\eps)$ and $\sM(\eps)$ to the intervals $[0, \eta^2_Y(X)]$ and $[0, \rho_m^2(X,Y)]$, respectively. Unlike the discrete case, it is clear that perfect privacy $\eps=0$ implies $\gamma=0$.  Thus perfect privacy yields trivial utility;  i.e., $\sM(0)=1$ and $\sW(0)=1$. 

Note that  $\gamma\mapsto \mmse(Y|Z_\gamma)$ is continuous and decreasing on $(0, \infty)$ \cite{MMSE_Guo} and $\gamma\mapsto \rho_m^2(X, Z_\gamma)$ is left-continuous and increasing on $(0, \infty)$ \cite[Theorem 2]{dembo1}. Thus we can define  $\gamma^\ast_\eps\coloneqq \max\{\gamma\geq0: \rho_m^2(X, Z_\gamma)\leq\eps\}$ for which we have \dsty{\sM(\eps)=\frac{\mmse(Y|Z_{\gamma^\ast_\eps})}{\var(Y)}}. The left-continuity of $\gamma\mapsto \rho_m^2(X, Z_\gamma)$ implies that  $\eps\mapsto \gamma^\ast_\eps$ is right-continuous, and thus $\eps\mapsto \sM(\eps)$ is right-continuous on $(0, \rho_m^2(X,Y))$.

\begin{example}
\label{example_Gaussian}
Let $(X_{\sG},Y_{\sG})$ be jointly Gaussian random variables with mean zero and correlation coefficient $\rho$ and let $Z_\gamma=\sqrt{\gamma}Y_\sG+N_{\sG}$. Since $\rho^2_m(X_{\sG},Z_{\gamma})=\rho^2(X_{\sG}, Z_{\gamma})$, we have that
\eq{\rho_m^2(X_{\sG}, Z_{\gamma})=\rho^2\frac{\gamma\var(Y_\sG)}{1+\gamma\var(Y_\sG)},}
and hence the mapping $\gamma\mapsto \rho_m^2(X_{\sG}, Z_{\gamma})$ is strictly increasing. As a consequence, for $0\leq\eps\leq\rho^2$, the equation $\rho_m^2(X_{\sG}, Z_{\gamma})=\eps$ has a unique solution
$$\gamma_{\eps}\coloneqq\frac{\eps}{\var(Y_\sG)(\rho^2-\eps)},$$
and $\rho_m^2(X_\sG, Z_\gamma)\leq\eps$ if and only if $\gamma\leq\gamma_{\eps}$. On the other hand,
$$\mmse(Y_{\sG}|Z_{\gamma}) = \frac{\var(Y_\sG)}{1+\gamma\var(Y_\sG)},$$
which shows that the map $\gamma\mapsto \mmse(Y_{\sG}|Z_{\gamma})$ is strictly decreasing. Therefore,
\begin{equation}\label{M_eps_gaussian}
  \sM({\eps})=\frac{\mmse(Y_{\sG}|Z_{\gamma_{\eps}})}{\var(Y_{\sG})}=1-\frac{\eps}{\rho^2}.
\end{equation}
Clearly, for jointly Gaussian $X_\sG$ and  $Y_\sG$, we have $\eta^2_{Z_{\gamma}}(X_{\sG})=\rho_m^2(X_{\sG},Z_{\gamma})$ for any $\gamma\geq 0$. Consequently, $\Gamma(\eps)=\partial\Gamma(\eps)$ and, for $0\leq \eps\leq \rho^2$,
\begin{equation}\label{Eq:SM_SW_Gaussian}
  \sM(\eps)=\sW(\eps)=1-\frac{\eps}{\rho^2}.
\end{equation}
\end{example}

Next, we obtain bounds on $\sM(\eps)$ for the special case of Gaussian non-private data $Y_\sG$. The proof of the following result is given in Appendix~\ref{Appendix_Bound}.

\begin{theorem}\label{Lemma_Gaussian_Y_Arbit_X}
Let $X$ be jointly distributed with Gaussian $Y_{\sG}$. Then,
$$1-\frac{\eps}{\rho^2(X, Y_{\sG})}\leq\sM(P_{XY_\sG},\eps)\leq 1-\frac{\eps}{\rho_m^2(X, Y_{\sG})},$$
\end{theorem}

Combined with \eqref{Eq:SM_SW_Gaussian}, this theorem shows that for a Gaussian $Y$, a Gaussian $X_\sG$ minimizes $\sM(\eps)$ among all continuous random variables $X$ having identical $\rho(X,Y_\sG)$ and maximizes $\sM(\eps)$ among all continuous random variables $X$ having identical $\rho_m(X,Y_\sG)$. These observations establish another extremal property of Gaussian distribution over AWGN channels, see e.g., \cite[Theorem 12]{MMSE_WU_Properties} for another example.
This theorem also implies that
\begin{eqnarray*}
\sM(P_{X_\sG Y_\sG},\eps)-\sM(P_{XY_\sG},\eps)&\leq& \frac{\eps}{\rho^2(X, Y_{\sG})}\\
&&-\frac{\eps}{\rho^2_m(X, Y_{\sG})},
\end{eqnarray*}
for Gaussian $X_{\sG}$ which satisfies $\rho_m^2(X_{\sG}, Y_{\sG})=\rho_m^2(X, Y_{\sG})$. This demonstrates that if the difference $\rho_m^2(X, Y_{\sG})-\rho^2(X, Y_{\sG})$ is small, then $\sM(P_{XY_\sG},\eps)$ is very close to $\sM(P_{X_\sG Y_\sG},\eps)$.

As stated before, for any given joint density $P_{XY}$, perfect privacy results in trivial utility, i.e., $\sM(0)=1$. Therefore, it is interesting to study the approximation of $\sM(\eps)$ for sufficiently small $\eps$, i.e., in the almost perfect privacy regime. The next result, whose proof is given in Appendix~\ref{Appendix:limsup}, provides such an approximation and also shows that the lower bound in Theorem~\ref{Lemma_Gaussian_Y_Arbit_X} holds for general $Y$ for $\eps$ in the almost perfect privacy regime. 
\begin{lemma}\label{Lemma_Approx_S}
	We have that
	$$\limsup_{\eps\to 0}\frac{1-\sM(\eps)}{\eps}\leq \frac{1}{\rho^2(X,Y)}.$$

\end{lemma} 

\section{Conclusion}
\label{Section:Conclusion}

We studied the problem of displaying $Y$ under a privacy constraint with respect to another correlated random variable $X$, where utility and privacy are measured in terms of the probability of correctly guessing and minimum mean-squared error in the discrete and continuous cases, respectively. 


In the discrete case, we introduced the privacy-constrained guessing function $\mathcalboondox{h}$ to quantify the fundamental tradeoff between privacy and utility. We proved that $\mathcalboondox{h}$ is piecewise linear for every $X$ and $Y$. When $X$ and $Y$ are binary, this result allowed us to obtain $\mathcalboondox{h}$ in closed form and to establish the optimility of the $Z$-channel. We then defined $\underline{\mathcalboondox{h}}$ analogously to $\mathcalboondox{h}$ with the additional assumption that $Z$ is supported over the alphabet of $Y$, thereby providing a lower bound for $\mathcalboondox{h}$. For arbitrary $X$ and $Y$, we derived $\underline{\mathcalboondox{h}}$ in closed form in the high utility regime and established the optimality of a generalized $Z$-channel in this regime. Finally, we specialized our results about $\underline{\mathcalboondox{h}}$ to the vector case, where $X^n$, $Y^n$, and $Z^n$ are assumed to be binary random vectors. Overall, these results provide tangible answers for the estimation theoretic privacy-utility tradeoff problem and the performance of $Z$-channels in the high utility regime.

In the continuous case, we proposed the estimation-noise-to-signal ratio function $\sM$ to capture the fundamental privacy-utility tradeoff with an intrinsic operational meaning. In the special case of additive Gaussian privacy filters, we showed that if $Y$ is Gaussian, then a Gaussian $X$ minimizes $\sM$ among all $(X, Y)$ with identical correlation coefficients and maximizes $\sM$ among all $(X, Y)$ with identical maximal correlations. We also obtained a tight lower bound for $\sM$ for general absolutely continuous random variables when $\eps$ is sufficiently small.


\appendices

\section{Proof of Theorem \ref{Thm:PiecewiseLinearity}}
\label{Appendix:ProofPiecewiseLinearity}

Before proving Theorem \ref{Thm:PiecewiseLinearity}, we need to establish some technical facts.

Consider the map $\H:\F\to[0,1]\times[0,1]$ given by
\eq{\H(F) = (\P(F),\U(F)),}
with $\P(F)$ and $\U(F)$ defined in \eqref{eq:DefPU}. For ease of notation, let
$\D = \left\{D\in\M_{N \times( N+1)}: \|D\|=1\right\}$ where $||\cdot||$ denotes the Euclidean norm in $\M_{N\times (N+1)} \equiv \R^{N(N+1)}$. For $G\in\F$, let
\eq{\D(G) = \left\{D\in \D: G+tD\in \F \textnormal{ for some } t>0\right\}.}
In graphical terms, $\D$ is the set of all possible directions in $\M_{N \times (N+1)}$ and $\D(G)$ is the set of directions that make $t\mapsto G+tD$ ($t\geq0$) stay locally in $\F$. 
\begin{lemma}
	\label{Lemma:CompactnessDirections}
	For every $G\in\F$, the set $\D(G)$ is compact.
\end{lemma}

\begin{proof}
	Let $A=\{(y,z): G_{y,z}=0\}$ and $B=\{(y,z): G_{y,z}=1\}$. It is straightforward to verify that
	\eq{\D(G) = \A\cap\B\cap\C\cap\D,}
	where
	\al{
		\A &= \bigcap_{(y,z)\in A} \left\{D\in\M_{N,(N+1)}: D_{y,z}\geq0\right\},\\
		\B &= \bigcap_{(y,z)\in B} \left\{D\in\M_{N,(N+1)}: D_{y,z}\leq0\right\},\\
		\C &= \left\{D\in\M_{N,(N+1)}: \sum_{z=1}^{N+1} D_{y,z} = 0, ~y=\repn{N}\right\}.
	}
	Observe that since  sets $\A$, $\B$, $\C$ and $\D$ are closed, so is $\D(G)$. Since $\D$ is bounded, we have that $\D(G)$ is bounded as well. In particular, $\D(G)$ is closed and bounded and thus compact.
\end{proof}

\begin{lemma}
\label{Lemma:ContinuityTao}
Let $G\in\F$ be given and define $\tau:\D(G)\to\R$ by
\eq{\tau(D) \coloneqq \sup \{t\geq0 \mid G+tD\in\F\}.}
The function $\tau$ is continuous on $\D(G)$.
\end{lemma}

\begin{proof}
Let $\text{ri}(\F)$ and $\text{rb}(\F)$ denote the relative interior and relative boundary of $\F$, respectively. In what follows, we assume that $G\in\text{rb}(\F)$. The proof for $G\in\text{ri}(\F)$ follows the same steps and the details are left to the reader. The proof of the lemma is by contradiction.  Assume that there exists a sequence $(D_n)_{n\geq0}\subset\D(G)$ such that $D_n \to D_0$ but $\tau(D_n) \not\to \tau(D_0)$ as $n\to\infty$. Since $\F$ is bounded, the sequence $(\tau(D_n))_{n\geq1}$ is necessarily bounded. Therefore, there must exist a subsequence $(D_{n_k})_{k\geq1}$ such that
\begin{equation}
\label{eq:ProofContinuityTauConvergence}
\lim_{k\to\infty} \tau(D_{n_k}) = r \neq \tau(D_0).
\end{equation}
By the maximality of $\tau(D)$, we have that $G+\tau(D)D\in\text{rb}(\F)$ for all $D\in\D(G)$. Notice that $\F$ is a convex polytope defined by the intersection of finitely many hyperplanes. In particular, $G+\tau(D)D$ belongs to one of the supporting hyperplanes of $\F$. Furthermore, the maximality of $\tau(D)$ can be used once again to show that $G+\tau(D)D$ belongs to a supporting hyperplane of $\F$ that does not contain $G$. Since there are finitely many supporting hyperplanes of $\F$, there exists a further subsequence $(D_{n'_k})_{k\geq1}$ and a hyperplane $H$ such that $G+\tau(D_{n'_k})D_{n'_k}\in H$ for all $k\geq1$ and $G\notin H$. Since $H$ and $\F$ are closed sets, we conclude that
$$\lim_{k\to\infty} G+\tau(D_{n'_k})D_{n'_k} = G+rD_0 \in H \cap \F.$$
By the maximality of $\tau(D_0)$ and \eqref{eq:ProofContinuityTauConvergence}, we have $\tau(D_0) > r$. Since $H$ is a hyperplane and $G\notin H$, it is easy to verify that
\begin{equation}
\label{eq:ProofContinuitySidesH}
\{G+tD_0 : t\in[0,\tau(D_0)]\} \cap H = \{G+rD_0\}.
\end{equation}
In particular, \eqref{eq:ProofContinuitySidesH} implies that $G$ and $G+\tau(D_0)D_0$ are on opposite sides of $H$.  Since $G\in\F$ and $H$ is a supporting hyperplane of $\F$, we conclude that $G+\tau(D_0)D_0\notin\F$. This contradicts the fact that $G+\tau(D)D\in\text{rb}(\F)\subset\F$ for all $D\in\D(G)$.
\end{proof}

 The following lemma shows the local linear nature of the mapping $\H$. Let $[G_1,G_2]=\{\lambda G_1+ (1-\lambda) G_2 : \lambda\in[0,1]\}$.
\begin{lemma}
\label{Lemma:LocalLinearity}
For every $G\in\F$, there exists $\delta>0$ such that $F\mapsto \H(F)$ is linear on $[G,G+\delta D]$ for every $D\in\D(G)$.
\end{lemma}

\begin{proof}
Let $P=[P(x,y)]_{x\in\X,y\in\Y}$ be the joint probability matrix of $X$ and $Y$, and $Q$ the diagonal matrix with $q_1,\ldots,q_N$ as diagonal entries where $q_y=\Pr(Y=y)$ for $y\in \Y$. For $G\in\F$ let $\tau:\D(G)\to\R$ be as defined in Lemma~\ref{Lemma:ContinuityTao}. The definition of $\D(G)$ clearly implies that $\tau(D)>0$ for all $D\in\D(G)$. For $x\in\X$, $z\in\Z$, and $D\in\D(G)$, consider the function $f_{x,z}^{(D)}:\R\to\R$ given by
\begin{equation}
\label{Def_AZX_BXZ}
f_{x,z}^{(D)}(t) \coloneqq [PG](x,z) + t[PD](x,z),
\end{equation}
where $PG$ (resp., $PD$) is the product of matrices $P$ and $G$  (resp., $P$ and $D$). 
Note that \dsty{\P(G+tD) = \sum_{z\in\Z} \max_{x\in\X} f_{x,z}^{(D)}(t)} for all $t\in[0,\tau(D)]$ (see \eqref{eq:DefPU}). Let
\begin{equation}
\begin{aligned}
a_z &= \max_{x\in\X} [PG](x,z), \\
\M_z&=\{x\in\X : [PG](x,z) = a_z\}, ~\text{and}\\
b_z^{(D)}&=\max_{x\in\M_z} [PD](x,z).
\end{aligned}
\label{Eq:A_M_B}
\end{equation}
Let \dsty{t_{x,z}^{(D)} \coloneqq -\frac{a_z-[PG](x,z)}{b_z^{(D)}-[PD](x,z)}} whenever $[PD](x,z) \neq b_z^{(D)}$, and $t_{x,z}^{(D)}=\infty$ otherwise. Notice that $f_{x,z}^{(D)}(t_{x,z}^{(D)}) = a_z+t_{x,z}^{(D)}b_z^{(D)}$. Since $t_{x,z}^{(D)}\neq0$ for all $x\notin\M_z$,
\eq{t^{(D)} \coloneqq \min_{z\in\Z} \min_{x\notin\M_z} \min\{|t_{x,z}^{(D)}|,\tau(D)\} > 0.}
It is easy to see that \dsty{a_z + tb_z^{(D)} = \max_{x\in\X} f_{x,z}^{(D)}(t)} for all $t\in[0,t^{(D)}]$. In particular,
\begin{eqnarray}
\P(G+tD) &=&\sum_{z=1}^{N+1} \max_{x\in\X} f_{x,z}^{(D)}(t)=\sum_{z=1}^{N+1} a_z +t\sum_{z=1}^{N+1} b_z^{(D)} \nonumber\\
&=&\P(G)+tb^{(D)}, \label{eq:PrivacyMasterEquation}
\end{eqnarray}
for every $D\in\D(G)$ and $t\in[0,t^{(D)}]$, where $b^{(D)}\coloneqq \sum_{z=1}^{N+1} b_z^{(D)}$. Consequently, $\P$ is linear on $[G,G+t^{(D)} D]$.
By Lemma~\ref{Lemma:ContinuityTao}, $\tau:\D(G)\to\R$ is continuous and bounded. Hence, the map $D \mapsto \min\{|t_{x,z}^{(D)}|,\tau(D)\}$ ($x\notin\M_z$) is also continuous. In particular, the map $D\mapsto t^{(D)}$ is continuous. By compactness of $\D(G)$ established in Lemma~\ref{Lemma:CompactnessDirections}, we conclude that \dsty{\delta_\P \coloneqq \min_{D\in\D(G)} t^{(D)}>0}. Thus, $\P$ is linear on $[G,G+\delta_\P D]$ for every $D\in\D(G)$.

\begin{figure}[t!]
\centering
\begin{tikzpicture}[thick, scale=1]
\begin{axis}[xmin=0,xmax=1,ymin=0,ymax=2,xtick=\empty, ytick=\empty,  y label style={at={(axis description cs:0.08,0.55)},rotate=-90,anchor=south},
  samples=100, x label style={at={(axis description cs:0.95,0)},rotate=0,anchor=south},
  samples=100,
  grid=both,xlabel= $t$, ylabel=$f^{(D)}_{x, z}$,
  no markers]
\addplot[domain=0:0.65][no marks,line width=0.7pt,samples=600] {0.5*x+1};
\addplot[domain=0:0.65][no marks,line width=0.7pt,samples=600] {1.1*x+0.5};
\addplot[domain=0:0.65][no marks,line width=0.7pt,samples=600] {3*x+0.08};
\addplot[domain=0:0.65][no marks,line width=0.7pt,samples=600] {1.1*x+1};
\addplot[mark=none, dotted, line width=0.5pt] coordinates {(0.65,0) (0.65,2)};
\addplot[mark=none, dotted, line width=0.5pt] coordinates {(0.4842,0) (0.4842,2)};
\end{axis}
\node[below] at (4.4, 0) {\footnotesize{$\tau(D)$}};
\node[below] at (3.27, 0.06) {\footnotesize{$t^{(D)}_{1, z}$}};
\node[below] at (1.2, 1.65) {\footnotesize{$f_{1, z}^{(D)}$}};
\node[below] at (1.2, 2.62) {\footnotesize{$f_{2, z}^{(D)}$}};
\node[below] at (1.2, 3.2) {\footnotesize{$f_{3, z}^{(D)}$}};
\node[below] at (1.2, 4.06) {\footnotesize{$f_{4, z}^{(D)}$}};
\node[below] at (0,0) {\footnotesize{$0$}};
\end{tikzpicture}
\caption{Typical functions $f^{(D)}_{x, z}$ ($x=\{1, 2, 3, 4\}$) for a given $z\in\Z$ and $D\in\D(G)$. In this example, we have $\M_{z}=\{3, 4\}$ and $a_z+t b_z^{(D)} = f_{4,z}^{(D)}(t)$. Notice that $t_{2, z}^{(D)}=\infty$ and $t_{3, z}^{(D)}=t_{4, z}^{(D)}=0$.}
\label{Fig:Typical_function}
\end{figure}
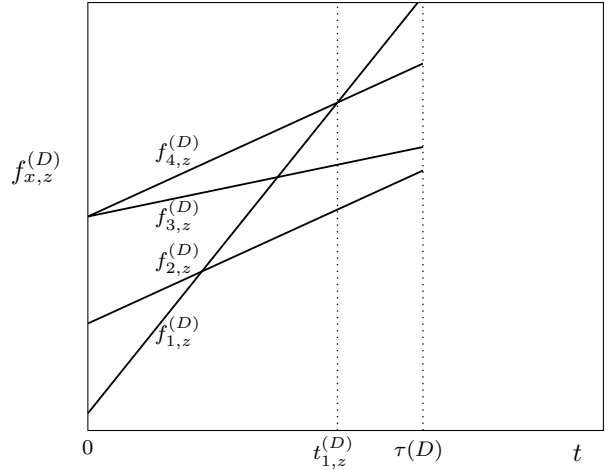

For $y\in\Y$, $z\in\Z$, and $D\in\D(G)$, consider the function $g_{y,z}^{(D)}:\R\to\R$ given by
\eq{g_{y,z}^{(D)}(t) = [QG](y,z) + t [QD](y,z).}
Observe that \dsty{\U(G+tD)=\sum_{z\in\Z} \max_{y\in\Y} g_{y,z}^{(D)}(t)} for all $t\in[0,\tau(D)]$ (see \eqref{eq:DefPU}). Similarly to \eqref{Eq:A_M_B}, let
\begin{equation*}
\begin{aligned}
\alpha_z &= \max_{y\in\Y} [QG](y,z), \\
\N_z&=\{y\in\Y : [QG](y,z) = \alpha_z\}, ~ \text{and}\\
\beta_z^{(D)}&=\max_{y\in\N_z} [QD](y,z).
\end{aligned}
\end{equation*}
Using a similar argument that resulted in \eqref{eq:PrivacyMasterEquation}, it can be shown that there exists $\delta_U>0$ such that
\begin{eqnarray}
\U(G+tD) &=& \sum_{z=1}^{N+1} g_{y_z,z}^{(D)}(t) = \sum_{z=1}^{N+1} \alpha_z  +t\sum_{z=1}^{N+1} \beta_{z}^{(D)}\nonumber\\
&=&\U(G)+t\beta^{(D)}, \label{eq:UtilityMasterEquation}
\end{eqnarray}
for every $D\in\D(G)$ and $t\in[0,\delta_\U]$,  where $\beta^{(D)}\coloneqq \sum_{z=1}^{N+1} \beta_{z}^{(D)}$. Consequently, $\U$ is linear on $[G,G+\delta_\U D]$ for every $D\in\D(G)$. Therefore, $F\mapsto\H(F)=(\P(F),\U(F))$ is linear on $[G,G+\delta D]$ for every $D\in\D(G)$, where $\delta = \min(\delta_\P,\delta_\U)$.
\end{proof}

We say that a filter $F\in\F$ is optimal if $\U(F) = \mathcalboondox{h}(\P(F))$. If $F$ is an optimal filter and $\P(F)=\eps$, we say that $F$ is optimal at $\eps$. The following result is a straightforward application of the concavity of $\mathcalboondox{h}$, and thus its proof is omitted.

\begin{lemma}
\label{Lemma:ThreeAligned}
For $G\in\F$, let $\delta>0$ be as in Lemma \ref{Lemma:LocalLinearity}. If there exist $D\in\D(G)$ and $0<t_1<t_2\leq\delta$ such that $G$, $G+t_1D$ and $G+t_2D$ are optimal filters, then $G+tD$ is an optimal filter for each $t\in[0,\delta]$.
\end{lemma}

A function $[\cP(X),\cP(X|Y)]\ni\eps \mapsto F_\eps\in\F$ is called a \textit{path} of optimal filters if $\P(F_\eps) = \eps$ and $\U(F_\eps)=\mathcalboondox{h}(\eps)$ for every $\eps\in[\cP(X),\cP(X|Y)]$. As mentioned in Section \ref{Subsection:GeometricalFormulation}, for every $\eps$ there exists $F_\eps$ such that $\P(F_\eps)=\eps$ and $\U(F_\eps)=\mathcalboondox{h}(\eps)$, i.e., a path of optimal filters always exists. In the rest of this section we establish the existence of a \textit{piecewise} linear path of optimal filters.
\begin{lemma}
\label{lemma:existence_Optimal_Filter_UP}
For every $\eps\in[\cP(X),\cP(X|Y))$, there exists $F_\eps\in\F$ and $D\in\D(F_\eps)$ such that $F_\eps$ is an optimal filter at $\eps$, $\P(F_\eps+\delta D)>\eps$, and $F_\eps+t D$ is an optimal filter for each $t\in[0,\delta]$ with $\delta>0$ as in Lemma \ref{Lemma:LocalLinearity} for $F_\eps$.
\end{lemma}

\begin{proof}
Let $K=2(\cP(X|Y)-\eps)^{-1}$. For every $n,m>K$, let $G_{n,m}$ be an optimal filter at $\eps+\frac{1}{n}+\frac{1}{m}$. For every $n>K$, the set $\{G_{n,m} : m>K\}$ is an infinite set. Since $\F$ is compact, $\{G_{n,m} : m>K\}$ has at least one accumulation point, say $G_n$. Let $(G_{n,m_k})_{k\geq1} \subset \{G_{n,m} : m>K\}$ be a subsequence with $\lim_k G_{n,m_k} = G_n$. By continuity of $\P$, $\U$, and $\mathcalboondox{h}$, we have that
\al{
\P(G_n) &= \lim_{k\to\infty} \P(G_{n,m_k}) = \eps+\frac{1}{n},\\
\U(G_n) &= \lim_{k\to\infty} \U(G_{n,m_k}) = \lim_{k\to\infty} \mathcalboondox{h}(\P(G_{n,m_k})) = \mathcalboondox{h}(\P(G_n)),
}
i.e., $G_n$ is an optimal filter at $\eps+\frac{1}{n}$. By the same arguments as before, the set $\left\{G_n : n>K\right\}$ has at least one accumulation point, say $F_\eps$, and this accumulation point is an optimal filter at $\eps$. Let $\delta>0$ be as in Lemma \ref{Lemma:LocalLinearity} for $F_\eps$. By construction of $F_\eps$, there exists $n_1>K$ such that $\|G_{n_1}-F_\eps\| < \frac{\delta}{2}$. The filter $G_{n_1}$ can be written as $G_{n_1}=F_\eps+t_1D_1$ with $t_1\in(0,\frac{\delta}{2})$ and $D_1\in\D(F_\eps)$. Recall that, by \eqref{eq:PrivacyMasterEquation} and \eqref{eq:UtilityMasterEquation}, for every $D\in\D(F_\eps)$ and $t\in[0,\delta]$,
\eq{\P(F_\eps+tD) = \eps + t b^{(D)} \quad \textnormal{ and } \quad \U(F_\eps+tD) = \mathcalboondox{h}(\eps) + t \beta^{(D)}.}
Notice that the maps $D \mapsto b^{(D)}$ and $D \mapsto \beta^{(D)}$ are continuous. Since $\P(G_{n_1}) = \eps +\frac{1}{n_1} > \eps$, we conclude that $b^{(D_1)}>0$ and, in particular, $\P(F_\eps+\delta D_1) > \eps$.

Let $(G_{n_1,m_k})_{k\geq1}\subset\{G_{n_1,m} : m>K\}$ be such that $\lim_k G_{n_1,m_k} = G_{n_1}$. For $k$ large enough, we can write $G_{n_1,m_k} = F_\eps + \theta_k E_k$ with $\theta_k\in[0,\delta]$ and $E_k\in\D(F_\eps)$. Since $\theta_k\to t_1$ and $E_k\to D_1$ as $k\to\infty$, there exists $n_2>K$ such that $\theta_{n_2}<\frac{\delta}{2}$ and $|b^{(E_{n_2})}-b^{(D_1)}|<\frac{b^{(D_1)}}{2}$. Let $t_2\coloneqq\theta_{n_2}$ and $D_2\coloneqq E_{n_2}$. Clearly, $t_2<\frac{\delta}{2}$ and $\frac{1}{2} b^{(D_1)} < b^{(D_2)} < 2b^{(D_1)}$. These inequalities yield $\P(F_\eps+\delta D_1) > \P(F_\eps+t_2D_2)$ and $\P(F_\eps+\delta D_2) > \P(F_\eps+t_1D_1)$. Thus, there exist $s_1,s_2\in[0,\delta]$ such that $\P(F_\eps+t_2D_2)=\P(F_\eps+s_1D_1)$ and $\P(F_\eps+t_1D_1)=\P(F_\eps+s_2D_2)$. In particular,
\eqn{eq:Swap}{\eps+t_2b^{(D_2)} = \eps+s_1b^{(D_1)} \quad \textnormal{ and } \quad \eps+t_1b^{(D_1)} = \eps+s_2b^{(D_2)}.}
By the optimality of $G_{n_1}=F_\eps+t_1D_1$ and $G_{n_1,m_{n_2}} = F_\eps+t_2D_2$,
\begin{eqnarray*}
\U(F_\eps+t_2D_2) &=& \mathcalboondox{h}(\eps) + t_2 \beta^{(D_2)} \\
&\geq& \mathcalboondox{h}(\eps) + s_1 \beta^{(D_1)} =\U(F_\eps+s_1D_1),
\end{eqnarray*}	
and 
\begin{eqnarray*}
\U(F_\eps+t_1D_1) &=& \mathcalboondox{h}(\eps) + t_1 \beta^{(D_1)} \\
&\geq& \mathcalboondox{h}(\eps) + s_2 \beta^{(D_2)} = \U(F_\eps+s_2D_2).
\end{eqnarray*}	
By the equations in \eqref{eq:Swap}, the above inequalities are in fact equalities. In particular, $F_\eps$, $F_\eps+t_1D_1$ and $F_\eps+s_1D_1$ are optimal filters. Invoking Lemma \ref{Lemma:ThreeAligned}, we conclude that $F_\eps+t D_1$ is an optimal filter for all $t\in[0,\delta]$.
\end{proof}

Using an analogous proof, we can also prove the following lemma.

\begin{lemma}
\label{lemma:existence_Optimal_Filter_Down}
For every $\eps\in(\cP(X),\cP(X|Y)]$, there exists $F_\eps\in\F$ and $D\in\D(F_\eps)$ such that $F_\eps$ is an optimal filter at $\eps$, $\P(F_\eps+\delta D)<\eps$, and $F_\eps+t D$ is an optimal filter for each $t\in[0,\delta]$ with $\delta>0$ as in Lemma \ref{Lemma:LocalLinearity} for $F_\eps$.
\end{lemma}

We are in position to prove Theorem \ref{Thm:PiecewiseLinearity}.

\begin{proof}[Proof of Theorem \ref{Thm:PiecewiseLinearity}]
For notational simplicity, we define $S\coloneqq\cP(X)$ and $T\coloneqq\cP(X|Y)$. In light of Lemmas~\ref{lemma:existence_Optimal_Filter_UP} and \ref{lemma:existence_Optimal_Filter_Down}, for every $\eps\in(S,T)$ there exist optimal filters $F_\eps$ and $G_\eps$ at $\eps$, $\delta_\eps>0$, $D_\eps\in\D(F_\eps)$, and $E_\eps\in\D(G_\eps)$ such that $F_\eps+t D_\eps$ and $G_\eps+t E_\eps$ are optimal filters for each $t\in[0,\delta_\eps]$, and $\P(G_\eps + \delta_\eps E_\eps)<\eps<\P(F_\eps+\delta_\eps D_\eps)$. Note that $\delta_\eps=\min\{\delta_{F_\eps},\delta_{G_\eps}\}$, where $\delta_{F_\eps}$ and $\delta_{G_\eps}$ are the constants obtained in Lemma~\ref{Lemma:LocalLinearity} for filters $F_\eps$ and $G_\eps$, respectively. For every $\eps\in(S,T)$, let $V_\eps = (\P(F_\eps+\delta_\eps E_\eps),\P(G_\eps+\delta_\eps D_\eps))$. Similarly, there exist
\begin{itemize}
	\item[a)] an optimal filter $F_S$ at $S$, $\delta_S>0$, and $D_S\in\D(F_S)$ such that $F_S+tD_S$ is an optimal filter for each $t\in[0,\delta_S]$ and $\P(F_S+\delta_S D_S)>S$;
	\item[b)] an optimal filter $G_T$ at $T$, $\delta_T>0$, and $E_T\in\D(G_T)$ such that $G_T+t E_T$ is an optimal filter for each $t\in[0,\delta_T]$ and $\P(G_T+\delta_T E_T) < T$.
\end{itemize}
Let $V_S = [S,\P(F_S+\delta_S D_S))$ and  $V_T = (\P(G_T+\delta_T E_T),T]$. The family $\left\{V_\eps: \eps\in[S,T] \right\}$ forms an open cover of $[S,T]$ (in the subspace topology). By compactness, there exist $S = \eps_0 < \cdots < \eps_l = T$ such that $\{V_{\eps_0},\ldots,V_{\eps_l}\}$ forms an open cover for $[S,T]$. For each $i\in\{0,\ldots,l-1\}$, the mapping
\eqn{eq:Paths}{[\eps_i,\P(F_{\eps_i}+\delta_{\eps_i} D_{\eps_i}))\ni \eps \mapsto F_{\eps_i} + \frac{\eps-\eps_i}{b^{(D_{\eps_i})}} D_{\eps_i} \in \F,}
is clearly linear. Similarly, for each $i\in\{\repn{l}\}$, the mapping
\eqn{eq:Paths2}{(\P(G_{\eps_i}+\delta_{\eps_i} E_{\eps_i}),\eps_i] \ni \eps \mapsto G_{\eps_i} + \frac{\eps-\eps_i}{b^{(E_{\eps_i})}} E_{\eps_i} \in \F,}
is also linear. Notice that \dsty{\P\left(F_{\eps_i} + \frac{\eps-\eps_i}{b^{(D_{\eps_i})}}D_{\eps_i}\right) = \eps = \P\left(G_{\eps_i} + \frac{\eps-\eps_i}{b^{(E_{\eps_i})}} E_{\eps_i}\right)}. Since $\{V_{\eps_0},\ldots,V_{\eps_l}\}$ forms an open cover for $[S,T]$, the mappings in  (\ref{eq:Paths}) and (\ref{eq:Paths2}) implement a piecewise linear path of optimal filters.
\end{proof}

The proof provided in this appendix establishes the existence of $\delta_*>0$, an optimal filter $F_*$ at $T\coloneqq \cP(X|Y)$, and $D_*\in\D(F_*)$ such that $\P(F_*+\delta_* D_*)<T$ (or equivalently $b^{(D_*)}<0$) and
\eq{\mathcalboondox{h}(\eps)=1+(\eps-T)\frac{\beta^{(D_*)}}{b^{(D_*)}},}
for every $\eps\in[T+\delta_* b^{(D_*)},T]$. This then implies that 
\eqn{eq:DirivativehMinimum}{\mathcalboondox{h}'(T) = \min_{F\in\F \atop \P(F)=T} \min_{D\in\D(F) \atop b^{(D)}<0} \frac{\beta^{(D)}}{b^{(D)}}.}

\section{Proof of Proposition~\ref{Proposition_Independent}}\label{Appenddix:Proposition}

Since $X$ is uniformly distributed in $\{1,\ldots,M\}$,
\begin{equation*}
-\log \cP(X) = \log M = H(X).
\end{equation*}
By the definition of $I_\infty(X;Z)$, we have that
\begin{eqnarray*}
	I_\infty(X;Z)&=&\log\left(\frac{\cP(X|Z)}{\cP(X)}\right)\\
	&=&H(X) + \log \left(\sum_{z\in \Z}P_Z(z)\max_{x\in \X}P_{X|Z}(x|z)\right)\\
	&\geq &H(X) + \sum_{z\in \Z}P_Z(z)\max_{x\in \X} \log P_{X|Z}(x|z),\\
\end{eqnarray*}
where the inequality follows from Jensen's inequality. Clearly, for each $z\in\Z$,
\begin{eqnarray*}
	\max_{x\in \X} \log P_{X|Z}(x|z) &\geq& \sum_{x\in \X}P_{X|Z}(x|z)\log P_{X|Z}(x|z)\\
	&=& - H(X|Z=z).
\end{eqnarray*}
Therefore,
\begin{equation*}
	I_\infty(X;Z) \geq H(X) - \sum_{z\in \Z} P_Z(z) H(X|Z=z) = I(X;Z).
\end{equation*}
Since $I_\infty(X;Z)=0$, we conclude that $I(X;Z)=0$ and thus $X\indep Z$.

\section{Proof of Theorem~\ref{Theorem_Linearity_BIBO}}
\label{Appendix_Linearity}

We first note that since $\mathcalboondox{h}$ is concave on $[\cP(X), \cP(X|Y)]$, its right derivative exists at $\eps=\cP(X|Y)$. Therefore, we have by concavity 
\eqn{eq:UpperBoundh}{\mathcalboondox{h}(\eps) \leq 1 - (\cP(X|Y)-\eps)\mathcalboondox{h}'(\cP(X|Y)),}
for all $\eps\in[p,\cP(X|Y)]$.
In Lemma~\ref{Lemma:Derivativeh} below, we show that
\begin{eqnarray*}
\mathcalboondox{h}'(\cP(X|Y)) &=& \frac{q}{\bar{\beta}p-\alpha\bar{p}} 1_{\{\alpha\bar{\alpha}\bar{p}^2<\beta\bar{\beta}p^2\}} \\
&&+\frac{\bar{q}}{\bar{\alpha}\bar{p}-\beta p}1_{\{\alpha\bar{\alpha}\bar{p}^2\geq\beta\bar{\beta}p^2\}}.
\end{eqnarray*}
Thus, \eqref{eq:UpperBoundh} becomes
\eqn{Cases:H}{\mathcalboondox{h}(\eps) \leq \begin{cases}1 - \zeta(\eps) q, & \alpha\bar{\alpha}\bar{p}^2<\beta\bar{\beta}p^2,\\1-\tilde{\zeta}(\eps)\bar{q}, & \alpha\bar{\alpha}\bar{p}^2\geq\beta\bar{\beta}p^2.\end{cases}}
To finish the proof of Theorem \ref{Theorem_Linearity_BIBO} we show that the Z-channel $\mathsf{Z}(\zeta(\eps))$ and the reverse Z-channel $\tilde{\mathsf{Z}}(\tilde{\zeta}(\eps))$ achieve \eqref{eq:UpperBoundh} and \eqref{Cases:H}, when $\alpha\bar{\alpha}\bar{p}^2<\beta\bar{\beta}p^2$ and  $\alpha\bar{\alpha}\bar{p}^2\geq\beta\bar{\beta}p^2$, respectively.

For $\alpha\bar{\alpha}\bar{p}^2<\beta\bar{\beta}p^2$, consider the filter $P_{Z|Y}=\left[\begin{matrix}1 & 0\\\zeta(\eps) & 1-\zeta(\eps)\end{matrix}\right]$. Notice that
\begin{equation}
\begin{aligned}
P_{XZ}&=\begin{bmatrix}\bar{p}(\bar{\alpha}+\alpha\zeta(\eps)) & \bar{p}\alpha(1-\zeta(\eps))\\p(\beta+\bar{\beta}\zeta(\eps)) & p\bar{\beta}(1-\zeta(\eps))\end{bmatrix}, ~\text{and}\\
P_{YZ}&=\begin{bmatrix}\bar{q} & 0\\q\zeta(\eps) & q(1-\zeta(\eps))\end{bmatrix}. \label{Z_Channel1}
\end{aligned}
\end{equation}

It is straightforward to verify that $\bar{p}(\bar{\alpha}+\alpha\zeta(\eps)) \geq p(\beta+\bar{\beta}\zeta(\eps))$. As a consequence, $\cP(X|Z)=\eps$. Since $\alpha\bar{\alpha}\bar{p}^2<\beta\bar{\beta}p^2$, we have that \dsty{\frac{\bar{q}}{q}>\zeta(\eps)}. Thus, $\cP(Y|Z)=1-\zeta(\eps)q$.

For $\alpha\bar{\alpha}\bar{p}^2\geq\beta\bar{\beta}p^2$, consider the filter $P_{Z|Y}=\left[\begin{matrix}1-\tilde{\zeta}(\eps) & \tilde{\zeta}(\eps)\\0 & 1\end{matrix}\right]$. Notice that
\begin{equation}
\begin{aligned}
P_{XZ}&=\begin{bmatrix}\bar{p}\bar{\alpha}(1-\tilde{\zeta}(\eps)) & \bar{p}(\alpha+\bar{\alpha}\tilde{\zeta}(\eps))\\p\beta(1-\tilde{\zeta}(\eps)) & p(\bar{\beta}+\beta\tilde{\zeta}(\eps))\end{bmatrix}, ~\text{and} \\
P_{YZ}&=\begin{bmatrix}\bar{q}(1-\tilde{\zeta}(\eps)) & \bar{q}\tilde{\zeta}(\eps)\\0 & q\end{bmatrix}.
\end{aligned}
\label{Z_channel2}
\end{equation}
Recall that $\bar{\alpha}\bar{p}>\beta p$ and also observe that  $p(\bar{\beta}+\beta\tilde{\zeta}(\eps)) \geq \bar{p}(\alpha+\bar{\alpha}\tilde{\zeta}(\eps))$. As a consequence, $\cP(X|Z)=\eps$. The fact that $\alpha\bar{\alpha}\bar{p}^2\geq\beta\bar{\beta}p^2$ implies  $q\geq\bar{q}\tilde{\zeta}(\eps)$. Therefore, $\cP(Y|Z)=1-\tilde{\zeta}(\eps)\bar{q}$.

\begin{lemma}
	\label{Lemma:Derivativeh}
	Let $X\sim\sBer(p)$ with $p\in[\frac{1}{2},1)$ and $P_{Y|X}\sim\mathsf{BIBO}(\alpha,\beta)$ with $\alpha,\beta\in[0,\frac{1}{2})$ such that $\bar{\alpha}\bar{p}>\beta p$. Then \dsty{\mathcalboondox{h}'(\cP(X|Y)) = \frac{q}{\bar{\beta}p-\alpha\bar{p}} 1_{\{\alpha\bar{\alpha}\bar{p}^2<\beta\bar{\beta}p^2\}} + \frac{\bar{q}}{\bar{\alpha}\bar{p}-\beta p}1_{\{\alpha\bar{\alpha}\bar{p}^2\geq\beta\bar{\beta}p^2\}}}.
\end{lemma}

\begin{proof}
	  As before, let $T\coloneqq\cP(X|Y)$.  We begin the proof by noticing that the Z-channels defined in \eqref{Z_Channel1} and \eqref{Z_channel2} provide a lower bound on $\mathcalboondox{h}(\eps)$ as follows:
	   \begin{equation}\label{eq:hLowerBound}
	   \mathcalboondox{h}(\eps)\geq 1-\zeta(\eps)q1_{\{\alpha\bar{\alpha}\bar{p}^2<\beta\bar{\beta}p^2\}}-\tilde{\zeta}(\eps)\bar{q}1_{\{\alpha\bar{\alpha}\bar{p}^2\geq\beta\bar{\beta}p^2\}}.
	   \end{equation}
		By concavity of $\mathcalboondox{h}$, this inequality implies
	\eq{\mathcalboondox{h}'(T) \leq \frac{q}{\bar{\beta}p-\alpha\bar{p}} 1_{\alpha\bar{\alpha}\bar{p}^2<\beta\bar{\beta}p^2} + \frac{\bar{q}}{\bar{\alpha}\bar{p}-\beta p}1_{\alpha\bar{\alpha}\bar{p}^2\geq\beta\bar{\beta}p^2}.}
	The rest of the proof is devoted to establishing the reverse inequality. To this end, we use the variational formula for $\mathcalboondox{h}'(T)$ given in \eqref{eq:DirivativehMinimum}.
		Let $P=[P(x,y)]_{x,y\in\{0,1\}}$ be the joint probability matrix of $X$ and $Y$.  Without loss of generality we can assume $\Z=\{z_1,z_2,z_3\}$. It follows from \eqref{eq:PrivacyMasterEquation} and \eqref{eq:UtilityMasterEquation} that for every $F\in\F\subset\M_{2\times 3}$ there exists $\delta>0$ such that
	\eqn{eq:DefbBeta}{\P(F+tD) = \P(F) +tb^{(D)},~\textnormal{ and }~ \U(F+tD) = \U(F) + t\beta^{(D)},}
	for every $t\in[0,\delta]$ and $D\in\D(F)$, where \dsty{b^{(D)}=\sum_{i=1}^3 \max_{x\in\M_{z_i}} [PD](x,z_i)} and \dsty{\beta^{(D)} = \sum_{i=1}^3 \max_{y\in\N_{z_i}} q(y)D(y,z_i)} with
	\al{
		\M_{z_i} &= \Big\{x\in\{0,1\}: (PF)(x,z_i) = \max_{x'\in\{0,1\}} (PF)(x',z_i)\Big\},\\
		\N_{z_i} &= \Big\{y\in\{0,1\}: q(y)F(y,z_i) = \max_{y'\in\{0,1\}} q(y')F(y',z_i)\Big\}.
	}
	
	
	Up to permutation of columns, which corresponds to permuting the elements of $\Z$, the set of filters $F\in\F$ such that $\P(F)=T$ equals
	\begin{equation}
	\begin{aligned}
	\left\{\begin{bmatrix} 1 & 0 & 0\\0 & u & v \end{bmatrix} : {0<v\leq u \atop u+v=1}\right\} &\bigcup \left\{\begin{bmatrix}0 & u & v\\1 & 0 & 0 \end{bmatrix} : {0<v\leq u \atop u+v=1}\right\}\\
	&\bigcup \left\{\begin{bmatrix}1 & 0 & 0\\0 & 1 & 0 \end{bmatrix}\right\}.
	\end{aligned}
	\label{FamilyFilters}
	\end{equation}
	To compute $\mathcalboondox{h}'(T)$ using formula  \eqref{eq:DirivativehMinimum}   
we need to compute $\beta^{(D)}$ and $b^{(D)}$ for each $D\in \D(F)$ with $F$ of the form described in \eqref{FamilyFilters}.

	Let \dsty{F=\begin{bmatrix}1 & 0 & 0\\0 & u & v\end{bmatrix}} for some $0<v\leq u$ and $u+v=1$. A direct computation shows that	
	\eqn{eq:PFuv}{PF = \begin{bmatrix} \bar{\alpha}\bar{p} & u\alpha\bar{p} & v\alpha\bar{p}\\ \beta p & u\bar{\beta}p & v\bar{\beta}p\end{bmatrix}.}
	In particular, $\M_{z_1}=\{0\}$, $\M_{z_2}=\{1\}$, and $\M_{z_3}=\{1\}$. For every $D\in\D(F)$, the matrix $PD$ is equal to
	\eq{\begin{bmatrix} \bar{\alpha}\bar{p}D_{11}+\alpha\bar{p}D_{21} & \bar{\alpha}\bar{p}D_{12}+\alpha\bar{p}D_{22} & \bar{\alpha}\bar{p}D_{13}+\alpha\bar{p}D_{23}\\ \beta p D_{11}+\bar{\beta} p D_{21} & \beta p D_{12}+\bar{\beta} p D_{22} & \beta p D_{13}+\bar{\beta} p D_{23}\end{bmatrix},}
	and hence $b^{(D)}=\bar{\alpha}\bar{p}D_{11}+\alpha\bar{p}D_{21}+\beta p D_{12}+\bar{\beta}p D_{22}+\beta p D_{13}+\bar{\beta} p D_{23}$. Notice that, for $1\leq i\leq 3$, we have that $D_{i1}+D_{i2}+D_{i3}=0$. In particular, $b^{(D)}=(\bar{\alpha}\bar{p}-\beta p) D_{11}+(\alpha\bar{p}-\bar{\beta} p) D_{21}$.
	Consider the matrices,
	\begin{equation*}
	\begin{bmatrix} \bar{q} & 0 \\ 0 & q \end{bmatrix} F = \begin{bmatrix} \bar{q} & 0 & 0 \\ 0 & qu & qv \end{bmatrix} ,
	\end{equation*}
and 
\begin{equation*}
\begin{bmatrix} \bar{q} & 0 \\ 0 & q \end{bmatrix} D = \begin{bmatrix} \bar{q}D_{11} & \bar{q}D_{12} & \bar{q}D_{13} \\ qD_{21} & qD_{22} & qD_{33} \end{bmatrix},
\end{equation*}	
	from which we obtain $\N_{z_1}=\{0\}$, $\N_{z_2}=\{1\}$, $\N_{z_3}=\{1\}$, and therefore, $\beta^{(D)}=\bar{q}D_{11}+qD_{22}+qD_{23}=\bar{q}D_{11}-qD_{21}.$
	In what follows we use the simple fact that \dsty{\frac{ax+y}{bx+y} \geq \min\left\{\frac{a}{b},1\right\}} for $a,b>0$  and $x,y\geq0$ with $x+y>0$. For notational simplicity, let $\eta\coloneqq \frac{\bar{q}}{q}$ and $\zeta\coloneqq \zeta(p)$, where $\zeta(\cdot)$ is defined in  \eqref{Def_Zeta_Zeta_tilde}.

	From the form of $F$, it is clear that $-D_{11}\geq0$ and $D_{21}\geq 0$. If $b^{(D)} < 0$, then $D_{11}$ and $D_{21}$ cannot be simultaneously zero, and hence
	\begin{eqnarray*}
	\frac{\beta^{(D)}}{b^{(D)}} &=& \frac{q}{\bar{\beta} p-\alpha\bar{p}} \frac{-\eta D_{11}+D_{21}}{-\zeta D_{11}+D_{21}} \\
	&\geq & \frac{q}{\bar{\beta} p-\alpha\bar{p}}\min\left\{\frac{\eta}{\zeta},1\right\}\\
	& =& \begin{cases}\frac{q}{\bar{\beta}p-\alpha\bar{p}}, & {\alpha\bar{\alpha}\bar{p}^2<\beta\bar{\beta}p^2},\\\frac{\bar{q}}{\bar{\alpha}\bar{p}-\beta p}, & {\alpha\bar{\alpha}\bar{p}^2\geq\beta\bar{\beta}p^2}.\end{cases}
	\end{eqnarray*}
	In particular, we obtain that
	\eqn{Filter1}{\min_{D\in\D(F) \atop b^{(D)}<0} \frac{\beta^{(D)}}{b^{(D)}} \geq \begin{cases}\frac{q}{\bar{\beta}p-\alpha\bar{p}}, & {\alpha\bar{\alpha}\bar{p}^2<\beta\bar{\beta}p^2},\\\frac{\bar{q}}{\bar{\alpha}\bar{p}-\beta p}, & {\alpha\bar{\alpha}\bar{p}^2\geq\beta\bar{\beta}p^2}.\end{cases}}
	The case \dsty{F=\begin{bmatrix}0 & u & v\\1 & 0 & 0\end{bmatrix}} for $0<v\leq u$ and $u+v=1$ is analogous.
	
	Now, let \dsty{F=\begin{bmatrix}1 & 0 & 0\\0 & 1 & 0\end{bmatrix}}. By \eqref{eq:PFuv} with $u=1$ and $v=0$, we obtain that $\M_{z_1}=\{0\}$, $\M_{z_2}=\{1\}$, and $\M_{z_3}=\{0, 1\}$. In a similar way, $\N_{z_1}=\{0\}$, $\N_{z_2}=\{1\}$, and $\N_{z_3}=\{0, 1\}$. Hence
	\begin{eqnarray*}
		b^{(D)} &= & \bar{\alpha}\bar{p}D_{11}+\alpha\bar{p}D_{21}+\beta p D_{12}+\bar{\beta} p D_{22}\\
		&&+\max\{\bar{\alpha}\bar{p}D_{13}+\alpha\bar{p}D_{23},\beta p D_{13}+\bar{\beta} p D_{23}\},\\
		\beta^{(D)} &=& \bar{q}D_{11}+qD_{22}+\max\{\bar{q}D_{13}, qD_{23}\}.
	\end{eqnarray*}
	We therefore need to consider the following cases:
	\begin{description}
		\item[Case I:] $\bar{\alpha}\bar{p} D_{13}+\alpha \bar{p} D_{23}\leq\beta p D_{13}+\bar{\beta}p D_{23}$ and $\bar{q}D_{13}\leq qD_{23}$. The computation in this case reduces to the computation for \dsty{F=\begin{bmatrix}1 & 0 & 0\\0 & u & v\end{bmatrix}}.
		
		\item[Case II:] $\bar{\alpha}\bar{p}D_{13}+\alpha\bar{p}D_{23}\leq\beta p D_{13}+\bar{\beta} p D_{23}$ and $\bar{q}D_{13}> qD_{23}$. Notice that these conditions imply that $\zeta D_{13}\leq D_{23}< \eta D_{13}$, and therefore this case requires $\zeta<\eta$ (or equivalently, $\alpha\bar{\alpha}\bar{p}^2<\beta\bar{\beta}p^2$). This yields
		\begin{equation*}
         b^{(D)}=(\bar{\alpha}\bar{p}-\beta p) D_{11}+(\alpha\bar{p}-\bar{\beta} p) D_{21},
		\end{equation*}
		and 
		\begin{equation*}
		\beta^{(D)}=qD_{22}-\bar{q}D_{12}.
		\end{equation*}
		Hence, we have
		\eq{\frac{\beta^{(D)}}{b^{(D)}} = \frac{q}{\bar{\beta}p-\alpha \bar{p}} \frac{D_{22}-\eta D_{12}}{\zeta D_{11}-D_{21}}.}
		By the form of $F$, we have that $-D_{11},D_{12},D_{21}\geq0$. The inequalities $\zeta<\eta$ and $\zeta D_{13}\leq D_{23}$ imply that \dsty{\frac{D_{22}-\eta D_{12}}{\zeta D_{11}-D_{21}} \geq 1}, and hence 
		\eqn{Filter2}{\frac{\beta^{(D)}}{b^{(D)}} \geq \frac{q}{\bar{\beta}p-\alpha \bar{p}}1_{\{\alpha\bar{\alpha}\bar{p}^2<\beta\bar{\beta}p^2\}}.}
		
		\item[Case III:] $\bar{\alpha}\bar{p}D_{13}+\alpha\bar{p}D_{23}> \beta p D_{13}+\bar{\beta} p D_{23}$ and $\bar{q}D_{13}\leq qD_{23}$. Notice that these conditions imply that $\eta D_{13}\leq D_{23}<\zeta D_{13}$, and hence this case requires $\zeta>\eta$ (or equivalently, $\alpha\bar{\alpha}\bar{p}^2>\beta\bar{\beta}p^2$). In this case, we have
		\begin{equation*}
b^{(D)}=(\beta p-\bar{\alpha}\bar{p})D_{12} + (\bar{\beta}p-\alpha\bar{p})D_{22}, 
		\end{equation*}
and 
		\begin{equation*}
		\beta^{(D)}=\bar{q}D_{11}-qD_{21}.
		\end{equation*}
		Therefore,
		\eq{\frac{\beta^{(D)}}{b^{(D)}} = \frac{\bar{q}}{\bar{\alpha}\bar{p}-\beta p} \frac{D_{11}-\eta^{-1}D_{21}}{-D_{12}+\zeta^{-1}D_{22}}.}
		By the form of $F$, we have that $-D_{22},D_{12},D_{21}\geq0$. The inequalities $\zeta^{-1}<\eta^{-1}$ and $\zeta D_{13}>D_{23}$ imply that \dsty{\frac{D_{11}-\eta^{-1}D_{21}}{-D_{12}+\zeta^{-1}D_{22}} > 1}, and hence \eqn{Filter3}{\frac{\beta^{(D)}}{b^{(D)}} > \frac{\bar{q}}{\bar{\alpha}\bar{p}-\beta p}1_{\{\alpha\bar{\alpha}\bar{p}^2>\beta\bar{\beta}p^2\}}.}
		
		\item[Case IV:] $\bar{\alpha}\bar{p}D_{13}+\alpha\bar{p}D_{23}>\beta p D_{13}+\bar{\beta} p D_{23}$ and $\bar{q}D_{13}> qD_{23}$. Notice that these two inequalities imply that $D_{23}< \min\{\zeta, \eta\} D_{13}$. For this case we have that
		\begin{equation*}
b^{(D)}=(\beta p-\bar{\alpha}\bar{p})D_{12} + (\bar{\beta}p-\alpha\bar{p})D_{22} ,
		\end{equation*}
		and 
		\begin{equation*}
\beta^{(D)}=qD_{22}-\bar{q}D_{12}.
		\end{equation*}
		Hence, we have
		\eq{\frac{\beta^{(D)}}{b^{(D)}}=\frac{q}{\bar{\beta}p-\alpha \bar{p}} \frac{\eta D_{12}-D_{22}}{\zeta D_{12}-D_{22}}.}
		By the form of $F$, we have that $-D_{22},D_{12}\geq0$. As before, we conclude that
		\begin{eqnarray}
		\frac{\beta^{(D)}}{b^{(D)}} &\geq& \frac{q}{\bar{\beta} p-\alpha\bar{p}}\min\left\{\frac{\eta}{\zeta},1\right\} \nonumber\\
		&=&\begin{cases}\frac{q}{\bar{\beta}p-\alpha\bar{p}}, & {\alpha\bar{\alpha}\bar{p}^2<\beta\bar{\beta}p^2},\\\frac{\bar{q}}{\bar{\alpha}\bar{p}-\beta p}, & {\alpha\bar{\alpha}\bar{p}^2\geq\beta\bar{\beta}p^2}.\end{cases} \label{Filter4}
		\end{eqnarray}
	\end{description}
	Combining \eqref{Filter1}, \eqref{Filter2}, \eqref{Filter3}, and \eqref{Filter4}, we obtain
	\eq{\min_{F\in\F \atop \P(F)=T} \min_{D\in\D(F) \atop b^{(D)}<0} \frac{\beta^{(D)}}{b^{(D)}} \geq \begin{cases}\frac{q}{\bar{\beta}p-\alpha\bar{p}}, & {\alpha\bar{\alpha}\bar{p}^2<\beta\bar{\beta}p^2},\\\frac{\bar{q}}{\bar{\alpha}\bar{p}-\beta p}, & {\alpha\bar{\alpha}\bar{p}^2\geq\beta\bar{\beta}p^2},\end{cases}}
	as desired.
\end{proof}

\section{Proof of Theorem~\ref{Thm:GeneralizedLocalLinearity}}
\label{Appendix:ProofGeneralizedLocalLinearity}

Recall that $\X=\{\repn{M}\}$ and $\Y=\Z=\{\repn{N}\}$,   $P=[P(x,y)]_{(x,y)\in\X\times\Y}$ is the joint probability matrix of $X$ and $Y$, and the marginals are $p_X(x)=\Pr(X=x)$ and $q_Y(y)=\Pr(Y=y)$ for every $x\in \X$ and $y\in\Y$. Similar to $\mathcalboondox{h}$, the function $\underline{\mathcalboondox{h}}$ admits the alternative formulation
\eq{\underline{\mathcalboondox{h}}(\eps) = \sup_{F\in\ul{\F}:~\ul{\P}(F)\leq\eps} \ul{\U}(F),}
where $\ul{\F}$ is the set of all stochastic matrices $F\in \M_{N\times N}$,
\begin{equation*}
\ul{\P}(F) = \sum_{z\in\Z} \max_{x\in\X} (PF)(x,z), 
\end{equation*}
and 
\begin{equation*}
\ul{\U}(F) = \sum_{z\in\Z} \max_{y\in\Y} q_Y(y)F(y,z).
\end{equation*}
We let $\ul{\D} = \left\{D\in\M_{N\times N}: \|D\|=1\right\}$ and, for each $F\in\ul{\F}$, we define 
\eq{\ul{\D}(F)\coloneqq\left\{D\in \ul{\D}: F+tD\in \ul{\F} \textnormal{ for some } t>0\right\}.}
Before proving Theorem~\ref{Thm:GeneralizedLocalLinearity}, we need to establish some technical lemmas. Notice that the proofs of Lemmas~\ref{Lemma:CompactnessDirections} and \ref{Lemma:LocalLinearity} do not depend on the alphabets $\X$, $\Y$, and $\Z$. Therefore,  $\ul{\D}(F)$ is compact for any $F\in \ul{\F}$ and also we obtain the following lemma.
\begin{lemma}
\label{Lemma:LocalLinearityPsi}
Let $\ul{\H}:\ul{\F}\to[0,1]\times[0,1]$ be the mapping given by $\ul{\H}(F)=(\ul{\P}(F),\ul{\U}(F))$. For every $F\in\ul{\F}$, there exists $\delta>0$ such that $\ul{\H}$ is linear on $[F,F+\delta D]$ for every $D\in\ul{\D}(F)$.
\end{lemma}

The convex analysis tools used to study $\mathcalboondox{h}$ heavily rely on the fact that $|\Z|=|\Y|+1$. Hence, they are unavailable in this case, and thus we need an alternative approach to establish the desired functional properties of $\underline{\mathcalboondox{h}}$. 

\begin{lemma}\label{Lemma:Contin}
If $\cP(X)<\cP(X|Y)$, then $\underline{\mathcalboondox{h}}$ is continuous at $\cP(X|Y)$.
\end{lemma}

\begin{proof}
Without loss of generality, we will assume that $q_Y(1)>0$. Let $D_*\in\ul{\D}({\rm I}_N)$ be given by
\eq{\small D_* = \begin{bmatrix} 0 & 0 & 0 & \cdots & 0\\ \lambda & -\lambda & 0 & \cdots & 0\\ \lambda & 0 & -\lambda & \cdots & 0\\ \vdots & \vdots & \vdots & \ddots & \vdots\\ \lambda & 0 & 0 & \cdots & -\lambda\end{bmatrix},}
where $\lambda=(2(N-1))^{-1/2}$. As in the proof of Lemma~\ref{Lemma:LocalLinearity}, one can show that there exist $\delta_1>0$ and $(x_z)_{z\in\Z}\subset\X$ such that for every $t\in[0,\delta_1]$,
\begin{eqnarray}
\ul{\P}({\rm I}_N+tD_*) &=& \sum_{z\in\Z} \max_{x\in\X} (P({\rm I}_N+tD_*))(x,z)\nonumber\\
&=&\sum_{z\in\Z} (P({\rm I}_N+tD_*))(x_z,z). \label{DefXZ}
\end{eqnarray}
In this case, we have that
\begin{eqnarray*}
\ul{\P}({\rm I}_N+tD_*) &=& P(x_1,1) + t\lambda \sum_{z=2}^N P(x_1,z) \\
&&+ (1-t\lambda) \sum_{z=2}^N P(x_z,z)\\
&=&\sum_{z\in\Z} P(x_z,z) \\
&&- t\lambda \left(\sum_{z\in\Z} P(x_z,z) - P(x_1,z)\right).
\end{eqnarray*}
Note that \ndsty{\cP(X|Y) = \ul{\P}({\rm I}_N) = \sum_{z\in\Z} P(x_z,z)}. Hence,
\eqn{eq:PrivacyContinuityT}{\ul{\P}({\rm I}_N+tD_*) = \cP(X|Y) - t\lambda\sigma,}
where \dsty{\sigma=\sum_{z\in\Z} (P(x_z,z) - P(x_1,z))}. Setting $t=0$ in \eqref{DefXZ}, we have that $P(x_z,z)\geq P(x,z)$ for all $(x,z)\in\X\times\Z$. If $P(x_z,z)=P(x_1,z)$ for all $z\geq1$, then
\eq{\cP(X|Y) = \sum_{z\in\Z} P(x_1,z) = p_X(x_1) \leq \cP(X),}
which contradicts the hypothesis of the lemma. Therefore, there exists $z\in\Z$ such that $P(x_z,z)>P(x_1,z)$ and hence $\sigma>0$. Similarly, there exists $\delta_2>0$ such that for every $t\in[0,\delta_2]$,
\eqn{eq:UtilityContinuityT}{\ul{\U}({\rm I}_N+tD_*) = q_Y(1) + (1-t\lambda)\sum_{z=2}^N q_Y(z) = 1-t\lambda(1-q_Y(1)).}
Let $\delta=\min(\delta_1,\delta_2)$. From  \eqref{eq:PrivacyContinuityT} and \eqref{eq:UtilityContinuityT}, we have for every $t\in[0,\delta]$
\eqn{LowerBound}{1-t\lambda(1-q_Y(1)) \leq \underline{\mathcalboondox{h}}(\cP(X|Y)-t\lambda\sigma) \leq 1.}
In particular,
\eq{\lim_{\eps\to\cP(X|Y)} \underline{\mathcalboondox{h}}(\eps) = \lim_{t\to0} \underline{\mathcalboondox{h}}(\cP(X|Y)-t\lambda\sigma) = 1 = \underline{\mathcalboondox{h}}(\cP(X|Y)),}
i.e., $\underline{\mathcalboondox{h}}$ is continuous at $\cP(X|Y)$.
\end{proof}

We say that $F\in\ul{\F}$ is an optimal filter at $\eps$ if $\ul{\U}(F) = \underline{\mathcalboondox{h}}(\eps)$ and $\ul{\P}(F)\leq\eps$. As opposed to $\mathcalboondox{h}$, the concavity of $\underline{\mathcalboondox{h}}$ is unknown and hence the existence of an optimal filter at $\eps$ with $\ul{\P}(F)=\eps$ is not immediate. Nonetheless, since $\ul{\P}$ and $\ul{\U}$ are continuous functions,  there exists an optimal filter $F$ at $\eps$ (with $\ul{\P}(F)\leq\eps$) for every $\eps\in[\cP(X),\cP(X|Y)]$. For any $F\in\ul{\F}$ and $\delta>0$, let $B(F,\delta)=\{G\in\ul{\F}:\|G-F\|<\delta\}$. 

\begin{lemma}\label{Lemma_Ball}
Let $\delta>0$ be as in Lemma~\ref{Lemma:LocalLinearityPsi} for ${\rm I}_N$, i.e., $\ul{\U}$ and $\ul{\P}$ are linear on $[{\rm I}_N, {\rm I}_N+\delta D]$ for every $D\in \ul{\D}({\rm I}_N)$. If $\cP(X)<\cP(X|Y)$ and $q_Y(y)>0$ for all $y\in\Y$, then there exists $\eps_\mathsf{L}<\cP(X|Y)$ such that for every $\eps\in[\eps_\mathsf{L},\cP(X|Y)]$ there exists an optimal filter $F_\eps$ at $\eps$ with $F_\eps\in B({\rm I}_N,\delta)$.
\end{lemma}

\begin{proof}
Let $\ul{\F}^1=\{F\in\ul{\F}:\ul{\U}(F)=1\}$ and let \dsty{\B = \bigcup_{F\in\ul{\F}^1} B(F,\delta)}. The proof is based on the following claim.

\noindent{\bf Claim.} There exists $\eps_\mathsf{L}<\cP(X|Y)$ such that if $F$ is an optimal filter at $\eps$ with $\eps\geq\eps_\mathsf{L}$, then $F\in\B$.

\begin{description}
	\item[\it Proof of the claim.] The proof is by contradiction. Assume that for every $\eps<\cP(X|Y)$ there exists an optimal filter $G_{\eps'}$ at $\eps'\in[\eps, \cP(X|Y))$ with $G_{\eps'}\notin\B$. Since $\underline{\mathcalboondox{h}}$ is a non-decreasing function, we have that $\ul{\U}(G_{\eps'}) = \underline{\mathcalboondox{h}}(\eps')\geq\underline{\mathcalboondox{h}}(\eps)$. Let $K\coloneqq (\cP(X|Y)-\cP(X))^{-1}$. For each $n>K$, let $F_n= G_{\cP(X|Y)-1/n}\not\in\B$. Since $\F\backslash\B$ is compact, there exist $\{n_1<n_2<\cdots\}$ and $F\in\F\backslash\B$ such that $F_{n_k} \to F$ as $k\to\infty$. By continuity of $\ul{\U}$ and $\underline{\mathcalboondox{h}}$ at $\cP(X|Y)$, established in Lemma~\ref{Lemma:Contin}, we have 
\begin{eqnarray*}
1&\geq &\ul{\U}(F) = \lim_{k\to\infty} \ul{\U}(F_{n_k})\\
 &\geq&\lim_{k\to\infty} \underline{\mathcalboondox{h}}(\cP(X|Y) - n_k^{-1}) = \underline{\mathcalboondox{h}}(\cP(X|Y)) = 1.
\end{eqnarray*}
\noindent In particular, we have that $F\in\ul{\F}^1\subset\B$, which contradicts the fact that $F\in\F\backslash\B$. \qed
\end{description}
\noindent The assumption $q_Y(y)>0$ for every $y\in\Y$ implies that $F\in\ul{\F}^1$ if and only if $F$ is a permutation matrix, i.e., $F$ can be obtained by permuting the columns of ${\rm I}_N$. In particular, the mapping $G\mapsto GF^{-1}$ is a bijection between $B(F,\delta)$ and $B({\rm I}_N,\delta)$ which preserves $\ul{\P}$ and $\ul{\U}$, i.e., $\ul{\P}(G)=\ul{\P}(GF^{-1})$ and $\ul{\U}(G)=\ul{\U}(GF^{-1})$ for every $G\in B(F,\delta)$. As mentioned earlier, there exists an optimal filter $F_\eps$ at $\eps$ for every $\eps\in[\cP(X),\cP(X|Y)]$. By the claim, $F_\eps$, for $\eps\geq\eps_\mathsf{L}$, belongs to $\B$ and, in particular, $F_\eps\in B(F,\delta)$ for some $F\in\ul{\F}^1$. By the aforementioned properties of the bijection $G\mapsto GF^{-1}$, the filter $F_\eps F^{-1}$ is an optimal filter at $\eps$ with $F_\eps F^{-1}\in B({\rm I}_N,\delta)$.
\end{proof}

Now we are in position to prove Theorem~\ref{Thm:GeneralizedLocalLinearity}.

\begin{proof}[Proof of Theorem~\ref{Thm:GeneralizedLocalLinearity}]
If $q_Y(y)=0$ for some $y\in\Y$, the effective cardinality of the alphabet of $Y$ is $|\Y|-1$ and thus $\underline{\mathcalboondox{h}}(\eps)$ equals $\mathcalboondox{h}(\eps)$ for every $\eps\in[\cP(X),\cP(X|Y)]$. In this case,  $\underline{\mathcalboondox{h}}$ is piecewise linear and \eqref{eq:PsiExtremePoint} follows trivially by Theorem~\ref{Thm:PiecewiseLinearity}. In what follows, we assume that $q_Y(y)>0$ for all $y\in\Y$.

Let $\delta>0$ and $\eps'_\mathsf{L}<\cP(X|Y)$ be as in Lemma~\ref{Lemma_Ball}. For each $\eps\in[\eps'_\mathsf{L},\cP(X|Y))$, let $G_\eps$ be an optimal filter at $\eps$ with $G_\eps\in B({\rm I}_N,\delta)$ whose existence was established in Lemma~\ref{Lemma_Ball}. Let $t_\eps\in[0,\delta]$ and $D_\eps\in\ul{\D}({\rm I}_N)$ be such that $G_\eps={\rm I}_N + t_\eps D_\eps$ for every $\eps\in[\eps'_\mathsf{L},\cP(X|Y))$. As in \eqref{eq:PrivacyMasterEquation} and \eqref{eq:UtilityMasterEquation} in the proof of Lemma~\ref{Lemma:LocalLinearity}, for every $t\in[0,\delta]$ and $D\in\ul{\D}({\rm I}_N)$,
\begin{equation}
\begin{aligned}
\ul{\P}({\rm I}_N+tD) &= \cP(X|Y) + tb^{(D)} \\
\ul{\U}({\rm I}_N+tD) &= 1 + t \beta^{(D)},
\end{aligned}
\label{eq:PrivacyUtilityUnderbar}
\end{equation}
where
\begin{equation}
\begin{aligned}
b^{(D)} &= \sum_{z\in\Z} \max_{x\in\M_z} (PD)(x,z) \\
\beta^{(D)} &= \sum_{z\in\Z} q(z) D(z,z),
\end{aligned}
\label{eq:bbeta}
\end{equation}
where $\M_z = \{x\in\X : P(x,z)\geq P(x',z)\textnormal{ for all }x'\in\X\}$. 
Since $\ul{\P}(F)\leq\cP(X|Y)$ for all $F\in\ul{\F}$, it is immediate that $b^{(D)}\leq 0$ for every $D\in\ul{\D}({\rm I}_N)$. Moreover, since $\ul{\P}(G_\eps)\leq \eps$, we have that $b^{(D_\eps)}<0$ for all $\eps\in[\eps'_\mathsf{L},\cP(X|Y))$. By definition of $\ul{\D}({\rm I}_N)$, it is clear that if $D\in\ul{\D}({\rm I}_N)$, then we have $D(y,y)\leq0$ for all $y\in\Y$, which together with the fact that $\|D\|=1$ for all $D\in\ul{\D}({\rm I}_N)$, implies that $\beta^{(D)}<0$ for all $D\in\ul{\D}({\rm I}_N)$.  We first establish the following intuitive claim.

\noindent{\bf Claim.} Let $\eps'_\mathsf{L}<\cP(X|Y)$ be as defined in Lemma~\ref{Lemma_Ball}. Then, there exists an optimal filter $G_\eps$ at $\eps$ for each $\eps\in [\eps'_\mathsf{L}, \cP(X|Y)]$ such that $\ul{\P}(G_\eps)=\eps$ and $\ul{\U}(G_\eps)=\underline{\mathcalboondox{h}}(\eps)$. 
\begin{description}
	\item[\it Proof of Claim.]  
	The filter $G_\eps={\rm I}_N + t_\eps D_\eps$ is optimal at $\eps$ for every $\eps\in[\eps'_\mathsf{L},\cP(X|Y))$.  To reach contradiction, assume that there exists $\eps_0<\eps$ such that $\ul{\P}(G_\eps)=\eps_0$.  According to \eqref{eq:PrivacyUtilityUnderbar}, we obtain $\cP(X|Y)+t_\eps b^{(D_\eps)}=\eps_0<\eps$ and hence $$t_\eps>\frac{\cP(X|Y)-\eps}{-b^{(D_\eps)}}\eqqcolon t'.$$ 
	Now consider the filter ${\rm I}_N + t' D_\eps$. Since $t'\leq \delta$, we have from \eqref{eq:PrivacyUtilityUnderbar} that $\ul{\P}({\rm I}_N + t' D_\eps)=\eps$ and 	
	$$\underline{\mathcalboondox{h}}(\eps)\stackrel{(a)}{=}1 + t_\eps \beta^{(D_\eps)}\stackrel{(b)}{<}\ul{\U}({\rm I}_N + t' D_\eps)=1+t'\beta^{(D_\eps)},$$
where $(a)$ is due to the optimality of $G_\eps$ and $(b)$ follows from the negativity of $\beta^{(D_\eps)}$. The above inequality contradicts the maximality of $\underline{\mathcalboondox{h}}(\eps)$.	This implies that $\ul{\P}(G_\eps)=\eps$ which, according to \eqref{eq:PrivacyUtilityUnderbar}, yields 
\eqn{EqualityEps}{\underline{\mathcalboondox{h}}(\eps) = 1 - (\cP(X|Y)-\eps)\frac{\beta^{(D_\eps)}}{b^{(D_\eps)}},}
for all $\eps\in[\eps'_\mathsf{L},\cP(X|Y))$. 
\qed
\end{description}
\noindent

Now fix $\eps'\in [\eps'_\mathsf{L}, \cP(X|Y)]$ with $\eps\leq\eps'$. On the one hand, according to \eqref{EqualityEps}, we know that 
\eqn{EqualityEPsPrime}{\underline{\mathcalboondox{h}}(\eps') = 1 - (\cP(X|Y)-\eps')\frac{\beta^{(D_{\eps'})}}{b^{(D_{\eps'})}}.} On the other hand, we obtain from \eqref{eq:PrivacyUtilityUnderbar} that  $0\leq\frac{\cP(X|Y)-\eps'}{-b^{(D_\eps)}} \leq t_\eps$ and hence
\aln{
\label{eq:PrivThm3}\ul{\P}\left({\rm I}_N + \frac{\cP(X|Y)-\eps'}{-b^{(D_\eps)}} D_\eps\right) &= \eps',\\
\label{eq:UtiThm3}\ul{\U}\left({\rm I}_N + \frac{\cP(X|Y)-\eps'}{-b^{(D_\eps)}} D_\eps\right) &= 1- (\cP(X|Y)-\eps') \frac{\beta^{(D_\eps)}}{b^{(D_\eps)}}.
}
Comparing \eqref{EqualityEPsPrime} and \eqref{eq:UtiThm3}, we conclude that 
\eq{1 - (\cP(X|Y)-\eps')\frac{\beta^{(D_{\eps'})}}{b^{(D_{\eps'})}} = \underline{\mathcalboondox{h}}(\eps') \geq 1- (\cP(X|Y)-\eps')\frac{\beta^{(D_\eps)}}{b^{(D_\eps)}},}
and hence the function \dsty{\eps \mapsto \frac{\beta^{(D_\eps)}}{b^{(D_\eps)}}} is non-increasing over $[\eps'_\mathsf{L},\cP(X|Y))$. Therefore, since \dsty{\frac{\beta^{(D_\eps)}}{b^{(D_\eps)}}>0}, the limit \dsty{\lim_{\eps\to\cP(X|Y)^-} \frac{\beta^{(D_\eps)}}{b^{(D_\eps)}}\eqqcolon A} exists.

Let $K=(\cP(X|Y)-\eps'_\mathsf{L})^{-1}$. For each $n>K$, let $F_n=G_{\cP(X|Y)-\frac{1}{n}}$. Write $F_n={\rm I}_N+t_nD_n$ with $t_n\in[0,\delta]$ and $D_n\in\ul{\D}({\rm I}_N)$. Since $\ul{\D}({\rm I}_N)$ is compact, there exist $\{n_1<n_2<\cdots\}$ and $D^*\in\ul{\D}({\rm I}_N)$ such that $D_{n_k}\to D^*$ as $k\to\infty$. By continuity of the mappings $D\mapsto b^{(D)}$ and $D\mapsto \beta^{(D)}$, we have that $b^{(D_{n_k})} \to b^{(D^*)}$ and $\beta^{(D_{n_k})}\to\beta^{(D^*)}$ as $k\to\infty$. 

\noindent{\bf Claim.} We have that $b^{(D^*)}<0$ and, in particular, \dsty{A=\frac{\beta^{(D^*)}}{b^{(D^*)}}}.

\begin{description}
	\item[\it Proof of Claim.] Recall that $F\in\ul{F}^1$ if and only if $F$ is a permutation matrix. In particular, $\ul{\F}^1$ is finite with $|\ul{\F}^1|=N!$. Recall that $b^{(D^*)}\leq 0$. Assume that $b^{(D^*)}=0$. Since \dsty{\frac{\beta^{(D_{n_k})}}{b^{(D_{n_k})}} \to A\in[0,\infty)} and $b^{(D_{n_k})}\to b^{(D^*)}=0$ as $k\to \infty$, we have that $\beta^{(D_{n_k})} \to 0$ and hence $\beta^{(D^*)}=0$. This implies that $\ul{\U}({\rm I}_N+tD^*) = 1$ for all $t\in[0,\delta]$, i.e., ${\rm I}_N+tD^*\in\ul{\F}^1$ for all $t\in[0,\delta]$. This contradicts the fact that $\ul{\F}^1$ is finite. \qed
\end{description}

\noindent The claim implies that for $\eps\in[\cP(X|Y)+\delta b^{(D^*)},\cP(X|Y)]$,
\al{
\ul{\P}\left({\rm I}_N + \frac{\cP(X|Y)-\eps}{-b^{(D^*)}} D^*\right) &= \eps,\\
\ul{\U}\left({\rm I}_N + \frac{\cP(X|Y)-\eps}{-b^{(D^*)}} D^*\right) &= 1- (\cP(X|Y)-\eps) A.
}
Recall that \dsty{\frac{\beta^{(D^*)}}{b^{(D^*)}} = A \leq \frac{\beta^{(D_\eps)}}{b^{(D_\eps)}}} for all $\eps\in[\eps'_\mathsf{L},\cP(X|Y))$. Let $\eps_\mathsf{L}\coloneqq\max\{\eps'_\mathsf{L},\cP(X|Y)+\delta b^{(D^*)}\}$. Then for all $\eps\in[\eps_\mathsf{L},\cP(X|Y)]$
\begin{eqnarray}
\underline{\mathcalboondox{h}}(\eps) &\geq& 1 - (\cP(X|Y)-\eps) \frac{\beta^{(D^*)}}{b^{(D^*)}} \nonumber\\
&\geq& 1- (\cP(X|Y)-\eps) \frac{\beta^{(D_\eps)}}{b^{(D_\eps)}} =\underline{\mathcalboondox{h}}(\eps), \label{eq:LinearityPsiProof}
\end{eqnarray}
where the equality follows from \eqref{EqualityEps}. This proves that $\underline{\mathcalboondox{h}}$ is linear on $\eps\in[\eps_\mathsf{L},\cP(X|Y)]$.

Recall that $\beta^{(D)}<0$ for all $D\in\ul{\D}({\rm I}_N)$. Clearly,  \eqref{eq:LinearityPsiProof} implies that 
\eqn{eq:MasterDerivativePsi}{\underline{\mathcalboondox{h}}'(\cP(X|Y)) = \min_{D\in\ul{\D}({\rm I}_N)} \frac{\beta^{(D)}}{b^{(D)}}.}
If $b^{(D)}=0$ for some $D\in\ul{\D}({\rm I}_N)$, the term \dsty{\frac{\beta^{(D)}}{b^{(D)}}} is defined to be $+\infty$. Notice that this convention agrees with the fact that if $b^{(D)}=0$ then $D$ cannot be an \emph{optimal direction}. Furthermore, for every $D'\in\ul{\D}({\rm I}_N)$ such that $\underline{\mathcalboondox{h}}'(\cP(X|Y)) = \frac{\beta^{(D')}}{b^{(D')}}$, there exists $\eps_\mathsf{L}<\cP(X|Y)$ (depending on $D'$) such that
\eqn{eq:MasterOptimalFilterPsi}{{\rm I}_N + \frac{\cP(X|Y)-\eps}{-b^{(D')}} D'}
achieves $\underline{\mathcalboondox{h}}(\eps)$ for every $\eps\in[\eps_\mathsf{L},\cP(X|Y)]$. In addition, assume that for each $y\in\Y$ there exists (a unique) $x_y\in\X$ such that $P_{X|Y}(x_y|y) > P_{X|Y}(x|y)$, for all $x\neq x_y$.
In particular, $\M_z=\{x_z\}$ for every $z\in\Z$ and hence \eqref{eq:bbeta} becomes
\eq{b^{(D)} = \sum_{z\in\Z} (PD)(x_z,z) \quad \textnormal{ and } \quad \beta^{(D)} = \sum_{z\in\Z} q_Y(z) D(z,z),}
for every $D\in\ul{\D}({\rm I}_N)$. Using the fact that \dsty{\sum_{z\in\Z} D(y,z) = 0} for all $y\in\Y$, we obtain
\begin{equation*}
b^{(D)} = - \sum_{y\in\Y} \sum_{z\neq y} (P(x_y,y)-P(x_z,y))D(y,z),
\end{equation*}
and 
\begin{equation*}
\beta^{(D)} = - \sum_{y\in\Y} \sum_{z\neq y} q_Y(y) D(y,z).
\end{equation*}
Therefore, for every $D\in\ul{\D}({\rm I}_N)$,
\eqn{eq:MasterEquationQuotient}{\frac{\beta^{(D)}}{b^{(D)}} = \frac{\sum_{y\in\Y} \sum_{z\neq y} q_Y(y) D(y,z)}{\sum_{y\in\Y} \sum_{z\neq y} (P(x_y,y)-P(x_z,y))D(y,z)}.}
Since \dsty{\frac{\sum_k a_k x_k}{\sum_k b_k x_k} \geq \min_k \frac{a_k}{b_k}} for $a_k>0$  and $b_k,x_k\geq0$ with $\sum_k x_k>0$, we obtain from \eqref{eq:MasterEquationQuotient} that for every $D\in\ul{\D}({\rm I}_N)$
\eq{\frac{\beta^{(D)}}{b^{(D)}} \geq \min_{(y,z)\in\Y\times\Z} \frac{q_Y(y)}{P(x_y,y)-P(x_z,y)}.}
Equation \eqref{eq:MasterDerivativePsi} implies that
\eq{\underline{\mathcalboondox{h}}'(\cP(X|Y)) \geq \min_{(y,z)\in\Y\times\Z} \frac{q_Y(y)}{P(x_y,y)-P(x_z,y)}.}
Assume that $(y_0,z_0)$ attains the above minimum.  We note that
one can easily show from  \eqref{LowerBound} that $0\leq\underline{\mathcalboondox{h}}'(\eps)\leq\frac{1-q_Y(1)}{\sigma}<\infty$, for some $\sigma>0$. Hence,  we  have $y_0\neq z_0$. Now, consider the direction $D_*$ such that
\eq{D_*(y,z) = \begin{cases} \lambda, & y=y_0,z=z_0 \\ -\lambda, & y=z=y_0 \\ 0, & \textnormal{otherwise,} \end{cases}}
where $\lambda=2^{-1/2}$. Equation \eqref{eq:MasterEquationQuotient} implies then that
\eq{\frac{\beta^{(D_*)}}{b^{(D_*)}} = \frac{q_Y(y_0)}{P(x_{y_0},y_0) - P(x_{z_0},y_0)},}
and hence
\begin{eqnarray*}
\underline{\mathcalboondox{h}}'(\cP(X|Y)) &\leq& \frac{q_Y(y_0)}{P(x_{y_0},y_0) - P(x_{z_0},y_0)} \\
&=& \min_{(y,z)\in\Y\times\Z} \frac{q_Y(y)}{P(x_y,y)-P(x_z,y)}.
\end{eqnarray*}
As a consequence,
\eq{\underline{\mathcalboondox{h}}'(\cP(X|Y)) = \min_{(y,z)\in\Y\times\Z} \frac{q_Y(y)}{P(x_y,y)-P(x_z,y)}.}
Moreover, \eqref{eq:MasterOptimalFilterPsi} implies that there exists $\eps_\mathsf{L}^{y_0,z_0}<\cP(X|Y)$ such that \dsty{{\rm I}_N + \frac{\cP(X|Y)-\eps}{-b^{(D_*)}} D_*} achieves $\underline{\mathcalboondox{h}}(\eps)$ for every $\eps\in[\eps_\mathsf{L}^{y_0,z_0},\cP(X|Y)]$. Note that
\eq{{\rm I}_N+\frac{\cP(X|Y)-\eps}{-b^{(D_*)}} D_*= \mathsf{Z}^{y_0,z_0}(\zeta^{y_0,z_0}(\eps)),}
where \dsty{\zeta^{y_0,z_0}(\eps) = \frac{\cP(X|Y) - \eps}{P(x_{y_0},y_0) - P(x_{z_0},y_0)}}.
\end{proof}

\section{Proof of Theorem~\ref{Prop:IIDDataUnderbar}}
\label{Appendix:ProofPropIIDDataUnderbar}

Let $P=[P(x^n,y^n)]_{x^n,y^n\in\{0,1\}^n}$ denotes the joint probability matrix of $X^n$ and $Y^n$ and $q(y^n)=\Pr(Y^n=y^n)$ for $y^n\in\{0,1\}^n$. Let ${\bf 0}=(0,0,\ldots,0)$ and ${\bf 1}=(1,1,\ldots,1)$. We will show that $(X^n,Y^n)$ satisfies the hypotheses of Theorem~\ref{Thm:GeneralizedLocalLinearity} with $y_0={\bf 1}$ and $z_0={\bf 0}$.

Under the assumptions ($\textnormal{a}_1$) and ($\textnormal{b}$), it is straightforward to verify that
\begin{equation}
\label{Joint_Dist_Thm3}
P(x^n,y^n) = (\bar{\alpha}\bar{p})^n \prod_{k=1}^n \left(\frac{p}{\bar{p}}\right)^{x_k} \left(\frac{\alpha}{\bar{\alpha}}\right)^{x_k \oplus y_k},
\end{equation}
for every $x^n,y^n\in\{0,1\}^n$. By assumption, $\cP(X^n)=p^n<\bar{\alpha}^n=\cP(X^n|Y^n)$. It is also straightforward to verify that $q(y^n)>0$ for all $y\in\{0,1\}^n$. Since $\bar{\alpha}\bar{p}>\alpha p$, we have from  \eqref{Joint_Dist_Thm3} that
\eq{\Pr(X^n=z^n,Y^n=z^n) > \Pr(X^n=x^n,Y^n=z^n),}
for all $x^n\neq z^n$. In the notation of Theorem~\ref{Thm:GeneralizedLocalLinearity}, $x^n_{z^n}=z^n$ for all $z^n\in\{0,1\}^n$.
Note that
\begin{eqnarray*}
&&\min_{y^n,z^n\in\{0,1\}^n} \frac{q(y^n)}{P(x^n_{y^n}, y^n)-P(x^n_{z^n}, y^n)} \\
&=&\min_{y^n\in\{0,1\}^n} \frac{q(y^n)}{P(y^n,y^n)- \min\limits_{z^n\neq y^n} P(z^n,y^n)}.
\end{eqnarray*}
It is easy to show that \dsty{\min_{z^n\neq y^n} P(z^n,y^n) = (\alpha p)^n \prod_{k=1}^n \left(\frac{p}{\bar{p}}\right)^{-y_k}} and that the minimum is  attained by $z^n=(\bar{y}_1, \bar{y}_2, \dots, \bar{y}_n)$. As a consequence,
\al{
&\min_{y^n,z^n\in\{0,1\}^n} \frac{q(y^n)}{P(x^n_{y^n}, y^n)-P(x^n_{z^n}, y^n)} \\
&= \min_{y^n\in\{0,1\}^n} \frac{\sum\limits_{x^n\in\{0,1\}^n} \prod\limits_{k=1}^n \left(\frac{p}{\bar{p}}\right)^{x_k-y_k} \left(\frac{\alpha}{\bar{\alpha}}\right)^{x_k \oplus y_k}}{1 - \left(\frac{p\alpha}{\bar{p}\bar{\alpha}}\right)^n \Pi_{y^n}^{-2}}\\
&= \min_{y^n\in\{0,1\}^n} \frac{\prod\limits_{k=1}^n \left[(\frac{p}{\bar{p}})^{-y_k}(\frac{\alpha}{\bar{\alpha}})^{y_k}+(\frac{p}{\bar{p}})^{1-y_k}(\frac{\alpha}{\bar{\alpha}})^{1-y_k}\right]}{1 - \left(\frac{p\alpha}{\bar{p}\bar{\alpha}}\right)^n \Pi_{y^n}^{-2}},}
where \dsty{\Pi_{y^n} = \prod_{k=1}^n \left(\frac{p}{\bar{p}}\right)^{y_k}}. Observe that the denominator is maximized when $y^n={\bf 1}$. Using the fact that $p\geq\frac{1}{2}\geq\bar{p}$, one can show that the numerator is minimized when $y^n={\bf 1}$. In particular,
\eq{\min_{y^n,z^n\in\{0,1\}^n} \frac{q(y^n)}{P(x^n_{y^n}, y^n)-P(x^n_{z^n}, y^n)} = \frac{(\alpha\bar{p}+\bar{\alpha}p)^n}{(\bar{\alpha}p)^n-(\alpha\bar{p})^n},}
and the minimum is attained by $(y_0^n,z_0^n)=({\bf 1},{\bf 0})$.

Therefore $(X^n,Y^n)$ satisfies the hypotheses of Theorem~\ref{Thm:GeneralizedLocalLinearity} with $(y_0^n,z_0^n)=({\bf 1},{\bf 0})$. Thus, there exists $\eps_\mathsf{L}'<\bar{\alpha}^n$ such that for every $\eps\in[\eps_\mathsf{L}',\bar{\alpha}^n]$
\eq{\underline{\mathcalboondox{h}}(\eps) = 1 - \frac{\bar{\alpha}^n-\eps}{(\bar{\alpha}p)^n-(\alpha\bar{p})^n} q^n.}
Moreover, $\mathsf{Z}^{{\bf 1},{\bf 0}}(\zeta^{y_0,z_0}(\eps))$ achieves $\underline{\mathcalboondox{h}}(\eps)$ for every $\eps\in[\eps_\mathsf{L}',\bar{\alpha}^n]$, where
\eq{\zeta^{y_0,z_0}(\eps) = \frac{\bar{\alpha}^n-\eps}{(\bar{\alpha}p)^n-(\alpha\bar{p})^n}.}
Recall that $\underline{\mathcalboondox{h}}(\eps) = \underline{\mathcalboondox{h}}_n^n(\eps^{1/n})$ and let $\eps_\mathsf{L}=(\eps_\mathsf{L}')^{1/n}$. Therefore, $\underline{\mathcalboondox{h}}_n^n(\eps) = 1 - \zeta_n(\eps) q^n$ for all $\eps\in[\eps_\mathsf{L},\bar{\alpha}]$ which is attained by the Z-channel $\mathsf{Z}_n(\zeta_n(\eps))$, where $\zeta_n(\eps)\coloneqq \zeta^{y_0,z_0}(\eps^n)$.

\section{Proof of Proposition~\ref{Propo_h_iid}}\label{Appendix:PropositionIID}
	For any privacy filter satisfying  \eqref{eq:Defhni}, $(X^n,Z^n)$ and $(Y^n,Z^n)$ are i.i.d. By Lemma~\ref{Lemma_P_C_IID}, we have $\cP(X^n|Z^n)=(\cP(X|Z))^n$ and $\cP(Y^n|Z^n)=(\cP(Y|Z))^n$ where $(X,Y,Z)$ has the common distribution of $\{(X_k,Y_k,Z_k)\}_{k=1}^n$. In particular,
	\eq{\mathcalboondox{h}_n^\mathsf{i}(\eps) = \sup_{\cP^{1/n}(X^n|Z^n) \leq \eps} \cP^{1/n}(Y^n|Z^n) = \sup_{\cP(X|Z) \leq \eps} \cP(Y|Z),}
	where the first supremum assumes \eqref{eq:Defhni} and the second supremum is implicitly constrained to $\Z=\{0,1\}$. The result then follows from Theorem~\ref{Theorem_Linearity_BIBO}.

\section{Proof of Corollary~\ref{Corollary:DifferenceGap}}
\label{Appendix:ProofCorollaryDifferenceGap}

Assume that $p>\frac{1}{2}$. By Theorem~\ref{Prop:IIDDataUnderbar}, for every $\eps\in[\eps_\mathsf{L},\bar{\alpha}]$ we have  \dsty{\underline{\mathcalboondox{h}}_n(\eps) = \left[A_n\eps^n+B_n\right]^{1/n}}, where \dsty{A_n = \frac{q^n}{(\bar{\alpha}p )^n-(\alpha\bar{p})^n}} and \dsty{B_n = 1-\frac{\bar{\alpha}^n q^n}{(\bar{\alpha}p )^n-(\alpha\bar{p})^n}}. In particular,
\aln{
\label{eq:DerivativeUnderbarCorollary} \underline{\mathcalboondox{h}}'_n(\eps) &= A_n\left(\frac{\eps}{\underline{\mathcalboondox{h}}_n(\eps)}\right)^{n-1},\\
\nonumber \underline{\mathcalboondox{h}}''_n(\eps) &= (n-1) \frac{A_nB_n}{\underline{\mathcalboondox{h}}_n^{n+1}(\eps)} \left(\frac{\eps}{\underline{\mathcalboondox{h}}_n(\eps)}\right)^{n-2}.
}
Since $p>\frac{1}{2}$ and $\alpha>0$, we have  $B_n\to1$ as $n\to\infty$. Let $N_0\geq1$ be such that $B_n\geq0$ for all $n\geq N_0$. In this case, we have that $\underline{\mathcalboondox{h}}''_n(\eps)\geq0$ for all $\eps\in[\eps_\mathsf{L},\bar{\alpha}]$ and $n\geq N_0$. In particular, $\underline{\mathcalboondox{h}}_n$ is convex on $[\eps_\mathsf{L},\bar{\alpha}]$. As a consequence, for all $\eps\in[\eps_\mathsf{L},\bar{\alpha}]$ and $n\geq N_0$
\eq{\underline{\mathcalboondox{h}}_n(\eps) \geq 1-(\bar{\alpha}-\eps)\underline{\mathcalboondox{h}}_n'(\bar{\alpha}).}
Since $\mathcalboondox{h}^\mathsf{i}_n(\eps) = \underline{\mathcalboondox{h}}_1(\eps) = 1-(\bar{\alpha}-\eps)\underline{\mathcalboondox{h}}_1'(\bar{\alpha})$ for all $\eps\in[p,\bar{\alpha}]$, the above inequality implies that
\eq{\underline{\mathcalboondox{h}}_n(\eps) - \mathcalboondox{h}^\mathsf{i}_n(\eps) \geq (\bar{\alpha}-\eps)(\underline{\mathcalboondox{h}}_1'(\bar{\alpha})-\underline{\mathcalboondox{h}}_n'(\bar{\alpha}))}
for all $\eps\in[\eps_\mathsf{L},\bar{\alpha}]$ and $n\geq N_0$. The result follows from \eqref{eq:DerivativeUnderbarCorollary}.

Now, assume that $p=\frac{1}{2}$. In this case, we have  for all $\eps\in[\eps_\mathsf{L},\bar{\alpha}]$
\eq{\underline{\mathcalboondox{h}}_n(\eps) = \left(\frac{\eps^n-\alpha^n}{\bar{\alpha}^n-\alpha^n}\right)^{1/n} \quad \textnormal{ and } \quad \mathcalboondox{h}^\mathsf{i}_n(\eps) = \frac{\eps-\alpha}{\bar{\alpha}-\alpha}.}
Let $\Xi_n:[\frac{1}{2},\bar{\alpha}]\to\R$ be given by $\Xi_n(\eps)=\underline{\mathcalboondox{h}}_n(\eps)-\mathcalboondox{h}^\mathsf{i}_n(\eps)$.

\noindent{\bf Claim.} The function $\Xi_n$ is decreasing on $[\frac{1}{2},\bar{\alpha}]$.

\begin{description}
	\item[\it Proof of Claim.] We shall show that $\Xi_n'(\eps)\leq0$ for all $\eps\in[\frac{1}{2},\bar{\alpha}]$. A straightforward computation shows that
	\eq{\Xi_n'(\eps) = \frac{1}{\left[1-\left(\frac{\alpha}{\eps}\right)^n\right]^{(n-1)/n}} \frac{1}{[\bar{\alpha}^n-\alpha^n]^{1/n}} - \frac{1}{\bar{\alpha}-\alpha}.}
	This function is clearly decreasing, and so it is enough to show that $\Xi_n'(\frac{1}{2})\leq0$. Note that $\Xi_n'(\frac{1}{2})\leq0$ if and only if
	\eqn{eq:InqDecreasing}{\frac{\left(1-\frac{\alpha}{\bar{\alpha}}\right)^n}{1-\left(\frac{\alpha}{\bar{\alpha}}\right)^n} \leq [1-(2\alpha)^n]^{n-1}.}
	Observe that \dsty{\frac{\left(1-\frac{\alpha}{\bar{\alpha}}\right)^n}{1-\left(\frac{\alpha}{\bar{\alpha}}\right)^n} \leq \left(1-\frac{\alpha}{\bar{\alpha}}\right)^{n-1}}. Using the fact that $4\alpha\bar{\alpha}\leq1$, it is straightforward to verify that \eqref{eq:InqDecreasing} holds. \qed
\end{description}

\noindent Since $\Xi_n$ is decreasing over $[\frac{1}{2},\bar{\alpha}]$, we obtain for all $\eps\in[\eps_\mathsf{L},\bar{\alpha}]$
\eq{0 \leq \underline{\mathcalboondox{h}}_n(\eps) - \mathcalboondox{h}^\mathsf{i}_n(\eps) \leq \Xi_n\left(\frac{1}{2}\right)=\frac{1}{2} \left[\left(\frac{1-(2\alpha)^n}{\bar{\alpha}^n-\alpha^n}\right)^{1/n} - 1\right].}
Since $1-(2\alpha)^n\leq1-\left(\frac{\alpha}{\bar{\alpha}}\right)^n$, it is straightforward to show that $\Xi_n\left(\frac{1}{2}\right)\leq \frac{\alpha}{2\bar{\alpha}}$, which completes the proof.

\section{Proof of Theorem~\ref{Theorem_MarkovMemory}}\label{Appendix_Proof_Memory}

As before, let $P=[P(x^n,y^n)]_{x^n,y^n\in\{0,1\}^n}$ denote the joint probability matrix of $X^n$ and $Y^n$ and let $q(y^n)=\Pr(Y^n=y^n)$ for $y^n\in\{0,1\}^n$. We first show that $(X^n,Y^n)$ satisfies the hypotheses of Theorem~\ref{Thm:GeneralizedLocalLinearity}, and thus we can use \eqref{eq:DerivativePsi} to obtain bounds on $\underline{\mathcalboondox{h}}'(\cP(X^n|Y^n))$. ((Note that $\cP(X^n)<\cP(X^n|Y^n)$  by the assumption.))

Assumptions ($\textnormal{a}_2$) and ($\textnormal{b}$) imply that, for all $x^n, y^n\in \{0, 1\}^n$
\begin{equation}\label{Jont_Dis_Thm4}
P(x^n,y^n) = (\bar{\alpha}\bar{r})^n \frac{\bar{p}}{\bar{r}} \left(\frac{p}{\bar{p}}\right)^{x_1} \left(\frac{\alpha}{\bar{\alpha}}\right)^{x_1 \oplus y_1} \Upsilon(x^n, y^n),
\end{equation}
where   $\displaystyle \Upsilon(x^n, y^n)=\prod_{k=2}^n \left(\frac{r}{\bar{r}}\right)^{x_k\oplus x_{k-1}} \left(\frac{\alpha}{\bar{\alpha}}\right)^{x_k \oplus y_k}$ and  the product equals one if $n=1$. Since $\alpha>0$, it is clear that $q(y^n)>0$ for all $y^n\in\{0,1\}^n$. Let $N_0(z^n) = |\{1\leq k\leq n : z_k=0\}|$ and $N_1(z^n) = |\{1\leq k\leq n : z_k=1\}|$ for  any binary vector $z^n\in\{0,1\}^n$. Recall that $n$ is odd, so either $N_0(z^n) < N_1(z^n)$ or $N_0(z^n) > N_1(z^n)$. The following lemma shows that for every $y^n\in\{0,1\}^n$ there exists (a unique) $x^n_{y^n}\in\{0,1\}^n$ such that $P(x^n_{y^n},y^n) > P(x^n,y^n)$ for all $x^n\neq x^n_{y^n}$.

\begin{lemma}\label{Lemma:JointProbability_MaximizationMemory}
Let $(X^n,Y^n)$ be as in the hypothesis of Theorem~\ref{Theorem_MarkovMemory}. Then, we have  for any $y^n\in\{0, 1\}^n$
\eq{P(x^n, y^n) \leq \begin{cases}(\bar{\alpha}\bar{r})^n \frac{\bar{p}}{\bar{r}} \left(\frac{\alpha}{\bar{\alpha}}\right)^{N_1(y^n)}, & \text{if}~ N_0(y^n)>N_1(y^n),\\
(\bar{\alpha}\bar{r})^n \frac{p}{\bar{r}} \left(\frac{\alpha}{\bar{\alpha}}\right)^{N_0(y^n)}, & \text{if}~ N_0(y^n) < N_1(y^n),\end{cases}}
for all $x^n\in\{0,1\}^n$ with equality if and only if $x^n=\bold{0}$ or $x^n=\bold{1}$, respectively.
\end{lemma}

To prove this lemma, we will make use of the following fact.

\noindent{\bf Claim.} Let $y^n\in\{0, 1\}^n$ be given. If $x^n\in\{0, 1\}^n$ maximizes $P(x^n, y^n)$, then $x_1=x_2=\dots=x_n$.
\begin{proof}[Proof of Claim]
We prove the result using backward induction. To do so, we assume that the maximizer $x^n$ satisfies $x_n=x_{n-1}=\cdots=x_{l}$ for $2\leq l\leq n$. It is sufficient to show that $x_n=\cdots=x_l=x_{l-1}$. In light of \eqref{Jont_Dis_Thm4}, we have
\eqn{eq:Induction}{P(x^n, y^n) = A_{l-1} \left(\frac{r}{\bar{r}}\right)^{x_l\oplus x_{l-1}} \prod_{k=l}^n \left(\frac{\alpha}{\bar{\alpha}}\right)^{x_l \oplus y_k},}
where\footnote{When $l\leq 3$, we use the convention that $\prod_{k=2}^{l-1} \left(\frac{r}{\bar{r}}\right)^{x_k\oplus x_{k-1}} \left(\frac{\alpha}{\bar{\alpha}}\right)^{x_k \oplus y_k} = 1$.}
\eq{A_{l-1} = (\bar{\alpha}\bar{r})^n \frac{\bar{p}}{\bar{r}} \left(\frac{p}{\bar{p}}\right)^{x_1} \left(\frac{\alpha}{\bar{\alpha}}\right)^{x_1 \oplus y_1} \Upsilon(x^{\ell-1}, y^{\ell-1}).}
Notice that $A_{l-1}$ depends only on $\repdc{x}{1}{l-1}$. By the induction hypothesis, we have $x_l=\cdots=x_n$. In particular, $x^n$ equals either
\eq{\tilde{x}^n\coloneqq\{\repdc{x}{1}{l-1},\underbrace{\bar{x}_{l-1},\ldots,\bar{x}_{l-1}}_{n-l+1}\},}
or
\eq{\hat{x}^n\coloneqq\{\repdc{x}{1}{l-1},\underbrace{x_{l-1},\ldots,x_{l-1}}_{n-l+1}\}.}
By \eqref{eq:Induction}, we have that
\eq{
P(\tilde{x}^n, y^n) = A_{l-1} \frac{r}{\bar{r}} \prod_{k=l}^n \left(\frac{\alpha}{\bar{\alpha}}\right)^{1-x_{l-1} \oplus y_k},}
and
\eq{P(\hat{x}^n, y^n) = A_{l-1} \prod_{k=l}^n \left(\frac{\alpha}{\bar{\alpha}}\right)^{x_{l-1} \oplus y_k}.
}
By the assumptions on $r$ and $\alpha$, we have
\al{\frac{r}{\bar{r}} \prod_{k=l}^n \left(\frac{\alpha}{\bar{\alpha}}\right)^{1-x_{l-1} \oplus y_k} &\leq \frac{r}{\bar{r}} < \left(\frac{\alpha}{\bar{\alpha}}\right)^{n-1} \\
	&\leq \left(\frac{\alpha}{\bar{\alpha}}\right)^{n-l+1} \leq \prod_{k=l}^n \left(\frac{\alpha}{\bar{\alpha}}\right)^{x_{l-1} \oplus y_k},}
which shows that $P(\tilde{x}^n, y^n)<P(\hat{x}^n, y^n)$ and hence $x^n=\hat{x}^n$. In other words, $x_{l-1}=x_l=\cdots=x_n$. This completes the induction step.
\end{proof}

\begin{proof}[Proof of Lemma~\ref{Lemma:JointProbability_MaximizationMemory}]
By the above claim, for any given $y^n\in\{0,1\}^n$, the maximizer $x^n\in\{0,1\}^n$  of $P(x^n,y^n)$ is either $x^n={\bf 0}$ or $x^n={\bf 1}$, for which we have
\aln{
\label{eq:Prob0} P({\bf 0},y^n) &= (\bar{\alpha}\bar{r})^n \frac{\bar{p}}{\bar{r}} \left(\frac{\alpha}{\bar{\alpha}}\right)^{N_1(y^n)},\\
\label{eq:Prob1} P({\bf 1},y^n) &= (\bar{\alpha}\bar{r})^n \frac{p}{\bar{r}} \left(\frac{\alpha}{\bar{\alpha}}\right)^{N_0(y^n)}.
}
Assume $N_0(y^n)>N_1(y^n)$ and recall that $\alpha p<\bar{\alpha} \bar{p}$. In this case,
\eq{p \left(\frac{\alpha}{\bar{\alpha}}\right)^{N_0(y^n)} \leq \frac{\alpha p}{\bar{\alpha}} \left(\frac{\alpha}{\bar{\alpha}}\right)^{N_1(y^n)} < \bar{p} \left(\frac{\alpha}{\bar{\alpha}}\right)^{N_1(y^n)},}
which implies $P({\bf 0}, y^n)>P({\bf 1}, y^n)$, and hence $x^n=\bold{0}$ is the only maximizer. If $N_0(y^n) < N_1(y^n)$, then $\left(\frac{\alpha}{\bar{\alpha}}\right)^{N_0(y^n)} > \left(\frac{\alpha}{\bar{\alpha}}\right)^{N_1(y^n)}$. Since $p\geq\bar{p}$, we conclude that
\eq{p \left(\frac{\alpha}{\bar{\alpha}}\right)^{N_0(y^n)} > \bar{p} \left(\frac{\alpha}{\bar{\alpha}}\right)^{N_1(y^n)}.}
Consequently, $P({\bf 1}, y^n)>P({\bf 0}, y^n)$ and hence $x^n=\bold{1}$ is the only maximizer.
\end{proof}

Note that
\begin{eqnarray}
  \cP(X^n|Y^n) &=& \sum_{y^n\in\{0,1\}^n} \max_{x^n\in\{0,1\}^n} P(x^n,y^n)\nonumber\\
  &\stackrel{(a)}{=}& \sum_{y^n:N_0(y^n)>N_1(y^n)} P({\bf 0},y^n)\nonumber\\
  &&+\sum_{y^n:N_0(y^n) < N_1(y^n)} P({\bf 1},y^n)\nonumber\\
   &\stackrel{(b)}{=}& \bar{\alpha}^n \bar{r}^{n-1} \sum_{k=0}^{(n-1)/2} \binom{n}{k} \left(\frac{\alpha}{\bar{\alpha}}\right)^k,
   \label{End_Point_Epsilon_Memory}
\end{eqnarray}
where $(a)$ is due to Lemma~\ref{Lemma:JointProbability_MaximizationMemory} and $(b)$ comes from \eqref{eq:Prob0} and \eqref{eq:Prob1}.

In order to be able to use Theorem~\ref{Thm:GeneralizedLocalLinearity}, we first need to show that $\cP(X^n)<\cP(X^n|Y^n)$. Note that $1=\sum_{k=0}^n{n \choose k}\alpha^n\bar{\alpha}^{n-k}$ and hence $\bar{\alpha}^n\sum_{k=0}^n{n \choose k}\left(\frac{\alpha}{\bar{\alpha}}\right)^k=1$. We can therefore write
\begin{eqnarray}
	\frac{1}{\bar{\alpha}^n}&=&\sum_{k=0}^n{n \choose k}\left(\frac{\alpha}{\bar{\alpha}}\right)^k\nonumber\\
	&=&\sum_{k=0}^{(n-1)/2}{n \choose k}\left(\frac{\alpha}{\bar{\alpha}}\right)^k\left(1+\left(\frac{\alpha}{\bar{\alpha}}\right)^{n-2k}\right)\nonumber\\
	&\leq&\sum_{k=0}^{(n-1)/2}{n \choose k}\left(\frac{\alpha}{\bar{\alpha}}\right)^k\left(1+\frac{\alpha}{\bar{\alpha}}\right)\nonumber\\
	&<&\sum_{k=0}^{(n-1)/2}{n \choose k}\left(\frac{\alpha}{\bar{\alpha}}\right)^k\left(1+\frac{\bar{p}}{p}\right)\nonumber\\
	&=&\frac{1}{p}\sum_{k=0}^{(n-1)/2}{n \choose k}\left(\frac{\alpha}{\bar{\alpha}}\right)^k, \label{Eq:PcStrict}
\end{eqnarray}
which implies that $\cP(X^n)<\cP(X^n|Y^n)$.

Now that all the hypotheses of Theorem~\ref{Thm:GeneralizedLocalLinearity} are shown to be satisfied, we can use \eqref{eq:DerivativePsi} to study $\underline{\mathcalboondox{h}}'(\cP(X^n|Y^n))$. The following lemma is important in bounding $\underline{\mathcalboondox{h}}'(\cP(X^n|Y^n))$. 

\begin{lemma}\label{Lemma:bound_q_y}
Let $(X^n,Y^n)$ be as in the hypothesis of Theorem~\ref{Theorem_MarkovMemory}. Then, for all $y^n\in\{0,1\}^n$,
\eq{q(y^n) \geq \alpha^n.}
\end{lemma}

\begin{proof}
From \eqref{Jont_Dis_Thm4}, we have
\begin{eqnarray}
P(x^n,y^n) &=& (\bar{\alpha}\bar{r})^n \frac{\bar{p}}{\bar{r}} \left(\frac{p}{\bar{p}}\right)^{x_1} \left(\frac{\alpha}{\bar{\alpha}}\right)^{x_1 \oplus y_1} \Upsilon_n(x^n, y^n)\nonumber\\
&\geq & \left(\frac{\alpha}{\bar{\alpha}}\right)^n (\bar{\alpha}\bar{r})^n \frac{\bar{p}}{\bar{r}} \left(\frac{p}{\bar{p}}\right)^{x_1} \prod_{k=2}^n \left(\frac{r}{\bar{r}}\right)^{x_k\oplus x_{k-1}}\nonumber\\
&=& \alpha^n \bar{r}^n \frac{\bar{p}}{\bar{r}} \left(\frac{p}{\bar{p}}\right)^{x_1} \prod_{k=2}^n \left(\frac{r}{\bar{r}}\right)^{x_k\oplus x_{k-1}}\nonumber.
\end{eqnarray}
Summing over all $x^n\in\{0,1\}^n$, we obtain
\begin{equation}\label{Proof_memory}
q(y^n)\geq \alpha^n \bar{r}^{n-1} \bar{p} \sum_{x^n\in\{0,1\}^n}\left(\frac{p}{\bar{p}}\right)^{x_1} \prod_{k=2}^n \left(\frac{r}{\bar{r}}\right)^{x_k\oplus x_{k-1}.}
\end{equation}
On the other hand, it is straightforward to verify that
\begin{equation}
\label{Proof_memory2}
\begin{aligned}
1&=\sum_{x\in\{0,1\}^n} \Pr(X^n=x^n)\\
&= \bar{r}^{n-1}\bar{p} \sum_{x^n\in\{0,1\}^n}\left(\frac{p}{\bar{p}}\right)^{x_1} \prod_{k=2}^n\left(\frac{r}{\bar{r}}\right)^{x_k\oplus x_{k-1}}.
\end{aligned}
\end{equation}
Plugging \eqref{Proof_memory2} into \eqref{Proof_memory}, the result follows.
\end{proof}

By \eqref{eq:DerivativePsi} and the previous lemma,
\eq{\underline{\mathcalboondox{h}}'(\cP(X^n|Y^n)) \geq \min_{y^n,z^n\in\{0,1\}^n} \frac{\alpha^n}{P(x^n_{y^n},y^n) - P(x^n_{z^n},y^n)}.}
Since both $x^n_{y^n}$ and $x^n_{z^n}$ are either $\bold{0}$ or $\bold{1}$, we have to maximize
\eq{\vartheta\coloneqq
\begin{cases}(\bar{\alpha}\bar{r})^n \frac{\bar{p}}{\bar{r}} \left(\frac{\alpha}{\bar{\alpha}}\right)^{N_1(y^n)} - (\bar{\alpha}\bar{r})^n \frac{p}{\bar{r}} \left(\frac{\alpha}{\bar{\alpha}}\right)^{N_0(y^n)},& \text{if}~ y^n\in \mathcal{R}_0,\\(\bar{\alpha}\bar{r})^n \frac{p}{\bar{r}} \left(\frac{\alpha}{\bar{\alpha}}\right)^{N_0(y^n)} - (\bar{\alpha}\bar{r})^n \frac{\bar{p}}{\bar{r}} \left(\frac{\alpha}{\bar{\alpha}}\right)^{N_1(y^n)}, & \text{if}~ y^n\notin \mathcal{R}_0,\end{cases}
}
where $\mathcal{R}_0=\{y^n\in\{0,1\}^n:N_0(y^n)>N_1(y^n)\}$. Clearly, $\vartheta$ is maximized when $y^n={\bf 1}$ and thus
\eq{\underline{\mathcalboondox{h}}'(\cP(X^n|Y^n)) \geq \frac{\bar{r}\alpha^n}{p(\bar{\alpha}\bar{r})^n-\bar{p}(\alpha\bar{r})^n}.}
By \eqref{eq:PsiExtremePoint} and the fact that $\underline{\mathcalboondox{h}}_n^n(\eps)=\underline{\mathcalboondox{h}}(\eps^n)$,
\eq{\underline{\mathcalboondox{h}}_n^n(\eps) \leq 1 - \bar{r}\frac{\cP(X^n|Y^n)-\eps^n}{p(\bar{\alpha}\bar{r})^n-\bar{p}(\alpha\bar{r})^n} \alpha^n,} where $\cP(X^n|Y^n)$ is computed in \eqref{End_Point_Epsilon_Memory}.

The lower bound follows from considering the direction $\tilde{D}\in\ul{\D}(\rm{I}_{2^n})$, whose entries are all zero except $\tilde{D}(\bold{1}, \bold{0})=\lambda$ and $\tilde{D}(\bold{1}, \bold{1})=-\lambda$ for $\lambda=2^{-1/2}$. In particular, plugging $\tilde{D}$ into \eqref{eq:MasterEquationQuotient}, we obtain an upper bound for $\underline{\mathcalboondox{h}}'(\cP(X^n|Y^n))$ and thus a lower bound for $\underline{\mathcalboondox{h}}(\eps)$ for the desired range of $\eps$. Note that the filter $\rm{I_{2^n}}+\zeta_n(\eps)\tilde{D}$ corresponds to the $2^n$-ary Z-channel $\mathsf{Z}_n(\zeta_n(\eps))$.


\section{Proof of Proposition~\ref{Proposition_Parametric_Dis_Privacy}}\label{Appendix:Parametric}
Since $r=0$, the joint distribution $P_{\theta Y^n}$ can be equivalently written as the joint probability matrix $P=[P(x^n,y^n)]_{x^n, y^n\in \{0,1\}^n}$ with $x_1=x_2=\dots=x_n=\theta$.  As in the proof of Theorem~\ref{Theorem_MarkovMemory}, the hypotheses of Theorem~\ref{Thm:GeneralizedLocalLinearity} are fulfilled. In particular,
\eqn{eq:Jordan0}{\underline{\mathcalboondox{h}}'(\cP(\theta|Y^n)) = \min_{y^n,z^n\in\{0,1\}^n} \frac{q(y^n)}{P(x^n_{y^n},y^n) - P(x^n_{z^n},y^n)}.}
In this case, \eqref{Jont_Dis_Thm4} becomes
\eq{
P({\bf 0},y^n) = \bar{p} \bar{\alpha}^n \left(\frac{\alpha}{\bar{\alpha}}\right)^{N_1(y^n)}, }
and 
\eq{P({\bf 1},y^n) = p \bar{\alpha}^n \left(\frac{\alpha}{\bar{\alpha}}\right)^{N_0(y^n)}.}
In particular,
\eq{\underline{\mathcalboondox{h}}'(\cP(\theta|Y^n)) =  \min_{y^n, z^n\in\{0,1\}^n} \frac{p \bar{\alpha}^n \left(\frac{\alpha}{\bar{\alpha}}\right)^{N_0(y^n)}+\bar{p} \bar{\alpha}^n \left(\frac{\alpha}{\bar{\alpha}}\right)^{N_1(y^n)}}{P(x^n_{y^n},y^n) - P(x^n_{z^n},y^n)}.}
Lemma \ref{Lemma:JointProbability_MaximizationMemory} implies that both $x^n_{y^n}$ and $x^n_{z^n}$ are either ${\bf 0}$ or ${\bf 1}$. If $N_0(y^n)>N_1(y^n)$, then
\al{&\frac{p \bar{\alpha}^n \left(\frac{\alpha}{\bar{\alpha}}\right)^{N_0(y^n)}+\bar{p} \bar{\alpha}^n \left(\frac{\alpha}{\bar{\alpha}}\right)^{N_1(y^n)}}{P(x^n_{y^n},y^n) - P(x^n_{z^n},y^n)} \\
	&~~~~~~~~~~~~~~~~~~~~~\geq \frac{p \bar{\alpha}^n \left(\frac{\alpha}{\bar{\alpha}}\right)^{N_0(y^n)}+\bar{p} \bar{\alpha}^n \left(\frac{\alpha}{\bar{\alpha}}\right)^{N_1(y^n)}}{\bar{p} \bar{\alpha}^n \left(\frac{\alpha}{\bar{\alpha}}\right)^{N_1(y^n)} - p \bar{\alpha}^n \left(\frac{\alpha}{\bar{\alpha}}\right)^{N_0(y^n)}},}
with equality if and only if $N_1(z^n)>N_0(z^n)$. It is not hard to show that
\eqn{eq:Jordan1}{\frac{p \bar{\alpha}^n \left(\frac{\alpha}{\bar{\alpha}}\right)^{N_0(y^n)}+\bar{p} \bar{\alpha}^n \left(\frac{\alpha}{\bar{\alpha}}\right)^{N_1(y^n)}}{\bar{p} \bar{\alpha}^n \left(\frac{\alpha}{\bar{\alpha}}\right)^{N_1(y^n)} - p \bar{\alpha}^n \left(\frac{\alpha}{\bar{\alpha}}\right)^{N_0(y^n)}} \geq \frac{\bar{p}+p\left(\frac{\alpha}{\bar{\alpha}}\right)^n}{\bar{p} - p \left(\frac{\alpha}{\bar{\alpha}}\right)^n},}
with equality if and only if $y^n={\bf 0}$. Similarly, if $N_1(y^n)>N_0(y^n)$, then
\al{&\frac{p \bar{\alpha}^n \left(\frac{\alpha}{\bar{\alpha}}\right)^{N_0(y^n)}+\bar{p} \bar{\alpha}^n \left(\frac{\alpha}{\bar{\alpha}}\right)^{N_1(y^n)}}{P(x^n_{y^n},y^n) - P(x^n_{z^n},y^n)} \\
	&~~~~~~~~~~~~~~~~~~~~~~\geq \frac{p \bar{\alpha}^n \left(\frac{\alpha}{\bar{\alpha}}\right)^{N_0(y^n)}+\bar{p} \bar{\alpha}^n \left(\frac{\alpha}{\bar{\alpha}}\right)^{N_1(y^n)}}{p \bar{\alpha}^n \left(\frac{\alpha}{\bar{\alpha}}\right)^{N_0(y^n)} - \bar{p} \bar{\alpha}^n \left(\frac{\alpha}{\bar{\alpha}}\right)^{N_1(y^n)}},}
with equality if and only if $N_0(z^n)>N_1(z^n)$. As before,
\eqn{eq:Jordan2}{\frac{p \bar{\alpha}^n \left(\frac{\alpha}{\bar{\alpha}}\right)^{N_0(y^n)}+\bar{p} \bar{\alpha}^n \left(\frac{\alpha}{\bar{\alpha}}\right)^{N_1(y^n)}}{\bar{p} \bar{\alpha}^n \left(\frac{\alpha}{\bar{\alpha}}\right)^{N_1(y^n)} - p \bar{\alpha}^n \left(\frac{\alpha}{\bar{\alpha}}\right)^{N_0(y^n)}} \geq \frac{p+\bar{p}\left(\frac{\alpha}{\bar{\alpha}}\right)^n}{p - \bar{p} \left(\frac{\alpha}{\bar{\alpha}}\right)^n},}
with equality if and only if $y^n={\bf 1}$. From \eqref{eq:Jordan1} and \eqref{eq:Jordan2}, we conclude that
\eq{\underline{\mathcalboondox{h}}'(\cP(\theta|Y^n)) = \frac{p+\bar{p}\left(\frac{\alpha}{\bar{\alpha}}\right)^n}{p - \bar{p} \left(\frac{\alpha}{\bar{\alpha}}\right)^n}=\frac{p\bar{\alpha}^n+\bar{p}\alpha^n}{p\bar{\alpha}^n-\bar{p}\alpha^n},}
and $y_0={\bf 1}$ and $z_0={\bf 0}$ achieve the minimum in \eqref{eq:Jordan0}. From the last part of Theorem~\ref{Thm:GeneralizedLocalLinearity} the optimality of the $2^n$-ary Z-channel $\mathsf{Z}_n(\zeta_n(\eps))$ is evident.

\section{Proof of Theorem~\ref{Theorem_equivakebt_strong_MC}}\label{Appendidx:MC}
	From \eqref{Eq:Eta_MMSE} and \eqref{Maximal_correlation_Equivalent} we obtain that
	\eq{\inf_{f\in\S_U} \frac{\mmse(f(U)|V)}{\var(f(U))} = 1-\sup_{f\in\S_U}\eta^2_V(f(U)) = 1-\rho_m^2(U,V).}
	From the previous equation it is clear that $\rho_m^2(U, V)\leq \eps$ if and only if
	\begin{equation*}
	\mmse(f(U)|V)\geq (1-\eps)\var(f(U)),
	\end{equation*}
	for all $f\in\S_U$. By \eqref{Def:Strong_Privacy}, we obtain $Z_\gamma\in\Gamma(\eps)$ if and only if $\rho_m^2(X,Z_\gamma)\leq \eps$.

\section{Proof of Theorem~\ref{Lemma_Gaussian_Y_Arbit_X}}\label{Appendix_Bound}

	Without loss of generality, assume $\E(X)=\E(Y_\sG)=0$. Since $Y_\sG$ is Gaussian,  \eqref{Eq:Eta_MMSE} implies that
	\begin{eqnarray}
	\sM(\eps) &=& \inf_{\gamma:\rho_m^2(X,Z_\gamma)\leq\eps} \frac{\mmse(Y_{\sG}|Z_\gamma)}{\var(Y_{\sG})}\nonumber\\
	&=& 1 - \sup_{\gamma:\rho_m^2(X,Z_\gamma)\leq\eps} \rho_m^2(Y_\sG;Z_\gamma).\label{eq:MasterBounds}
	\end{eqnarray}
	A straightforward computation leads to
	\aln{
		\label{eq:MaximalCorrelationGG}\rho_m^2(Y_\sG,Z_\gamma) &= \rho^2(Y_\sG,Z_\gamma) = \frac{\gamma\var(Y_\sG)}{1+\gamma\var(Y_\sG)},\\
		\nonumber \rho_m^2(X,Z_\gamma) &\geq \rho^2(X,Z_\gamma) =  \rho^2(X,Y_\sG) \rho_m^2(Y_\sG,Z_\gamma).
	}
	The preceding inequality and  \eqref{eq:MasterBounds} imply
	\eq{\sM(\eps) \geq 1 - \sup_{\gamma:\rho_m^2(X,Z_\gamma)\leq\eps} \frac{\rho_m^2(X,Z_\gamma)}{\rho^2(X,Y_\sG)} \geq 1 - \frac{\eps}{\rho^2(X,Y_\sG)},}
	which proves the lower bound.
	
	The strong data processing inequality for maximal correlation \cite[Lemma 6]{Asoode_submitted} states that $\rho_m^2(X, Z_{\gamma})\leq \rho_m^2(X,Y_\sG)\rho_m^2(Y_\sG, Z_{\gamma})$. In particular, if \dsty{\rho_m^2(Y_\sG, Z_{\gamma}) \leq \frac{\eps}{\rho_m^2(X,Y)}}, then $\rho_m^2(X,Z_\gamma) \leq \eps$. Therefore, \eqref{eq:MasterBounds} implies
	\begin{eqnarray*}
		\sM(\eps) &\leq& 1 - \sup_{\gamma:\rho_m^2(Y_\sG, Z_{\gamma}) \leq \frac{\eps}{\rho_m^2(X,Y_\sG)}} \rho_m^2(Y_\sG;Z_\gamma)\\
		& =& 1 - \frac{\eps}{\rho_m^2(X,Y_\sG)},
	\end{eqnarray*}
	
	where the last equality follows from the continuity of $\gamma \mapsto \rho_m^2(Y_\sG, Z_{\gamma})$, established in \eqref{eq:MaximalCorrelationGG}, finishing the proof of the upper bound.

\section{Proof of Lemma~\ref{Lemma_Approx_S}}\label{Appendix:limsup}
	Let 
	\begin{equation}\label{Def_Gamma_eps}
	\gamma^*_{\eps}\coloneqq \max\{\gamma\geq 0:\rho_m^2(X_\sG, Z_\gamma)\leq \eps\}.
	\end{equation}
	Recall that 
	\begin{equation}\label{Correlation_Gamma_Eps}
	\rho^2_m(X, Z_\gamma)\geq \rho^2(X, Z_\gamma)=\frac{\gamma \rho^2(X,Y)\var(Y)}{1+\gamma\var(Y)}.
	\end{equation}
	Since $\eps\to 0$, we can assume that $\eps< \rho^2(X,Y)$. Thus,  from \eqref{Correlation_Gamma_Eps} we obtain 
	\begin{equation}\label{Correlation_Gamma_eps2}
	\gamma^*_{\eps}\leq \frac{\eps}{\var(Y)(\rho^2(X, Y)-\eps)}. 
	\end{equation} 
	In particular, $\gamma_{\eps}^*\to 0$ as $\eps\to 0$.
	Since $\gamma\mapsto\mmse(Y|Z_\gamma)$ is decreasing, we have that $\sM(\eps)=\mmse(Y|Z_{\gamma^*_\eps})$. 
	Therefore, the first-order approximation of $\sM(\cdot)$ around zero yields
	\begin{eqnarray*}
		\sM(\eps)&=&1+\frac{\gamma^*_{\eps}}{\var(Y)}\frac{\text{d}}{\text{d}\gamma^*_{\eps}}\mmse(Y|Z_{\gamma^*_{\eps}})\Big|_{\eps=0}+o(\gamma^*_{\eps})\\
		&\stackrel{(a)}{=}& 1-\var(Y)\gamma^*_{\eps}+o(\gamma^*_{\eps})\\
		&\stackrel{(b)}{\geq}& 1-\frac{\eps}{\rho^2(X, Y)}+o(\eps)
	\end{eqnarray*} 
	where $(a)$ follows from the fact that $\frac{\text{d}}{\text{d}\gamma}\mmse(Y|Z_{\gamma})=-\E[\var^2(Y|Z_\gamma)]$ \cite[Prop. 9]{MMSE_Guo} and $(b)$ follows from \eqref{Correlation_Gamma_eps2}.

\section*{Acknowledgment}
The authors would like to acknowledge two anonymous reviewers for their insightful comments and, in particular, one of them for the derivation in \eqref{Eq:PcStrict}. Furthermore, the first author acknowledges useful discussions with M. M\'edard and F. P. Calmon.

\bibliographystyle{IEEEtran}
\bibliography{bibliography}

\end{document}

%% file: mypreamble.tex
\usepackage[english]{babel}
\usepackage{enumitem}
\usepackage{enumerate}
\usepackage{color}
\usepackage[T1]{fontenc}
\usepackage{subfigure,caption}
\usepackage{dsfont}
\usepackage[final]{graphicx}
\usepackage[T1]{fontenc}
\usepackage{amsmath}
\usepackage{mathtools, cuted}
\usepackage{amsthm}
\usepackage{amstext}
\usepackage{amssymb}
\usepackage{mathrsfs}
\usepackage{cite}
\usepackage{mathtools}
\usepackage{tikz}
\usepackage{tkz-tab}
\usepackage{pgfplots}

\usepgfplotslibrary{fillbetween}
\usepackage{marginnote}
\usetikzlibrary{babel}
\newcommand{\repn}[1]{1,\ldots,#1}
\DeclareFontFamily{U}{BOONDOX-calo}{\skewchar\font=45 }
\DeclareFontShape{U}{BOONDOX-calo}{m}{n}{
  <-> s*[1.05] BOONDOX-r-calo}{}
\DeclareFontShape{U}{BOONDOX-calo}{b}{n}{
  <-> s*[1.05] BOONDOX-b-calo}{}
\DeclareMathAlphabet{\mathcalboondox}{U}{BOONDOX-calo}{m}{n}
\SetMathAlphabet{\mathcalboondox}{bold}{U}{BOONDOX-calo}{b}{n}
\DeclareMathAlphabet{\mathbcalboondox}{U}{BOONDOX-calo}{b}{n}

\usetikzlibrary{arrows,positioning, shapes}
\usepackage{float}
\usepackage{caption}

\definecolor{OliveGreen}{rgb}{0,0.4,0}
\definecolor{wine-stain}{rgb}{0.5,0,0}
\usepackage{hyperref}
\hypersetup{colorlinks,
    colorlinks,
    citecolor=red,
    filecolor=green,
    linkcolor=blue,
    linktoc=page,
    urlcolor=blue
}

\usepackage{hyperref}
\def\R{\mathbb{R}}

\def\eps{\varepsilon}

\def\A{{\mathcal A}}
\def\B{{\mathcal B}}
\def\D{{\mathcal D}}
\def\E{{\mathbb E}}
\def\F{{\mathcal F}}
\def\H{{\mathcal H}}

\def\C{{\mathcal C}}
\def\P{{\mathcal P}}

\def\X{{\mathcal X}}
\def\Y{{\mathcal Y}}
\def\M{{\mathcal M}}
\def\N{{\mathcal N}}
\def\S{{\mathcal S}}
\def\L{{\mathcal L}}

\def\Z{{\mathcal Z}}
\def\U{{\mathcal U}}

\def\indep{{\perp\!\!\!\perp}}

\def\sBer{{\mathsf{Bernoulli}}}

\def\sM{{\mathsf {sENSR}}}

\def\sW{{\mathsf {wENSR}}}
\def\sG{{\mathsf G}}
\def \var {{\mathsf{var}}}
\def \cov {{\mathsf{cov}}}
\def \mmse {{\mathsf {mmse}   }}
\newcommand{\cP}{\mathsf{P}_{\mathsf{c}}}

\newcommand{\dsty}[1]{$\displaystyle #1$}
\newcommand{\ndsty}[1]{$#1$}
\def \sM {{\mathsf{sENSR}}}
\def \sW {{\mathsf{wENSR}}}

\newcommand{\eq}[1]{\begin{equation*}
#1
\end{equation*}}
\newcommand{\eqn}[2]{\begin{equation}
\label{#1}
#2
\end{equation}}
\newcommand{\al}[1]{\begin{align*}
#1
\end{align*}}
\newcommand{\aln}[1]{\begin{align}
#1
\end{align}}
\newcommand{\ul}[1]{\underline{#1}}

\newcommand{\repdc}[3]{#1_{#2} , \ldots , #1_{#3}}
\usepackage[font=small,labelsep=space]{caption}
\DeclareCaptionLabelSeparator{dot}{.~}
\newcounter{example}
\newenvironment{example}[1][]{\refstepcounter{example}\par\medskip
   \noindent \textit{Example~\theexample. #1} \rmfamily}{\medskip}
\newtheorem{definition}{Definition}
\newtheorem{theorem}{Theorem}
\newtheorem{corollary}{Corollary}
\newtheorem*{corollary-non}{Corollary}

\newtheorem{proposition}{Proposition}
\newtheorem{lemma}{Lemma}
\theoremstyle{remark}

\newcommand{\markov}{\mathrel\multimap\joinrel\mathrel-%
\mspace{-9mu}\joinrel\mathrel-}

\usepackage{times}
\usepackage{tikz}
\usepackage{amsmath}
\usepackage{verbatim}
\usetikzlibrary{arrows,shapes}
\tikzstyle{RectObject}=[rectangle,fill=white,draw,line width=0.2mm]
\tikzstyle{line}=[draw]
\tikzstyle{arrow}=[draw, -latex]
\usetikzlibrary{decorations.pathmorphing}
\usetikzlibrary{calc,shapes, positioning}
\usepackage{graphicx}
\usepackage{caption}
\usetikzlibrary{shapes.geometric}

\allowdisplaybreaks